\documentclass[sigplan,authorversion]{acmart}

\usepackage[utf8]{inputenc}
\usepackage[T1]{fontenc}
\usepackage{xspace}

\usepackage[capitalize,noabbrev]{cleveref}
\usepackage{stmaryrd}
\usepackage{mathtools}
\usepackage{bbold}
\usepackage{enumitem}

\usepackage{thm-restate}

\usepackage{subcaption}

\usepackage[textwidth=1.9cm,textsize=small]{todonotes}

\AtEndPreamble{
\theoremstyle{acmplain}
\newtheorem{claim}[theorem]{Claim}
\newenvironment{claimproof}[1]{\par\noindent \textit{Proof of claim:}\space#1}{\hfill $\blacksquare$}

\theoremstyle{acmdefinition}
\newtheorem{remark}[theorem]{Remark}
\newtheorem*{remark*}{Remark}
}

\newcommand*{\eg}{e.g.\@\xspace}
\newcommand*{\ie}{i.e.\@\xspace}

\newcommand{\bu}{\mathbf{u}}
\newcommand{\bv}{\mathbf{v}}
\newcommand{\bw}{\mathbf{w}}
\newcommand{\bx}{\mathbf{x}}
\newcommand{\by}{\mathbf{y}}
\newcommand{\ba}{\mathbf{a}}
\newcommand{\bb}{\mathbf{b}}
\newcommand{\bc}{\mathbf{c}}
\newcommand{\bs}{\mathbf{s}}

\newcommand{\Aa}{\mathcal{A}}
\newcommand{\Bb}{\mathcal{B}}
\newcommand{\Cc}{\mathcal{C}}

\newcommand{\X}{\mathcal{X}}
\newcommand{\V}{\mathcal{V}}
\newcommand{\Mm}{\mathcal{M}}

\newcommand{\bbN}{\mathbb{N}}
\newcommand{\N}{\bbN}
\newcommand{\Npos}{\N \setminus \set{0}}
\newcommand{\bbZ}{\mathbb{Z}}
\newcommand{\Z}{\bbZ}
\newcommand{\bbQ}{\mathbb{Q}}
\newcommand{\Q}{\bbQ}
\newcommand{\Qpos}{\Q_{\ge 0}}
\newcommand{\Qposs}{\Q_{> 0}}

\newcommand{\bbS}{\mathbb{S}}
\renewcommand{\S}{\bbS}
\newcommand{\bbR}{\mathbb{R}}
\newcommand{\R}{\bbR}
\newcommand{\Rpos}{\R_{>0}}

\newcommand{\affine}{\text{affine}}
\newcommand{\LQ}{\mathbb{LogQ}}

\renewcommand{\leq}{\leqslant}
\renewcommand{\geq}{\geqslant}
\renewcommand{\le}{\leqslant}
\renewcommand{\ge}{\geqslant}

\newcommand{\semiringS}{\S(\oplus,\odot)}
\newcommand{\semiringQ}{\texorpdfstring{\ensuremath{\Qpos(+,\cdot)}}{Q\_+(+,·)}}
\newcommand{\semiringQmax}{\Qpos(\max,\cdot)}
\newcommand{\semiringZ}{\Z(\min,+)}
\newcommand{\semiringZR}{\LQ(\min,+)}
\newcommand{\zero}{\mathbb{0}}
\newcommand{\one}{\mathbb{1}}

\newcommand{\set}[1]{\left\{ #1 \right\}}
\newcommand{\logbr}[1]{\log\left(#1\right)}

\newcommand{\sem}[1]{\left\llbracket #1 \right\rrbracket}
\newcommand{\transpose}{\mathsf{T}}

\newcommand{\simpleCRA}{Independent-CRA\xspace}
\newcommand{\simpleCRAlong}{Independent cost register automaton\xspace}
\newcommand{\ssla}{simple linearly-ambiguous\xspace}
\newcommand{\Acc}{Acc}
\newcommand{\Runs}[1][0]{Runs_{>#1}}
\newcommand{\rununique}[1]{\overrightarrow{(#1)}}

\newcommand{\val}{\operatorname{val}}
\newcommand{\valinf}{\operatorname{val}_\leftrightarrow}
\newcommand{\valpref}{\operatorname{val}_\rightarrow}

\newcommand{\poly}{\operatorname{poly}}
\newcommand{\pos}{\operatorname{pos}}
\newcommand{\corr}{\operatorname{corr}}
\newcommand{\reach}{\operatorname{reach}}
\newcommand{\1}{\mathbb{1}}
\newcommand{\runs}{\Acc}

\newcommand{\orthants}{\mathcal{O}_d}
\newcommand{\orthmax}[1]{A_{#1}^+}
\newcommand{\orthmin}[1]{A_{#1}^-}
\newcommand{\orthantpath}[1]{\mathcal{O}_{#1}}
\newcommand{\steps}{\rightarrow^*}
\newcommand{\stepsc}{\steps_c}
\newcommand{\stepsca}{\steps_{c,A}}
\newcommand{\trans}[1]{\stackrel{#1}{\longrightarrow}}
\newcommand{\eps}{\varepsilon}

\newcommand{\exptime}{\textsc{ExpTime}\xspace}

\newcommand{\norm}[1]{\left\lVert#1\right\rVert}
\newcommand{\abs}[1]{\left\lvert#1\right\rvert}

\newcommand{\incsf}{\mathsf{inc}}
\newcommand{\decsf}{\mathsf{dec}}
\newcommand{\zerosf}{\mathsf{zero}}
\newcommand{\sfstart}{\mathsf{start}}
\newcommand{\sfend}{\mathsf{end}}
 
\copyrightyear{2022}
\acmYear{2022}
\setcopyright{rightsretained}
\acmConference[LICS '22]{37th Annual ACM/IEEE Symposium on Logic in Computer Science (LICS)}{August 2--5, 2022}{Be'er Sheva, Israel} 
\acmBooktitle{37th Annual ACM/IEEE Symposium on Logic in Computer Science (LICS) (LICS '22), August 2--5, 2022, Be'er Sheva, Israel}
\acmDOI{10.1145/3531130.3533336} 
\acmISBN{978-1-4503-9351-5/22/08}

\title{The boundedness and zero isolation problems for weighted automata over nonnegative rationals}

\author{Wojciech Czerwi\'{n}ski}
\affiliation{University of Warsaw
\country{Poland}
}

\email{wczerwin@mimuw.edu.pl}

\author{Engel Lefaucheux}
\affiliation{Université de Lorraine, Inria, LORIA, Nancy
\country{France}
}

\email{engel.lefaucheux@inria.fr}

\author{Filip Mazowiecki}
\affiliation{Max Planck Institute for Software Systems, Saarland Informatics Campus
\city{Saarbrücken}
\country{Germany}
}

\email{filipm@mpi-sws.org}

\author{David Purser}
\affiliation{Max Planck Institute for Software Systems, Saarland Informatics Campus
\city{Saarbrücken}
\country{Germany}
}
\email{dpurser@mpi-sws.org}

\author{Markus A. Whiteland}
\affiliation{Max Planck Institute for Software Systems, Saarland Informatics Campus
\city{Saarbrücken}
\country{Germany}
}
\affiliation{University of Liège
\country{Belgium}
}

\email{mwhiteland@uliege.be}


\begin{document}

\begin{abstract}\sloppy
We consider linear cost-register automata (equivalent to weighted automata) over the semiring of nonnegative rationals, which generalise probabilistic automata. The two problems of boundedness and zero isolation ask whether there is a sequence of words that converge to infinity and to zero, respectively. 
In the general model both problems are undecidable so we focus on the copyless linear restriction. There, we show that the boundedness problem is decidable. 

As for the zero isolation problem we need to further restrict the class. We obtain a model, where zero isolation becomes equivalent to universal coverability of orthant vector addition systems (OVAS), a new model in the VAS family interesting on its own. In standard VAS runs are considered only in the positive orthant, while in OVAS every orthant has its own set of vectors that can be applied in that orthant. Assuming Schanuel's conjecture is true, we prove decidability of universal coverability for three-dimensional OVAS, which implies decidability of zero isolation in a model with at most three independent registers.
\end{abstract}

\maketitle

\section{Introduction}
\label{sec:introduction}

Weighted automata are a natural model of computation that generalise finite automata~\cite{droste2009handbook} and linear recursive sequences~\cite{BarloyFLM20}. They have various equivalent presentations: \eg finite automata, rational series, matrix representation~\cite{Schutzenberger61b,BerstelR88}; or recently linear cost-register automata (linear CRA)~\cite{AlurDDRY13}. A typical example is a probabilistic automaton $\Aa$ that assigns to each word $w$ its probability of acceptance, denoted $\Aa(w)$~\cite{paz71,GimbertO10,DBLP:conf/concur/FijalkowR017,DaviaudJLMP021}. More generally, weighted automata are defined with respect to a semiring: a domain with two binary operations. In the example of probabilistic automata the domain is the nonnegative rationals (thus $\Aa(w) \in \Qpos)$ with the usual operations: $+$ and $\cdot$.

Depending on the context, different semirings for weighted automata have been studied. For instance when considering learning, the semirings are usually fields, like the rationals or reals~\cite{fliess1974matrices,BeimelBBKV00}. Most results on learning weighted automata depend on Sch{\"{u}}tzenberger's polynomial time algorithm deciding the equivalence problem of weighted automata over fields~\cite{Schutzenberger61b}. 
On the other hand, when considering regular expressions, weighted automata are usually studied over the tropical semiring, \ie $\N \cup \set{+\infty}$ with the operations: $\min$ and $+$. The star height problem for regular languages can for instance 
be reduced to the boundedness problem
of such automata. Hashiguchi showed that this problem is decidable~\cite{Hashiguchi88}. Due to Hashiguchi's proof being difficult, many alternative proofs of this result appeared; among them: via Simon's factorisation trees~\cite{Simon94}; and via games~\cite{Bojanczyk15}.

This paper is primarily interested in weighted automata over the semiring of nonnegative rationals with $+$~and~$\cdot$, denoted $\semiringQ$. This is the minimal weighted automata model that captures probabilistic automata, but does not impose any restrictions on the model. Probabilistic automata assign only probabilities to words, \ie values in the interval $[0,1]$. This requires some restrictions, \eg the transitions are defined by probabilistic distributions. Similar generalisations of probabilistic automata were studied \eg in~\cite{turakainen1969generalized,DBLP:conf/concur/ChistikovKMP20}.

One of the most natural questions for such automata are the threshold problems: \ie given an automaton $\Aa$ and a constant $c$, decide whether $\Aa(w) \le c$ or whether $\Aa(w) \ge c$ for all words $w$.
We study existential variants of these problems, where only $\Aa$ is given in the input:
the \emph{zero isolation} asks whether there exists $c > 0$ such that for all words $w$ it holds $\Aa(w) \ge c$ and \emph{boundedness} asks whether there exists $c < +\infty$ such that for all words $w$ it holds $\Aa(w) \le c$.
More intuitively, the complements of the two problems ask whether there exist a sequence of words $w_1, w_2,\ldots$ such that $\lim_{i \to +\infty}\Aa(w_i)$ equals $0$ and $+\infty$, respectively.

Notice that most of the mentioned problems are well-defined already for probabilistic automata. Moreover, since probabilistic automata are known to be closed under complement (it is easy to define $\Bb(w) = 1 - \Aa(w)$) the two threshold problems are equivalent and undecidable~\cite{paz71}.
In probabilistic automata the zero isolation problem, due to complementation, is equivalent to the value-1 problem: this is also undecidable~\cite{GimbertO10}, but decidable for the special class of \emph{leaktight} probabilistic automata~\cite{DBLP:journals/corr/FijalkowGKO15}.
The boundedness problem is not interesting for probabilistic automata (since the output is always bounded by $1$), but a folklore argument shows that it is undecidable for $\semiringQ$ (see \cref{sec:results}).

Since the above problems are undecidable in general, we are interested in these problems on subclasses of weighted automata. A common restriction is bounding the ambiguity, \ie the number of accepting runs. The two most interesting classes are finitely-ambiguous and polynomially-ambiguous automata; when the number of accepting runs is bounded by: a constant (universal for all words), and by a polynomial (in the size of the input word), respectively. Both classes have nice characterisations, by excluding some simple patterns in the automata~\cite{WeberS91}. In particular, it is easy to check if an automaton is finitely-ambiguous or polynomially-ambiguous.

Both threshold problems are undecidable for polynomially-ambiguous probabilistic automata~\cite{DBLP:conf/concur/FijalkowR017,DaviaudJLMP021}. 
In the finitely-ambiguous case they are decidable~\cite{DBLP:conf/concur/FijalkowR017}, and one can infer that they remain decidable in the general setting of finitely-ambiguous weighted automata over $\semiringQ$~\cite{DaviaudJLMP021}. Unlike for probabilistic automata, the two threshold problems are different (the closure under complement is not true in general over $\semiringQ$), and while one of the inequalities is trivial to decide, the other one is known to be decidable~\cite{DaviaudJLMP021} only assuming Schanuel's conjecture~\cite{schanuelsconj}. 
Similarly, for boundedness and zero isolation, even though one could suspect they are equivalent problems, we also see a difference.  One can show that for finitely-ambiguous weighted automata over $\semiringQ$ the boundedness problem is trivially decidable; and exploiting~\cite{Chistikov21} we show that zero isolation is decidable subject to Schanuel's conjecture (see \cref{sec:results}). The argument in the latter case is more involved. The aforementioned decidability results for zero isolation on leaktight probabilistic automata do not hold over $\semiringQ$.

The decidability border between the finitely-ambiguous and polynomially-ambiguous classes is not surprising. It is often the case that undecidable problems for weighted automata are decidable for the finitely-ambiguous class~\cite{FiliotMR19}; and remain undecidable even for very restricted variants of polynomially-ambiguous automata, \eg \emph{copyless} linear CRA~\cite{AlmagorCMP20}. However, it is not always the case, for example the $\epsilon$-gap threshold problem is decidable for polynomially-ambiguous probabilistic automata~\cite{DaviaudJLMP021}, and undecidable in general~\cite{CondonL89}.
For zero isolation and boundedness the undecidability reductions do not work for polynomially-ambiguous automata, which is the starting point of our paper.

\paragraph*{Our contributions and techniques}
We study boundedness and zero isolation for \emph{copyless linear CRA}, introduced in~\cite{AlurDDRY13}, and known to be strictly contained in polynomially-ambiguous weighted automata~\cite{AlmagorCMP20}. We show that boundedness is decidable for copyless linear CRA. Our proof shows that unboundedness can be detected with simple patterns in the style of patterns for finitely-ambiguous and polynomially-ambiguous automata in~\cite{WeberS91}. Intuitively, an automaton is unbounded if and only if either there is a loop of value larger than $1$ or there is a pattern that generates unboundedly many runs of the same value. Like in~\cite{WeberS91} the patterns are easy to detect even in polynomial time, the difficulty is to prove correctness of the characterisation. Similarly, as in one of the mentioned proofs of Hashiguchi's theorem~\cite{Simon94}, we find a way to abstract the set of generated matrices into a finite monoid, that allows us to exploit Simon's factorisation trees. Otherwise, the proof is rather different from~\cite{Simon94}, as we need to exploit the particular shapes of the matrices (imposed by the copyless restriction), while the proof in~\cite{Simon94} works for the general class of matrices. We conjecture that our pattern characterisation works for the whole class of polynomially-ambiguous automata.

For the zero isolation problem we have to further restrict the class of copyless linear CRA to a class in which the registers do not interact, that we call \simpleCRA. A similar model of CRA with independent registers was already defined in \cite{DaviaudJRV17}. We start with a chain of reductions to equivalent problems. Firstly, we show that zero-isolation over $\semiringQ$ is essentially equivalent to the boundedness problem over the semiring $\semiringZ$, \ie the same problem as in Hashiguchi's theorem with the exception that the domain includes negative numbers. This problem is known to be undecidable for the full class of weighted automata~\cite{AlmagorCMP20}, but for polynomially-ambiguous, or even copyless linear CRA, decidability was left as an open problem in the same paper. Secondly, we further reduce this problem to a variant of the coverability problem for a new class of \emph{orthant vector addition systems} (OVAS). 

The OVAS class lies between the standard VAS~\cite{CzerwinskiLLLM21} and its integer relaxation~\cite{HaaseH14}. Intuitively, in the standard VAS runs are considered only in the positive orthant, while in the integer relaxation runs go through the whole space. In OVAS every orthant has its own set of vectors that can be applied in that orthant. The \emph{universal coverability} problem asks whether from any starting point the positive orthant can be reached.
We prove that
universal coverability is decidable in dimension $3$.
The proof is nontrivial and relies on a notion of a \emph{separator} between the reachability set and the positive orthant that can be expressed in the first order logic over the reals. Depending on the encoding of the numbers, we can either rely on Tarski's theorem~\cite{DBLP:journals/jsc/Grigorev88}, or the formula might require the exponential function. In the latter case decidability depends on Schanuel's conjecture~\cite{schanuelsconj}.
Since most of the proof works in any dimension, we believe that this is an important step to prove the theorem for arbitrary dimensions. 
Interestingly, the proof relies on results about reachability for continuous VAS~\cite{BlondinFHH17}. From universal coverability we infer decidability of zero isolation for copyless linear CRA with $3$ independent registers.
More importantly, we establish a nontrivial connection between: zero isolation over $\semiringQ$; boundedness over $\semiringZ$; and our new model OVAS. We are convinced that the latter model is of independent interest. Interestingly, we show that the usual coverability problem (with a fixed initial point) in undecidable.

We leave as an open problem decidability of zero isolation for polynomially ambiguous weighted automata over $\semiringQ$.
Nevertheless we show that the problem is undecidable for copyless CRA (nonlinear). The latter class is known to be: strictly between the finitely-ambiguous and the full class of weighted automata~\cite{MazowieckiR15}; and incomparable with the polynomially-ambiguous class~\cite{MazowieckiR18,MazowieckiR19}.

The closest results to our work are presented in~\cite{DaviaudJLMP021} and~\cite{DBLP:conf/concur/ChistikovKMP20,Chistikov21}. In the first mentioned paper the authors study the containment problem for finitely-ambiguous probabilistic automata, where one of the automata is unambiguous. The latter restriction makes the problem essentially equivalent to the threshold problems for the general class of weighted automata over $\semiringQ$. The papers~\cite{DBLP:conf/concur/ChistikovKMP20,Chistikov21} deal with the Big-O problem for finitely-ambiguous weighted automata over $\semiringQ$, which given two automata asks if there is a constant $C > 0$ such that $\Aa(w) \le C \cdot \Bb(w)$ for all words $w$. By fixing $\Aa$ or $\Bb$ to a positive constant we get the zero isolation and boundedness problems, respectively. 
Boundedness is sometimes also called limitedness, but should not be confused with the finiteness problem. Finiteness asks whether the range of a weighted automaton over the rationals is finite, and is known to be decidable~\cite{BumpusHKST20,MandelS77}.

We state and organise our results in \cref{sec:results}, after formally defining the setting in \cref{sec:preliminaries}.
 
\section{Preliminaries}
\label{sec:preliminaries}

We write $\Q$, $\Q_{\ge 0}$, $\Q_{>0}$, $\Q_{\le 0}$, $\Q_{<0}$ for the sets of rationals, nonnegative rationals, etc; and we use similar notation for other domains.
Throughout the paper we assume that the base of the logarithm is $2$ unless otherwise stated.
By $\LQ$ we denote the set of logarithms of positive rational numbers: $\LQ = \{\log(q) \mid q \in \Q_{>0}\}$. Observe that $\Q \subseteq \LQ$ and that $\LQ$ is closed under addition.
For $a, b \in \N$, $a \leq b$ we write $[a,b]$ as a shorthand for $\set{a,\dots,b}$.
Given a vector $\bv = (v_1,\ldots,v_d) \in \R^d$ we write $\bv[i] = v_i$ for every $i \in [1,d]$.
For $\bv, \bw \in \R^d$ we write $\bv \le \bw$ if $\bv[i] \le \bw[i]$ for all $i \in [1,d]$.
The norm of a vector $\bv$ is defined as $\norm{\bv} = \max_{i \in [1,d]} |\bv[i]|$.
Given a finite set $Q$, where $|Q| = d$ sometimes we consider vectors in $\R^Q$ understood as vectors in $\R^d$ for some implicit bijection between $Q$ and $[1,d]$.

Let $\semiringS$ be a commutative semiring with the sum $\oplus$ and product operations $\odot$.
 We will use $\S$ to denote the domain of the semiring $\semiringS$.
In this paper most of the time we will consider two types of semirings.
The \emph{standard semiring}, where the domain is nonnegative rational numbers $\semiringQ$
with the standard sum and product operations.
The \emph{tropical semirings} $\semiringZ$ and $\semiringZR$ with domains $\Z \cup \set{+\infty}$ and $\LQ \cup \set{+\infty}$, respectively, where $\oplus$ is $\min$ and $\odot$ is $+$.
Whenever the semiring is not specified we write $\zero$ and $\one$ for the zero and one of the semiring.
Over $\semiringQ$ these are as expected $\zero = 0$ and $\one = 1$; but over $\semiringZR$ and $\semiringZ$ these are $\zero = +\infty$ and $\one = 0$.

\subsection{Weighted automata}

A \emph{weighted automaton} (WA) over a semiring $\semiringS$ is a tuple $\Aa = (\Sigma, I, F, (M_a)_{a\in \Sigma})$, where: $\Sigma$ is a finite alphabet; $I,F \in \S^d$ are $d$-dimensional vectors; and $M_a \in \S^{d\times d}$ are $d$-dimensional square matrices for some fixed $d \in \N$.
For every word $w = w_1\cdots w_n \in \Sigma^*$ we define the matrix $M_w = M_{w_1} M_{w_2}\cdots M_{w_n}$, where the matrices are multiplied with respect to the sum and the product of $\S$.
If $w$ is the empty word then $M$ is the identity matrix.
For every word $w \in \Sigma$ the automaton outputs $\Aa(w) = I^\transpose M_{w} F \in \S$.
Thus $\Aa$ can be seen as a function $\Sigma^* \to \S$.
Whilst formally $\Aa$ does not have states, one can think that coordinates in $I$, $F$, and the matrices $(M_a)_{a\in\Sigma}$ are indexed by states rather than natural numbers.
In which case, we write $q\xrightarrow{w|r}q'$ if $M_w[q,q'] = r$ (regardless of whether $w$ is a word or character).
We also say that $I[q]$ and $F[q]$ are the initial and the final value of $q$ for every state $q$.

A run $\rho$ over a word $w=w_1 \cdots w_n$ in $\Aa$ is a sequence of states interleaved with values: $q_0,v_1,q_1,\dots,v_n,q_n$ such that $q_{i-1}\xrightarrow{w_i|v_i}q_i$ for $i\in\{1,\dots,n\}$.
We then associate the \emph{value of the run} $\val({\rho}) = I[q_0] \odot  v_1\odot \ldots\odot v_n \odot  F[q_n]$.
We say that $\rho$ is an \emph{accepting run} if $\val(\rho) \neq \zero$. 
Equivalently all elements in the product $I[q_0]$, $v_1$, \ldots, $v_n$, $F[q_n]$ are different from $\zero$ (for the semirings in this paper).
We denote the set of all accepting runs of $\Aa$ over $w$ by $\Acc(\Aa,w)$.
Then $\Aa(w) = \bigoplus_{\rho_i \in \Acc(\Aa,w)}\val(\rho_i)$.
The equivalence with the matrix definition is clear for all commutative semirings since runs that are not accepting contribute $\zero$ to the sum.

Consider a weighted automaton $\Aa$.
We write that $\Aa$ is:
\begin{itemize}[leftmargin=*]
 \item \emph{unambiguous} if $|\Acc(\Aa,w)| \le 1$ for all $w \in \Sigma^*$;
 \item \emph{finitely-ambiguous} if there exists $k \in \N$ such that \\ $|\Acc(\Aa,w)| \le k$ for all $w \in \Sigma^*$;
 \item \emph{polynomially-ambiguous} if there exists a polynomial function $p$ such that $|\Acc(\Aa,w)| \le p(|w|)$ for all $w \in \Sigma^*$.
If $p$ is linear we also say that $\Aa$ is linearly-ambiguous.
\end{itemize}

Below we show two examples of weighted automata over the semiring $\semiringQ$.

\begin{example}\label{example:weightedmod2}
Consider $\Aa = (\set{a}, I, F, M_a)$, where $I = (1,0)$, $F = (0,1)$, and $M_a = \left(\begin{smallmatrix}
0&1\\
1&0
\end{smallmatrix}\right)
$.
Then $\Aa(a^n) = n \mod 2$.
The automaton $\Aa$ is unambiguous (see \Cref{fig:weightedautomata}).
\end{example}

\begin{example}\label{example:weightedn}
Consider $\Bb = (\set{a}, I, F, M_a')$, the same as $\Aa$ except that $M_a' = \left(\begin{smallmatrix}
1&1\\
0&1
\end{smallmatrix}\right)
$.
Then $\Bb(a^n) = n$.
The automaton $\Bb$ is linearly-ambiguous (see \Cref{fig:weightedautomata}).
Moreover, it can be shown that the function defined by $\Bb$ cannot be defined by a finitely-ambiguous automaton (see \eg~\cite[Lemma~12]{BarloyFLM20}).
\end{example}

\begin{figure*}
\centering
\includegraphics[width=\textwidth]{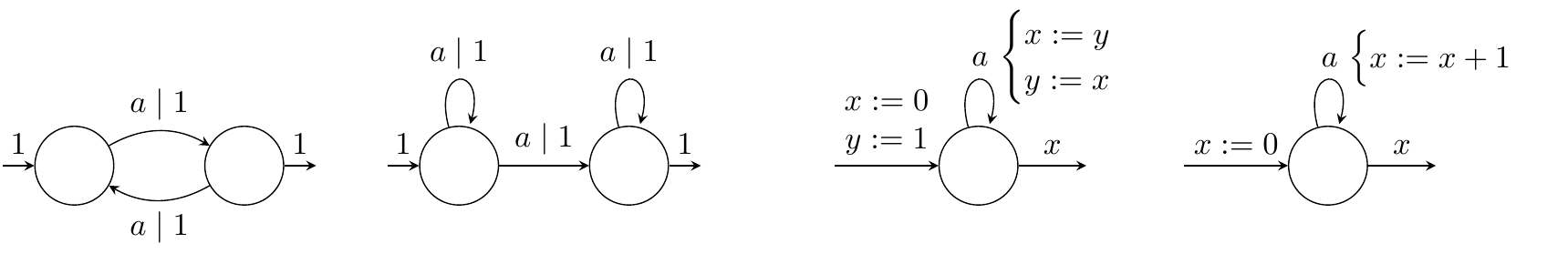}
\Description{Four automata. The first has two states with alternation between them. The second has two states with self loops and a path from the first to the second. The third has a single self loop state which swaps two registers. The fourth has a single self loop state which increments a register by 1.}
\caption{On the left: pictures for~\Cref{example:weightedmod2} and~\Cref{example:weightedn}, respectively. Vectors $I$ and $F$ are presented with ingoing and outgoing edges, respectively; and the elements of $M_a$ by edges between states. In all cases $0$ values are omitted. It is easy to see that the left automaton is unambiguous (even deterministic) and that the automaton on the right has $n$ accepting runs on $a^n$, and thus is linearly-ambiguous. On the right: pictures of stateless CRA equivalent to the ones on the left. The initial and final coefficients are presented by the ingoing and outgoing arrows, respectively. For example: on the left $F(x) = 1$, $F(y) = 0$; and on the right $F(x)=1$. Both CRA are linear and copyless.}\label{fig:weightedautomata}
\end{figure*}

\subsection{Cost-register automata and its restrictions}
For a semiring $\S$ and a set of registers $\X$ we write $\affine(\S,\X)$ for the set of affine expressions, \ie, expressions of the form $c \oplus \bigoplus_{x \in \X} (s_x \odot x)$, where $s_x , c\in\S$ and $x \in \X$.
A different presentation of weighted automata are \emph{linear cost-register automata} (linear CRA).
A linear CRA $\Bb$ is defined as a tuple $(\Sigma, Q, q_0, I, F , \X, \delta)$, where: $\Sigma$ is a finite alphabet; $Q$ is a finite set of states, with $q_0\in Q$ the designated initial state; $\X$ is a finite set of registers; $I : \X \to \S$ and $F : Q \times \X \to \S$ are, respectively, the initial values and final coefficients of registers; and $\delta\colon Q\times \Sigma \to Q \times (\X\to \affine(\S,\X))$ is a deterministic transition function.\footnote{Linear CRA were originally defined with linear updates (rather than affine).
Affine updates can be simulated by linear updates by introducing one extra register with value fixed to $\one$.
We use affine updates because the register constraints we introduce later do not apply to this special register.
}

A configuration of a CRA is a pair $(q,\sigma)$, consisting of a state and a valuation of the registers $\sigma \in  \S^{\X}$.
The initial configuration is $(q_0,I)$.
For every letter $a \in \Sigma$ we write $(q,\sigma) \xrightarrow{a} (q',\sigma')$ if  $\delta(q,a) = (q',\sigma_{q,a})$ and $\sigma'(x) = \sigma_{q,a}(x)(\sigma)$, that is $\sigma_{q,a}(x)$, where every register is substituted with its previous valuation $\sigma$.
Since $\Bb$ is deterministic for every input word $w = w_1\cdots w_n$ there is a unique \emph{run}, defined as a sequence of configurations: $(q_0,\sigma_0),\ldots, (q_n,\sigma_n)$, where $(q_0,\sigma_0)$ is the initial configuration and $(q_{i-1},\sigma_{i-1}) \trans{w_i} (q_i,\sigma_i)$ for every $i \in \set{1,\ldots,n}$.
In such a case, we write $(q_0,\sigma_0) \trans{w} (q_n,\sigma_n)$.
Finally, if the configuration after reading $w$ is $(q,\sigma_w)$ then the output of $\Bb(w)$ is $\bigoplus_{x \in \X} F[q,x] \odot \sigma_w(x)$.
Thus $\Bb$ is a function $\Sigma^* \to \S$.

Our linear CRA are defined with states $Q$, as per their first introduction~\cite{AlurDDRY13}.
However, in the general case it is not hard to see that the \emph{stateless} (or single state) model is equivalent.
Indeed, it suffices to encode states into registers and consider $Q \times \X$ as registers.
Note that this construction does not have to hold for restricted CRAs.
Thus a stateless linear CRA is defined as a tuple $\Bb = (\Sigma, I, F, \X, \delta)$.
All the notations are the same as for CRAs, but we will omit states in the stateless case, \eg a configuration of $\Bb$ is a valuation of the registers $\sigma : \X \to \S$.

We say that two automata are equivalent if they define the same function $\Sigma^* \to \S$.
Note that for every weighted automaton $\Aa$ there exists an equivalent (stateless) linear cost-register automaton $\Bb$, and conversely too.
Indeed, it suffices to identify the dimension $d$ in weighted automata with the size $|\X|$.
Then $\sigma_a(x)$ in $\delta$ can be seen as rows of matrices in $\S^{|\X| \times |\X|}$.
Formally, this is proved, \eg in~\cite[Theorem~9]{AlurDDRY13}. In the remainder of the paper we will work with linear CRA and its subclasses.

We are also interested in automata where the linear update functions are \emph{copyless}.
A valuation $\sigma : \X \to \affine(\S,\X)$ is copyless if every register $x\in \X$ occurs at most once across all affine expressions.
Formally, using the notation $\sigma(x) = \bigoplus_{z \in \X} (s_{x,z} \odot z) \oplus c_{x,z}$, for every $z \in \X$ at most one $s_{x,z}$ is different from $\zero$.
A CRA $\Bb$ is \emph{copyless} if $\sigma_{q,a}$ is copyless for every transition $\delta(q,a) = (q',\sigma_{q,a})$.

\begin{example}\label{example:CRA}
In~\Cref{fig:weightedautomata} we show that for both $\Aa(a^n) = n \mod 2$ in \Cref{example:weightedmod2} and $\Bb(a^n) = n$ in \Cref{example:weightedn} there are equivalent stateless copyless linear CRA. 
\end{example}

In general it is known that the classes are related as follows: copyless linear CRA are contained in copyless CRA, and the latter is contained in linear CRA~\cite{MazowieckiR15}.
Moreover, copyless linear CRA are contained in the class of linearly-ambiguous weighted automata~\cite[Remark 4]{AlmagorCMP20}.
A detailed presentation showing how CRA and weighted automata compare in terms of expressiveness is in~\Cref{fig:lattice}.

\subsection{\simpleCRA, a more restricted CRA}\label{subsec:simplecra}
We say that $\Bb = (\Sigma, I, F, \X, \delta)$ is an \emph{\simpleCRA} if $\Bb$ is a stateless linear CRA such that $\sigma_a(x) = (c_{a,x}\odot x) \oplus d_{a,x}$, where $c_{a,x}, d_{a,x} \in \S$ for every $\delta(a) = \sigma_a$.
In other words the new value of every register does not depend on other registers.
Observe that \simpleCRA{} are a subclass of stateless copyless linear CRA. A similar model was studied in~\cite{DaviaudJRV17}.

\begin{example}\label{example:simpleCRA}
The right automaton in~\Cref{fig:weightedautomata} is an example \simpleCRA.
It is not hard to show that there is no \simpleCRA that is equivalent to the automaton $\Aa(a^n) = n \mod 2$.
Indeed, it follows immediately from the definition that if $\Cc$ is an \simpleCRA and $\Cc(\epsilon) = \Cc(a^2) = 0$ then $\Cc(a^n) = 0$ for all $n \in \N$.
\end{example}

\subsection{Decision problems}
We define decision problems with respect to semirings.
The problems are well-defined for functions $\Sigma^* \to \S$, in particular for weighted automata, linear CRA and \simpleCRA.
The decision problems are well-defined with respect to all considered semirings, but we will mostly focus on the semiring $\semiringQ$

The \emph{$\le$-threshold problem}: given an automaton $\Aa$ and a number $c$ (from the domain of the semiring) is it the case that $\Aa(w) \le c$ for all $w \in \Sigma^*$.
The \emph{$\ge$-threshold problem} is defined similarly, where $\Aa(w) \le c$ is replaced with $\Aa(w) \ge c$.

The \emph{boundedness} problem: given an automaton $\Aa$ does there exist a finite number $c$ such that $\Aa(w) \le c$ for all $w \in \Sigma^*$.
A sequence $w_1,w_2,\ldots$ of words with $\lim_{i \to \infty}\Aa(w_i) = +\infty$ is a witness of unboundedness.

The \emph{zero isolation} problem: given an automaton $\Aa$ is it the case that there exists a positive rational number $c > 0$ such that $\Aa(w) \ge c$ for all $w \in \Sigma^*$.
A sequence $w_1,w_2,\ldots$ of words with $\displaystyle\lim_{i \to \infty}\Aa(w_i) = 0$ is a witness of nonisolated zero.
 
\section{Detailed state of the art and our results}
\label{sec:results}

We already remarked that for probabilistic automata both the $\le$-threshold and $\ge$-threshold problems are well-known to be undecidable~\cite{paz71}, even when the model is restricted to linearly ambiguous~\cite[Theorem~2]{DaviaudJLMP021}.
The zero isolation problem is also undecidable for probabilistic automata~\cite{GimbertO10}.
Hence, these three problems are also undecidable for weighted automata over $\semiringQ$. 

We remarked that the boundedness problem is not interesting for probabilistic automata, since all words have value bounded by $1$.
For weighted automata over $\semiringQ$ the problem is undecidable; we are not aware whether this fact is stated in the literature.
Nevertheless, it can be proven within this paragraph (a similar argument appears \eg in the proof of \cite[Theorem~1]{BlondelT2000boundedness}).
Consider the undecidable $\le$-threshold problem for probabilistic automata: given a probabilistic automaton $\Aa$ is it the case that $\Aa(w) \le \frac{1}{2}$ for all $w \in \Sigma^*$.
One can easily define $\Bb$ (which is no longer probabilistic, but over $\semiringQ$) such that $\Bb(w_1\#\dots\#w_n) = 2\Aa(w_1)\cdot\ldots\cdot 2\Aa(w_n)$, where $ \# $ is some fresh symbol, which intuitively restarts the automaton.
Then $\Bb$ is bounded if and only if $\Aa(w) \le \frac{1}{2}$ for all words $w$.

\begin{corollary}\label{cor:stateofart}
The $\le$-threshold, $\ge$-threshold, zero isolation, and boundedness problems are undecidable for weighted automata over the semiring $\semiringQ$.
The first two problems are undecidable even for linearly ambiguous models.
\end{corollary}

On the positive side, when ambiguity is restricted to be finitely ambiguous we can infer some decidability results for the $\le$-threshold and $\ge$-threshold problems from~\cite{DaviaudJLMP021}.

\begin{restatable}{proposition}{PropStateofArt}\label{prop:finambig:zeroboundedness}
For finitely ambiguous weighted automata over $\semiringQ$ the $\le$-threshold problem and the boundedness problem are decidable, and the $\ge$-threshold problem and the zero isolation problem are decidable assuming Schanuel's conjecture is true.
\end{restatable}

In this paper we are mostly interested in the boundedness and zero isolation problems over $\semiringQ$ for copyless linear CRA and \simpleCRA.
Below we state our main results.

\begin{restatable}{theorem}{ThmBoundednessRat}\label{theorem:boundedness_rational}
Boundedness for copyless linear CRA over $\semiringQ$ is decidable in polynomial time.
\end{restatable}

\begin{theorem}\label{theorem:isolation}
Zero isolation for \simpleCRA in dimension~$3$ over $\semiringQ$ is decidable, subject to Schanuel's conjecture. For copyless CRA zero isolation is undecidable.
\end{theorem}

As mentioned in the introduction, the main contribution of the results are: the techniques in the decidability result that we believe might generalise to arbitrary dimension; and the nontrivial connections with other problems.
To prove \Cref{theorem:isolation} we will show that the zero isolation problem is essentially equivalent to the boundedness problem over $\semiringZR$. Thus the positive part of \Cref{theorem:isolation} will be a corollary of the following (the negative part is deferred to the appendix).

\begin{theorem}\label{theorem:boundedness_min}
Zero isolation for \simpleCRA in dimension~$3$ over $\semiringZR$ is decidable, subject to Schanuel's conjecture.
For \simpleCRA in dimension~$3$ over $\semiringZ$ the boundedness problem is decidable in \exptime (independent of Schanuel's conjecture).
\end{theorem}

\begin{proof}[\textbf{Proof of \cref{prop:finambig:zeroboundedness}}]
\emph{Threshold problems:}
Consider the following containment problem: given two probabilistic automata $\Aa$ and $\Bb$ is it the case that $\Aa(w) \le \Bb(w)$ for all words $w$.
When $\Aa$ is finitely ambiguous and $\Bb$ is unambiguous then the problem is decidable~\cite[Proposition~16]{DaviaudJLMP021}.
When $\Aa$ is unambiguous and $\Bb$ is finitely ambiguous then the problem is decidable, assuming Schanuel's conjecture is true~\cite[Theorem~17]{DaviaudJLMP021}.

\begin{sloppypar}Consider an input for one of the threshold problems: a finitely ambiguous weighted automaton $\Aa = (\Sigma, I, F, (M_a)_{a\in \Sigma})$ over $\semiringQ$ and $c \in \Qpos$.
Let $N$ be the sum of all constants that appear in $I$, $F$, and $M_a$ for all $a \in \Sigma$ and let $C = \max(c,N)$.
We define the automaton  $\Aa/C = (\Sigma, I', F', (M_a')_{a\in \Sigma})$, where $I'(q) = I(q)/C$, $F(q) = F'(q)/C$, and $M_a'(p,q) = M_a(p,q)/C$.
It is easy to see that $\Aa/C$ is a probabilistic automaton and that $\Aa/C(w) = \Aa(w)/C^{|w|+2}$ for all $w \in \Sigma^*$.
It remains to observe that it is easy to define an unambiguous probabilistic automaton $\Bb$ such that $\Bb(w) = c/C^{|w| + 2}$ for all $w \in \Sigma^*$.
Thus the threshold problems can be reduced to the containment problems between $\Aa$ and $\Bb$.
We conclude by the mentioned results from~\cite{DaviaudJLMP021}.
\end{sloppypar}

\emph{Boundedness:}
Since there are finitely many runs, check that at least one run is unbounded, which occurs if and only if some accessible cycle has weight greater than one.

\emph{Zero isolation:}
We reduce to the Big-O problem, which asks whether there exists $C>0$ such that for all $w\in \Sigma^*$ $\Aa(w)\le C\cdot \Bb(w)$. The problem is decidable for finitely-ambiguous $\Aa,\Bb$ assuming Schanuel's conjecture is true~\cite[Theorem 9.2]{Chistikov21}. Let $\Aa(w) = 1$ for all $w\in \Sigma^*$.  Then there exists $C>0$ such that $\Bb(w) \ge \frac{1}{C}$ for all $w\in \Sigma^*$ (\emph{zero isolation}) if and only if $\Aa(w)$ is big-O of $\Bb(w)$.
\end{proof}

\paragraph{Organisation}
In the following sections we will prove Theorems~\ref{theorem:boundedness_rational}, \ref{theorem:isolation} and \ref{theorem:boundedness_min}.
\Cref{sec:boundedness} proves decidability of the boundedness problem for copyless linear CRA over $\semiringQ$ (\Cref{theorem:boundedness_rational}). 
\Cref{0-isolation} shows the chain of reductions from zero isolation for weighted automata over $\semiringQ$, through boundedness for weighted automata over $\semiringZR$, up to universal coverability in OVAS. 
Finally, \Cref{sec:coverability} shows that universal coverability is decidable in dimension three, proving \Cref{theorem:isolation} and \Cref{theorem:boundedness_min}. 
In~\Cref{fig:lattice} we present the results also explaining how \simpleCRA{} and copyless linear CRA relate to other classes of weighted automata in terms of expressiveness.

\begin{figure*}
\includegraphics[width=\textwidth]{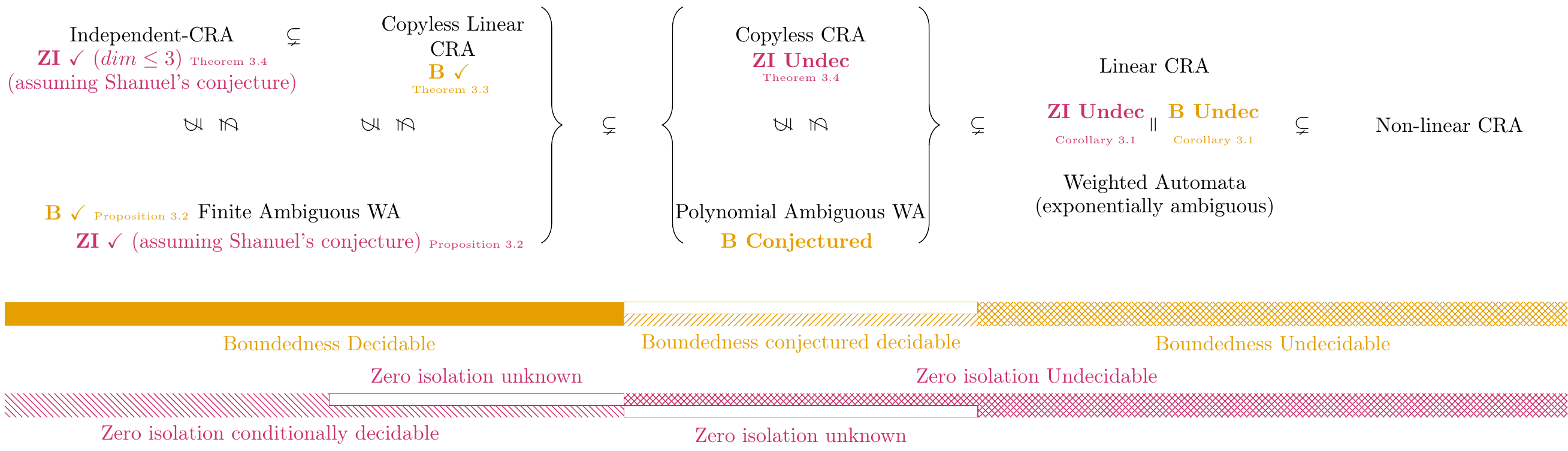}
\Description{Lattice between subclasses of weighted automata based on ambiguity and register update restrictions in the CRA setting. Annotated by the decidability status.}
\caption{Decidability status of the boundedness (B) and zero isolation (ZI) problems in the semiring $\semiringQ$.
Notation $X\subsetneq Y$ indicates that $Y$ recognises every function of $X$ in every semiring, but there exists a semiring where $Y$ recognises at least one other function not recognised by $X$.
Notations $X\not\subseteq Y, Y\not\subseteq X$ mean that there exists some semiring where a function from $X$ is not in $Y$ and visa-versa.
By~\Cref{example:weightedmod2} and~\Cref{example:simpleCRA} we see that the function $a^n \to n \mod 2$ is in the class of finitely-ambiguous weighted automata but not in \simpleCRA.
Conversely, by~\Cref{example:weightedn} and~\Cref{example:simpleCRA} we see that the function $a^n \to n$ is in \simpleCRA but not in the class of finitely-ambiguous weighted automata.
Other examples of non-inclusions go beyond copyless linear CRA and they are not always over $\semiringQ$.
See~\cite{MazowieckiR15,MazowieckiR18,MazowieckiR19,AlmagorCMP20,BarloyFLM20}. }
\label{fig:lattice}
\end{figure*}

\section{Boundedness for copyless linear CRA over \semiringQ}
\label{sec:boundedness}

The goal of this section is to establish that the boundedness problem for copyless linear CRA 
over $\semiringQ$ is decidable in polynomial time, that is, 
\Cref{theorem:boundedness_rational}.

Our first step is to translate copyless linear CRA into WA with certain properties.
More precisely, the WA will be nearly deterministic, except for a single state introducing ambiguity.
The resulting automaton will be linearly ambiguous.

To operate this transformation, we need some additional notations for WA.
Fix a WA $\Aa = (\Sigma, I, F, (M_a)_{a\in \Sigma})$ and let $\rho = q_0,v_1,q_1,\dots,v_n,q_n$ be a run in $\Aa$.
We say that $\rho$ \emph{starts} in $q_0$ and \emph{ends} in $q_n$ to indicate the first and the last state, respectively. 
We define $\valinf(\rho) = v_1 \cdots v_n$ ($\valinf(\rho) = 1$ if $n = 0$).
Thus $\val(\rho) = I(q_0) \cdot \valinf(\rho) \cdot F(q_n)$.

The run $\rho$ is a \emph{$q_0$-cycle} (or simply a cycle) if $q_0 = q_n$.
A cycle is \emph{simple} if $q_0,q_1\ldots,q_{n-1}$ are all different 
(\ie only the first and the last states are the same).
We say that $\rho'$ is a \emph{subrun} of $\rho$ if $\rho' = q_i,v_{i+1},q_{i+1},\dots,v_j,q_j$ for some $1 \le i \le j \le n$.
If $\rho'$ is also a cycle then as $q_i = q_j$
the sequence $q_0,v_1,q_1,\dots, q_{i-1}, v_{i}, q_j, v_{j+1}, \ldots v_n,q_n$ obtained
by removing $\rho'$ from $\rho$ is a run of $\Aa$.

We translate copyless linear CRA into a new subclass of WA.
Note that similar observations to the following definition and lemma were made in \cite[Proposition 2]{AlmagorCMP20}.

\begin{definition}\label{def:simple_linear}
A WA $\Bb$ is a \emph{\ssla weighted automaton} if its set of states can be written as $\{p,q_1,\dots,q_n\}$ with the following properties:
\begin{enumerate}[label=(\roman*),leftmargin=0.5cm]
\item for all $a\in\Sigma$ there is a transition $p\xrightarrow{a|1} p$ and
no other transition loops on, or enters, $p$;
and \item the automaton $\Bb$ restricted to the states $\{q_1,\dots,q_n\}$ is deterministic.
\end{enumerate}
\end{definition}

We refer to the state $p$ in \Cref{def:simple_linear} as the \emph{distinguished} state of the \ssla WA.
Notice that the only ambiguity in the automaton comes from the transitions that leave the distinguished state $p$ to some other state.

\begin{restatable}{lemma}{LemSimpleLinearAmb}\label{lemma:stupidlysimplylinearlyambiguous}
Let $\Aa$ be a copyless linear CRA over $\semiringQ$.
One can build in polynomial time a \ssla WA $\Bb$ over $\semiringQ$ such that $\Aa$ is bounded if and only if $\Bb$ is bounded.
\end{restatable}

Relying on \cref{lemma:stupidlysimplylinearlyambiguous} we may focus on \ssla WA.
We prove that unboundedness of such an automaton is characterised by certain patterns occurring in it.
\Cref{lemma:boundednesslemma} shows what happens when such patterns are not present, and it is the key technical contribution in the proof.
Then to prove \cref{theorem:boundedness_rational} we only need to detect patterns violating the assumptions of \cref{lemma:boundednesslemma}.

We define specific sets of runs based on whether they exceed a given threshold:
given $r \in \Qpos$ we set
\begin{align*}
\Runs[r](w) = \set{\rho \mid \rho \text{ is a run over } w, \valinf(\rho) > r}.
\end{align*}

\begin{lemma}\label{lemma:boundednesslemma}
Let $\Bb=(\Sigma, I, F, (M_a)_{a\in \Sigma})$ be a \ssla WA with distinguished state $p$.
Assume that for every word $u$ and every $q$-cycle $\rho$ over $u$, where $q\neq p$, both conditions hold:
\begin{enumerate}[nolistsep,label=(\roman*)]
 \item\label{pattern1} $\valinf(\rho) \leq 1$;
 \item\label{pattern2} if $\valinf(\rho) = 1 $ then $M_{u}[p,q] = 0$.
\end{enumerate}
Then $|\Runs[\frac{1}{k}](w)| \le \poly\log(k)$ for every $k\ge 2$ and every word $w$.
\end{lemma}

The constants implied by the $\poly\log$ in \cref{lemma:boundednesslemma}
depend on the rational numbers occurring in the transitions of $\Bb$.
However, it is crucial that the bound on $|\Runs[\frac{1}{k}](w)|$ does not dependent on $w$.
To get some intuition we show how the lemma concludes the proof of \cref{theorem:boundedness_rational}.

\begin{proof}[Sketch of \cref{theorem:boundedness_rational}]
If the automaton violates the assumptions of \cref{lemma:boundednesslemma}, one can construct a witness for unboundedness.
Conversely, divide all runs into $P_i = \{\rho \mid \frac{1}{c^{i}} < \valinf(\rho) \le \frac{1}{c^{i-1}}\}$ for $i \in \set{1,\ldots,n}$ and some constant $c$. 
By \cref{lemma:boundednesslemma} we have $|P_i| \le \poly(\log c^i ) = \poly(i\log c) =  \poly(i)$.
We obtain
\begin{multline*}
\Bb(w) \le \sum_{\rho \in \Runs(w)}  \val(\rho)  \le  \sum_{i=1}^n \frac{|P_i|}{c^{i-1}}  \le  \sum_{i=1}^\infty \frac{\poly(i)}{c^{i-1}}. \end{multline*}
Notice that the series converges, independent of $w$.
\end{proof}

Before establishing \cref{lemma:boundednesslemma} we need to introduce some notation and intermediary results. Roughly, our goal is to obtain a finite representation of the set of matrices $M_{w}$. This will allows us to invoke Simon's Factorisation Forest Theorem that gives a tree representation on runs on $w$, such that nodes (corresponding to subwords of $w$) have height independent on $w$. Then, intuitively, the degree of $\poly \log$ in \cref{lemma:boundednesslemma} corresponds to the height of the node.

Let $\Bb$ be as in \cref{lemma:boundednesslemma}.
As usual we will identify the dimensions of the vectors and matrices with the set of states $Q = \set{p,q_1,\ldots,q_n}$, where $p$ is the distinguished state.
Recall that: $\Bb$ is deterministic when restricted to $Q \setminus \set{p}$; $M_w[p,p] = 1$; and $M_w[q,p] = 0$ for every $q \in Q \setminus \set{p}$ and every word $w \in \Sigma^*$.
Thus for every $q \neq p$ and every $w \in \Sigma^*$, there exists at most one $q'$ such that $M_w[q,q'] > 0$.
We further observe that for every pair of states $q,q' \in Q \setminus \set{p}$ and every word $w \in \Sigma^*$ there is at most one run $\rho$ over $w$ starting in $q$ and ending in $q'$ such that $\valinf(\rho) >0$.
If there is such a run then we will denote $\rho = \rununique{q,w,q'}$ and $\valinf(\rho) = M_w[q,q']$.
We say that $r \in \Qpos$ is an admissible weight if there exists a word $w$ and a run $\rho$ over $w$ such that $\valinf(\rho) = r$.

For every $q \in Q \setminus \set{p}$ let
\begin{align}\label{eqsq}
\begin{split}
s_{q} &= \min_{w \in \Sigma^*} \set{\frac{1}{\valinf(\rho)} \mid \rho \in \Runs(w), \rho \text{ starts in } q };\\
e_{q} &= \min_{w \in \Sigma^*} \set{\frac{1}{\valinf(\rho)} \mid \rho \in \Runs(w), \rho \text{ ends in } q}.
\end{split}
\end{align}

\begin{restatable}{claim}{ClaimLower}\label{claim:lower}
$0 < s_q, e_q \le 1$ and both are computable rationals.
\end{restatable}

\begin{restatable}{claim}{ClaimFinitely}\label{claim:finitely}
Let $x > 0$.
There are finitely many admissible weights larger than $x$.
\end{restatable}

Recall that $M_w \in (\Qpos)^{Q \times Q}$ for every $w \in \Sigma^*$. Let $\varepsilon \not \in \Qpos$ be a fresh symbol.
We define the abstraction $\overline{M} \in (\Qpos \cup \set{\varepsilon})^{Q \times Q}$ as follows
\begin{align}\label{eq:matrix}
\mkern-17mu\overline{M}[q,q'] \;=\; \begin{cases}
                        0 & \text{ if } M[q,q'] = 0\\
                        M[q,q'] &\text{ if } q\neq p \text{ and } M[q,q'] \ge e_{q} \cdot s_{q'}\\
                        \varepsilon & \text{ otherwise},
                       \end{cases}
\end{align}
for all $q,q'\in Q$, and $e_q$, $s_{q'}$ as defined in \cref{eqsq}.
In words, $\overline{M}$ is the same as $M$, but some positive entries are replaced with $\varepsilon$. Notice that by \cref{claim:finitely} this set of matrices is finite, as intended.
The special symbol $\varepsilon$
appears within $\overline{M}_w$ in two cases: it replaces $M_w[q,q']$ if $q \neq p$ and this value is small enough; and it is used to indicate whether there are any positive runs from $p$ to $q'$ (their exact values are not important). 
In particular, all non-zero weights of transitions from $p$ are set to $\varepsilon$.
An example translation is in \cref{fig:translation}.
The claim below states the purpose of $\varepsilon$ formally.

\begin{figure*}
\centering
\begin{subfigure}[t]{(\textwidth-\columnsep)/2}\centering
\includegraphics[height=2.5cm]{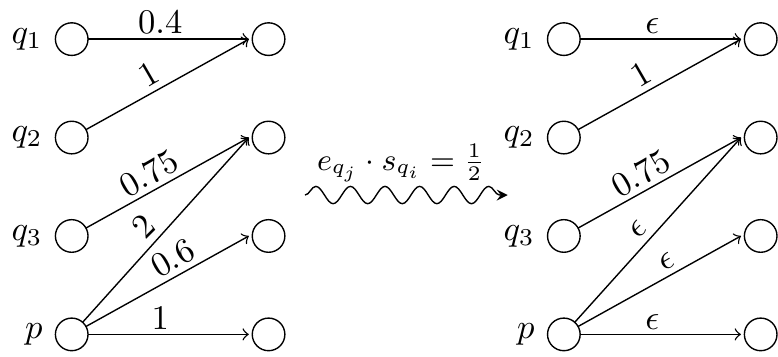}
\Description{An element translated into the monoid element.}
\caption{Translation from $M_w$ to $\overline{M}_w$. We assume that $e_{q_i} \cdot s_{q_j} = \frac{1}{2}$ for all $i$ and $j$ for simplicity.}
\end{subfigure} \hfill
\begin{subfigure}[t]{(\textwidth-\columnsep)/2}\centering\includegraphics[height=2.5cm]{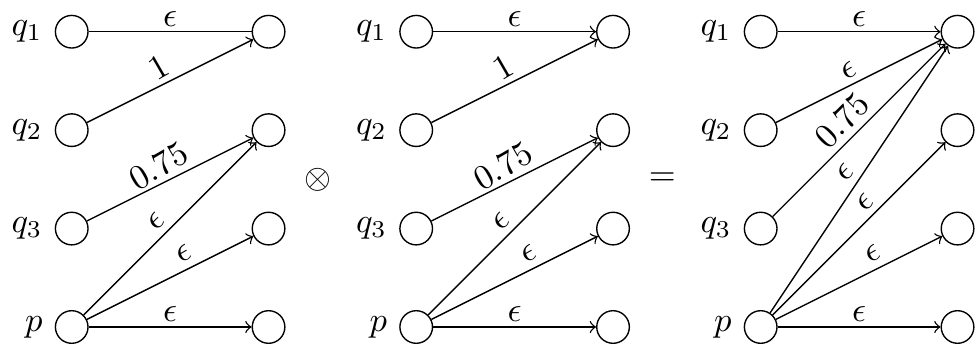}
\Description{An example of monoid multiplication.}
\caption{An example product of two elements.}
\end{subfigure}
\caption{Abstracted matrices $\Mm = \set{\overline{M}_w \mid w \in \Sigma^*}$.}
\label{fig:translation}
\end{figure*}

\begin{restatable}{claim}{ClaimEps}\label{claim:epsilon}
For every $w \in \Sigma^*$ and $q,q' \in Q \setminus \set{p}$:
\begin{enumerate}[nolistsep,leftmargin=*]
  \item if $\overline{M}_w[q,q'] \neq 0$,
then $\overline{M}_w[q,q'] = \varepsilon$ if and only if, for every run $\rho$ such that $\rununique{q,w,q'}$ is its subrun, $\valinf(\rho) < 1$.
 \item $\overline{M}_w[p,q'] = \varepsilon$ if and only if $M_w[p,q'] > 0$.
\end{enumerate}
\end{restatable}

We define the sum, product and order of $\varepsilon$ with rationals.
One can think that $\varepsilon$ represents a number above zero but `smaller' than the positive rationals. The only operation where this intuition breaks is addition, where $\varepsilon$ is an absorbing element.
This will be explained later.
Formally:
$
0 < \varepsilon < r;
\varepsilon \cdot 0 = 0 \cdot \varepsilon  = 0,\ \varepsilon \cdot r = r \cdot \varepsilon = \varepsilon \cdot \varepsilon = \varepsilon;
 \varepsilon + x = x + \varepsilon = \varepsilon + \varepsilon = \varepsilon
$
for every $x \in \Qpos$ and $r \in \Qposs$.

We define the product of abstracted matrices.
For every $M,N \in (\Qpos \cup \set{\varepsilon})^{Q \times Q}$ let $M N$ be the usual product of matrices, \ie $M N[q,q'] = \sum_{q'' \in Q} M[q,q''] \cdot N[q'',q']$.
Then we define 
\begin{align*}M \otimes N[q,q'] \;=\; \begin{cases}
                     \varepsilon & \text{ if } 0 < M  N[q,q'] < e_{q} \cdot s_{q'} \\
                     M N[q,q'] & \text{ otherwise}.
                    \end{cases}
\end{align*}
For matrices $\overline{M}_w, \overline{M}_u$ the states in $Q\setminus \set{p}$ are deterministic, thus to define $\overline{M}_w \otimes \overline{M}_u[q,q']$ for $q \neq p$ we will need to sum elements at most one of which is nonzero.
In case of $q = p$, we sum several positive elements.
However, in this case we will only be interested in whether
the transition is positive or zero; this explains our definition
of addition with $\varepsilon$.

\begin{restatable}{claim}{ClaimMonoid}\label{claim:monoid}
The set $\set{\overline{M}_w \mid w \in \Sigma^*}$ is finite and $\overline{M}_{wu} = \overline{M}_w \otimes \overline{M}_u$ for every $w,u\in \Sigma^*$.
Thus $\set{\overline{M}_w \mid w \in \Sigma^*}$ is a finite monoid with the product $\otimes$.
\end{restatable}

We let $\Mm = \set{\overline{M}_w \mid w \in \Sigma^*}$  denote the finite monoid of \cref{claim:monoid}.
An element $M \in \Mm$ is \emph{idempotent} if $M \otimes M = M$.

Consider a sequence of elements $e_1, e_2, \ldots e_n$ from $\Mm$.
A factorisation of these elements is a labelled tree whose set of nodes is a subset of $\set{(i,j) \mid 1 \le i \le j \le n}$.
Intuitively, a node $(i,j)$ corresponds to an infix $e_i,\ldots,e_j$.
Formally: the leaves are $(1,1), \ldots, (n,n)$; the root is $(1,n)$; and for every $(i,j)$ its children are $(i_0+1,i_1)$, $(i_1+1, i_2)$, \ldots, $(i_{s-1}+1, i_s)$, where $i -1 = i_0 < i_1 < i_2 \ldots < i_s = j$.
The index $i_0 = i-1$ is chosen so that even the first pair $(i_0 + 1, i_1) = (i,i_1)$ can be expressed as $(i_{x-1}+1,i_x)$.
Every node $(i,j)$ is labelled with $e_i \otimes e_{i+1} \otimes \ldots \otimes e_j \in \Mm$.
Notice that the label of every parent is equal to the product of the labels of its children in the right order.
We say that a node is idempotent if its label is idempotent.
We will use the following result from \cite{Simon90}.

\begin{lemma}[Simon's Factorisation Forest Theorem]\label{simon}
Consider a sequence of elements $e_1, e_2, \ldots, e_n$ from a finite monoid $S$.
There exists a factorisation into a tree of height at most $9|S|$ such that every inner node has either two children, or all its children are idempotents with the same label.
\end{lemma}

We can now establish \cref{lemma:boundednesslemma}.

\begin{proof}[Proof of \cref{lemma:boundednesslemma}]
Fix $k \ge 2$, $w = w_1\ldots w_n \in \Sigma^*$ and let $s_q, e_q$ be defined as in \cref{eqsq} for all $q\in Q\setminus\{p\}$.
Given $1 \le s \le t \le n$ we will denote the infix $w_{s,t} = w_s \ldots w_t$.
Let $a = \max\set{\frac{1}{e_q \cdot s_{q'}} \mid q,q' \in Q \setminus \set{p}}$.
Notice that $a \ge 1$ and that all admissible weights are bounded by $a$.
Let $b$ be some rational number such that: $b < 1$; and for all $w \in \Sigma^*$ and $q,q' \in Q \setminus \set{p}$ if $\overline{M}_w[q,q'] = \varepsilon$ then $M_w[q,q'] < b$.
Notice that by \cref{claim:finitely} $b$ is well-defined, unless $\varepsilon$ does not occur in any matrix in $\Mm$ and then $b$'s value will not be relevant (fix \eg $b = \frac{1}{2}$ then).
Consider a factorisation from \cref{simon} of $\overline{M}_{w_1}$, \ldots, $\overline{M}_{w_n}$ and let $H \le 9|\Mm|$ be its height. 
We also fix a constant $\theta = \max(2, |Q|, 1 + H\log_{\frac{1}{b}}a)$
(the choice will become clear in the following).

\begin{restatable}{claim}{ClaimIdemp}\label{claim:idempotents}
Let $M \in \Mm$ be an idempotent and let $w^1, \ldots, w^m$ be words such that $\overline{M}_{w^i} = M$ for all $i \in \set{1,\ldots,m}$.
Let $\rho \in \Runs(w^1 \ldots w^m)$ that starts in $p$ and let $\rho_1,\ldots, \rho_m$ be subruns of $\rho$ on the corresponding words $w^1,\ldots, w^m$.
If $q_1$, \ldots, $q_m$ is the sequence of states where the $\rho_1$,\ldots, $\rho_m$ end, respectively, then there exists $i \in \set{1,\ldots,m}$ such that: $q_1 = q_2 \ldots = q_{i-1} = p$ and $q_{i+1} = q_{i+2} \ldots = q_m$.
Moreover, either $i = m$ or $M[q_{i+1},q_{i+1}] = \varepsilon$.
\end{restatable}

\begin{claim}\label{claim:countingruns}
Let $(s,t)$ be a node in the factorisation of height $i$, $0 \le i \le H$.
Then
\[\abs{\Runs[\frac{1}{k}a^{i-H}](w_{s,t})} \le \left(\theta (1 + \log_{\frac{1}{b}} k)\right)^{i+1} + |Q|.\]
\end{claim}

Intuitively, either a node has not many children, then the number of runs cannot increase by a lot; or if there are many children then most runs will have a small value.

\begin{claimproof}
For simplicity we will write $\log$ for $\log_{\frac{1}{b}}$.
Since the automaton restricted to $Q \setminus \set{p}$ is deterministic, it suffices to prove that the number of runs starting in $p$ is bounded by $(\theta (1 + \log k))^{i+1}$.
We proceed by induction on $i$.
In the base case, when $i = 0$, $(s,t)$ is a leaf and $w_{s,t}$ is a letter.
Then there are at most $|Q|$ runs from $p$.
We conclude since $(\theta (1 + \log k))^1  \ge |Q|$ for $k \ge 2$ by the choice of $\theta$. 

For the induction step assume that the claim holds for all $0,\ldots, i$ and we prove it true for $i+1$.
Since $i+1 > 0$ $(s,t)$ is an inner node.
Let $(s_0+1,s_1)$, $(s_1+1, s_2), \ldots (s_{m-1}+1, s_m)$ be the children of $(s,t)$ such that $s -1= s_0 < s_1 < s_2 < \ldots < s_m = t$ and the height of every child is at most $i$.
Consider a run $\rho \in \Runs[\frac{1}{k}a^{i+1-H}](w_{s,t})$ starting in $p$.
Then $\rho$ can be decomposed into $m$ runs $\rho_1$, \ldots, $\rho_m$ over $w_{s_0+1,s_1}$, \ldots, $w_{s_{m-1}+1,s_m}$, respectively.
Notice that $\valinf(\rho) = \prod_{j=1}^m \valinf(\rho_j)$.
As $\valinf(\rho)>\frac{1}{k}a^{i+1-H}$ this means that
$\rho_x \in \Runs[\frac{1}{k}a^{i-H}](w_{s_{x-1}+1,s_x})$ for all
$x \in \set{1,\ldots,m}$. Indeed, by the choice of $a$ we know that
\[
a \ge \prod_{j=1}^{x-1} \valinf(\rho_j) \cdot \prod_{j=x+1}^{m} \valinf(\rho_j),
\]
hence
\[
\valinf(\rho_x)a \ge \valinf(\rho) >\frac{1}{k}a^{i+1-H}.
\]
We denote by $q_j$ the ending state of $\rho_j$ for $j \in \set{1,\ldots, m}$ (which is also the starting state of $\rho_{j+1}$ for $j < m$).
We consider two cases depending on the number of children $m$.

First, suppose there are two children, \ie $m =2$.
Let us count the number of possible $\rho$, depending on whether $q_1 = p$ or $q_1 \neq p$.
In the first case since there is exactly one run from $p$ to $p$, the runs $\rho$ differ only on $\rho_2$ and thus the number of such runs is bounded by
$|\Runs[\frac{1}{k}a^{i-H}](w_{s_1+1,s_2})|$.
In the second case since the transitions from $Q \setminus \set{p}$ are deterministic the number of runs is bounded by $|\Runs[\frac{1}{k}a^{i-H}](w_{s_0+1,s_1})|$.
By the induction assumption altogether this is bounded $2(\theta (1 + \log k))^{i+1} \le (\theta (1 + \log k))^{i+2}$ by the choice of $\theta$.

Second, by \cref{simon} suppose that all children are idempotents with the same label, denote it $M$.
By \cref{claim:idempotents}, there is an index $x \in \set{1,\ldots, m}$ such that $q_1 = q_2 \ldots = q_{x-1} = p$, $q_{x+1} = q_{x+2} \ldots = q_m$, and if $x < m$ then $M(q_{x+1}, q_{x+1}) =\varepsilon$.
By definition of $b$ we get that $\valinf(\rho_y) \le b$ for $x < y \le m$.
Thus $\valinf(\rho) \le a \cdot b^{m-x}$ and since $\rho\in\Runs[\frac{1}{k}a^{(i+1)-H}](w_{s,t})$ we get $a \cdot b^{m-x} \ge \frac{a^{i+1-H}}{k}$, which implies $m-x \le H\log a + \log k$.
Thus there are at most $\theta + \log k \le \theta(1 + \log k)$ valid indices for $x$.
Let us count all possible $\rho$, depending on the value $x$.
For a fixed $x$ the number of possible $\rho$ is bounded by $|\Runs[\frac{1}{k}a^{i-H}](w_{s_{x-1}+1, s_x})|$.
This is because the automaton is deterministic on $Q \setminus \set{p}$.
Thus by the induction assumption the number of all possible $\rho$ is bounded by $\theta(1 + \log k) \cdot (\theta (1 + \log k))^{i+1} = (\theta (1 + \log k))^{i+2}$.
\end{claimproof}

\cref{lemma:boundednesslemma} follows by applying \cref{claim:countingruns} with $i = H$.
\end{proof}

We conjecture that the results can be generalised to polynomially ambiguous weighted automata.
 
\section{From \simpleCRA to OVAS}
\label{0-isolation}

\label{sec:vass}

The proof of the following is delegated to the appendix.

\begin{theorem}\label{thm:rational_to_tropical}
For \simpleCRA{} the problems of zero-isolation over $\semiringQ$ and boundedness over $\semiringZR$ are interreducible in polynomial time.
\end{theorem}

The rough intuition is that for a weighted automaton $\Aa$ over $\semiringQ$ one can define $\Aa_{\log}$, where every weight $c$ is replaced with $-\log c$. Notice that $c_i \to 0$ iff $-\log c_i \to +\infty$. If $\Aa$ would be considered over $\semiringQmax$ (\ie when accepting runs are aggregated with $\max$ instead of $+$) then this theorem is essentially a syntactic translation. 
Thus the crux of \cref{thm:rational_to_tropical} is to show that it is equivalent to consider the maximum run, rather than the aggregation with $+$.

We recall some definitions to define a new VAS model.
Given a positive integer $d \in \N$ an \emph{orthant} in $\R^d$
is a subset of the form
$\{\bx = (x_1,\ldots,x_d) \colon \epsilon_1 x_1 \ge 0$, \ldots,
$\epsilon_d x_d \ge 0 \}$ for some $\epsilon_i \in \set{-1,1}$.
We write $\orthants$ for the set of all orthants. Notice that
$|\orthants| = 2^d$. For example when $d = 2$ there are four
orthants also called quadrants. Let $A_\heartsuit,A_\Diamond \in \orthants$,
where $A_s = \set{\bx \mid \epsilon_1^s x_1 \ge 0, \ldots, \epsilon_d^s x_d \ge 0}$
for $s \in \set{\heartsuit,\Diamond}$. 
We write $A_\heartsuit \preceq A_\Diamond$ if $\epsilon_i^\heartsuit \le \epsilon_i^\Diamond$ for all $i \in \set{1,\ldots,d}$. 
This is a partial order on
$\orthants$, where the \emph{negative orthant}, defined by
$x_1 \le 0$, \ldots, $x_d \le 0$, is the smallest element; and
the \emph{positive orthant}, defined by $x_1 \ge 0$, \ldots,
$x_d \ge 0$, is the largest element. Given an orthant
$A = \set{\bx \mid \epsilon_1 x_1 \ge 0, \ldots, \epsilon_d x_d \ge 0}$
we will be often interested in points
$A^{\preceq} = \set{\bx \mid \exists A_\bx \in\orthants \text{ s.t. }  A \preceq A_\bx \text{ and }  \bx \in A_\bx}$.
Let $i_1,\ldots,i_k$ be all indices such that
$\epsilon_{i_j} = 1$. Notice that
$A^{\preceq}= \set{\bx \mid \epsilon_{i_1} x_{i_1} \ge 0, \ldots, \epsilon_{i_k} x_{i_k} \ge 0}$.

Notice that some vectors belong to more than one orthant, when some of their coordinates are zero. Given a vector $\bv$ we denote by $\orthmax{\bv}$ and $\orthmin{\bv}$ the largest and the smallest orthants that contain $\bv$, respectively. Notice that these are well defined since $\preceq$ induces a lattice on $\orthants$.

We define a model related to vector addition systems over integers~\cite{HaaseH14}. Consider a positive integer $d \in \N$. A \emph{$d$-dimensional orthant vector addition system} ($d$-OVAS or OVAS if $d$ is irrelevant) is $\V = (T_A)_{A \in \orthants}$, where every $T_A$ is a finite set of vectors $T_A \subseteq \R^d$ with the following property.
If $A \preceq B$ then $T_A \subseteq T_B$. We will refer to this property as \emph{monotonicity} of $\V$.
It will be convenient to denote $T = \bigcup_{A \in \orthants} T_A$.
We define the norm of $\V$ as $\norm{\V} = \max \set{\norm{\bv} \mid \bv \in T}$. The transitions in $\V$ are encoded efficiently, \ie for every $\bv \in T$ it suffices to store the minimal orthants $A$ such that $\bv \in T_A$. Note that $\bv$ may be minimal for multiple incomparable orthants.

A \emph{run} from $\bv_0$ over $\V$ is a sequence $\bv_0,\bv_1,\ldots,\bv_n$ such that $\bv_{i+1} - \bv_i \in T_{\orthmax{v_i}}$ for all $i \in  \set{0,\ldots,n-1}$. 
If such a run exists then we write $\bv_0 \steps \bv_n$. We allow $n = 0$ and thus $\bv \steps \bv$ for every vector $\bv$.

The \emph{universal coverability} problem is defined as follows. Given a $d$-OVAS $\V$ decide if for every vector $\bv \in \R^d$, there exists $\bw$ in the positive orthant (\ie $\bw \in \R^d_{\ge 0}$) such that there is a run $\bv \steps \bw$. If there are such runs then we say that $\V$ is a positive instance of universal coverability.

The \emph{coverability} problem is similar but the initial point is fixed. Formally, given a $d$-OVAS $\V$ and a vector $\bv \in \R^d$ decide if there is a run $\bv \steps \bw$ for some $\bw\in \R^d_{\ge 0}$. 

\begin{restatable}{theorem}{ThmBoundOvas}\label{theorem:boundedness_ovass}
The boundedness problem for \simpleCRA over $\semiringZR$ and the universal coverability problem for OVAS are interreducible
in polynomial time.
\end{restatable}

\begin{proof}[Proof sketch]We only give an intuition.
Given an \simpleCRA $\Aa = (\Sigma, I, F, \X, (\delta_a)_{a \in \Sigma})$ the dimension of the OVAS  $\V$ is $|\X| = d$. 
For every letter $a \in \Sigma$ let $\delta_a(x) = \min(x + c_{a,x}, d_{a,x})$. 
The idea is that $c_{a,x}$ are the coordinates of a corresponding vector $\bc_a$ in $\V$, while $d_{a,x}$ determine the orthants in which it is available. Intuitively, $\Aa$ consumes letters in the reversed order compared to applying the corresponding vectors in $\V$. Then $d_{a,x} = +\infty$ does not impose any restrictions, while $d_{a,x} < + \infty$, means that the value of register $x$ needs to be big enough for the transition to be fired.
\end{proof}

\subsection{OVAS with continuous semantics}
\label{sec:covass}
Let $\V = (T_A)_{A \in \orthants}$ be a $d$-OVAS. A \emph{continuous run}
from $\bv$ over $\V$ is a sequence $\bv_0,\bv_1,\ldots,\bv_n$ such that for every $i \in \set{0,\ldots,n-1}$ there exists an orthant $A_i$, where: $\bv_i, \bv_{i + 1} \in A_i$;
and there exists $\delta_i \in \Rpos$ such that $\delta_i(\bv_{i+1} - \bv_i) \in T_{A_i}$.
Notice that the former implies that orthants are crossed only by pausing on the boundaries. If such a run exists then we write $\bv_0 \stepsc \bv_n$.

We say $A_i$ is the orthant witnessing the transition if $A_i$ is the maximal orthant such that $\bv_i, \bv_{i + 1} \in A_i$, and then we write that $A_{0}$, \ldots, $A_{{n-1}}$ is the witnessing sequence of orthants.

We remark that it would be possible to drop the additional restriction of pausing at the boundaries. One would have to require $\delta_i\ge 1$ (otherwise, the vector $\delta_i(\bv_{i+1} - \bv_i)$ essentially becomes available in $T_{A_{i+1}}$). Moreover, with the restriction of pausing at the boundaries the behaviour within an orthant is similar to the standard continuous VAS model~\cite{BlondinFHH17}. This will be convenient in \cref{sec:coverability}, in particular to invoke \cref{prop:cvass}.

\begin{remark}\label{remark:runs}
A continuous run, where $\delta_i = 1$ for all $i$, is also a run.
\end{remark}

The \emph{universal continuous coverability} problem is defined as the universal coverability problem, where $\bv \steps \bw$ is replaced with $\bv \stepsc \bw$. Similarly, we will say \eg that $\V$ is a positive instance of universal continuous coverability. In this subsection we will prove the following theorem.

\begin{theorem}\label{theorem:continuous_semantics}
Let $\V = (T_A)_{A \in \orthants}$ be a $d$-OVAS. Then $\V$ is a positive instance of universal coverability if and only if it is a positive instance of universal continuous coverability.
\end{theorem}

To prove the theorem we require several auxiliary lemmas about continuous runs over a $d$-OVAS $\V = (T_A)_{A \in \orthants}$. \cref{fig:fourpicturesaboutvass} shows geometric intuitions. In the following we assume that all vectors are over $\R^d$, unless specified otherwise.

\begin{figure}[!ht]
\centering
\includegraphics[width=0.48\linewidth]{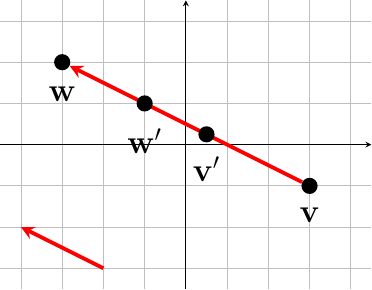}
\includegraphics[width=0.48\linewidth]{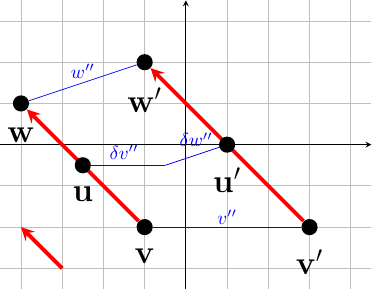}\\
\includegraphics[width=0.48\linewidth]{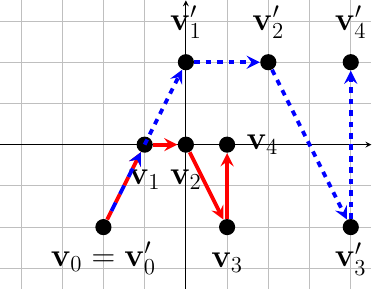}
\includegraphics[width=0.48\linewidth]{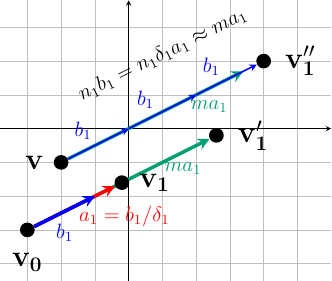}
\Description{A single line starting a v, ending at w and marked with v prime and w prime along the way.}
\Description{A single line starting at v, ending at w and marked with u along the way. Then v prime, translated by v prime prime starts a new parallel line to w prime with u prime marked along the way.}
\Description{A path starting at v zero is shown in red. Another path starting in the same place is then shown with every edge length increased, creating an expanded path with the same shape.}
\Description{A path from v zero to v one prime is converted into a path from v prime to v one prime prime.}
\caption{Examples are in dimension $2$. 
Top left: a picture for \cref{lemma:runs_continous}. The transition $(-2,1)$ is available in some orthants, where $\bv$ and $\bw$ belong. Since $\bw - \bv = 3(-2,1)$ the lemma says that for all points $\bv'$, $\bw'$ on the line from $\bv$ to $\bw$ there is a continuous run $\bv' \stepsc \bw'$ (provided $\bv'$ is closer than $\bw'$ to $\bv$). 
\\ Top right: a picture for \cref{lemma:continuous_shift}. The transition $(-1,1)$ is available in some orthants, where $\bv$ and $\bw$ belong. Since $\bw - \bv = 3(-1,1)$, $\bw' - \bv' = 4(-1,1)$, $\bv' \ge \bv$ and $\bw' \ge \bw$ there is a continuous run $\bv' \stepsc \bw'$. 
\\ Bottom left: a picture for \cref{lemma:runs_scaling}. The sequence $\bv_0,\ldots,\bv_4$ is a continuous run and $\bv_0$ has both coordinates negative. The sequence $\bv_0',\ldots,\bv_4'$ mimics the first sequence but the distances between nodes are scaled by $m = 2$. While $\bv_0',\ldots,\bv_4'$ is not a continuous run (\eg there is no orthant that contains both $\bv_2'$ and $\bv_3'$), there is a continuous run between each pair of consecutive nodes $\bv_i' \stepsc \bv_{i+1}'$ (and thus $\bv_0' \stepsc \bv_{n}'$).
\\ Bottom right: a picture for proof of \cref{theorem:continuous_semantics}, the reduction from continuous semantics to discrete semantics,  explaining $\bv,\bv_i,\bv_i',\bv_i''$ for first step.
}

\label{fig:fourpicturesaboutvass}
\end{figure}

Given two vectors $\bv$ and $\bw$, we define the set of maximal orthants on the path from $\bv$ to $\bw$:
$$
\orthantpath{\bv,\bw} = \set{\orthmax{\delta\bv + (1-\delta)\bw} \mid \delta\in[0,1]}.
$$

\begin{restatable}{lemma}{LemRunsCont}\label{lemma:runs_continous}
Fix some vectors $\bv$ and $\bw$. Suppose that, for all $C\in\orthantpath{\bv,\bw}$, there exists $\delta_C \in \R_{>0}$ such that $\delta_C(\bw - \bv) \in T_{C}$. Then $\bv' \stepsc \bw'$ for every $\bv' =  \lambda_1\bv + (1 - \lambda_1)\bw$, $\bw' = \lambda_2\bv + (1 - \lambda_2)\bw$, where $0 \le \lambda_2 \le \lambda_1 \le 1$.
\end{restatable}
\begin{restatable}{lemma}{LemContShift}\label{lemma:continuous_shift}
Let $\bv \stepsc \bw$. Suppose, for all $C\in\orthantpath{\bv,\bw}$, there exists $\delta_C \in \R_{>0}$ such that $\delta_C(\bw - \bv) \in T_{C}$. Let $\bv' \ge \bv$ and $\bw' \ge \bw$ be such that $\bw' - \bv' = \delta' (\bw - \bv)$ for some $\delta' \in \R_{>0}$. Then $\bv' \stepsc \bw'$.
\end{restatable}

\begin{restatable}{lemma}{LemRunsScale}\label{lemma:runs_scaling}
Let $\bv_0,\bv_1,\ldots,\bv_n$ be a continuous run such that $\bv_0 \in \R_{< 0}^d$ and let $\ba_i = \bv_i - \bv_{i-1}$ for $i \in \set{1,\ldots,n}$. Let $m \ge 1$ and consider the sequence $\bv_0',\bv_1',\ldots,\bv_n'$, defined by $\bv_0' = \bv_0$, $\bv_{i}' = \bv_{i-1}' + m\ba_i$. Then $\bv_{i}' \stepsc \bv_{i+1}'$ for all $i \in \set{0,\ldots,n-1}$.
\end{restatable}

We also need the following technical lemma, which is a direct consequence of the simultaneous version of the Dirichlet's approximation theorem. Intuitively, it says that given a finite set of reals we can multiply them all with the same natural number so that all resulting numbers are arbitrarily close to integers.

\begin{restatable}{lemma}{LemKReals}\label{lemma:kreals} \cite[Theorem 1A]{Schmidt1980}
Let $r_1,\ldots r_k \in \R$. For every $\epsilon > 0$ there exists $m \in \Npos$ such that for every $i \in \set{1,\ldots,k}$ there exists $z_i \in \Z$, where $|mr_i - z_i| < \epsilon$.
\end{restatable}

\begin{proof}[Proof of \cref{theorem:continuous_semantics}]
We show how to find a continuous witnessing run from a discrete run, and visa versa. 
A discrete run is not a continuous run, because a continuous run must pause at the boundary, and can only go further into the adjacent orthant if this direction remains available. We take a run from a sufficiently smaller starting point and shift it to our desired starting point. After the shift, if we cross orthants then directions will be always available due to monotonicity. 

Conversely, when converting a continuous run to a discrete run, we face a continuous run using non-integer multiples of vectors. We use \cref{lemma:runs_scaling} to scale-up the run, using \cref{lemma:kreals} to find a multiple so that the scaled real coefficients are sufficiently close to an integer. 

Let $\V = (T_A)_{A \in \orthants}$ be a $d$-OVAS.
We start with the implication that if $\V = (T_A)_{A \in \orthants}$ is a positive instance of universal coverability then it is also a positive instance of universal continuous coverability. Thus, fix any $\bv \in \R^d$. We aim to prove that there is a continuous run $\bv \stepsc \bw$ for some $\bw \in \R^d_{\ge 0}$. 

Let $\bu = (\norm{\V},\ldots,\norm{\V}) \in \R^d$. By assumption there is a run $\bv_0,\ldots, \bv_n$, where $\bv_0 = (\bv - \bu)$ and $\bv_n \in \R_{\ge 0}$. Let $\ba_{i+1} = \bv_{i+1} - \bv_{i}$ for $i \in \set{0,\ldots,n-1}$, \ie the differences between consecutive elements in the run. Since $\ba_{i+1}$ is a transition, $\norm{\ba_{i+1}} \le \norm{\bu}$. 
To conclude the proof we need to show that $\bv_{i} + \bu \stepsc \bv_{i+1} + \bu$ for every $i \in \set{0,\ldots,n-1}$; indeed, $\bv_0 + \bu = \bv$ and $\bv_n + \bu \in \R_{\ge 0}^d$. 

Since $\bv_0,\dots,\bv_n$ is a run, then $(\bv_{i+1}- \bv_{i})\in T_{\orthmax{\bv_{i}}}$ for every $i\in \set{0,\ldots,n-1}$.
Let $\bv_{i,\epsilon}= \epsilon(\bv_{i+1}+\bu) + (1-\epsilon)(\bv_{i}+\bu)$, for $\epsilon\in[0,1)$. We argue that
$\delta(\bv_{i+1}- \bv_{i})=\delta((\bv_{i+1}+\bu)-(\bv_i+\bu))\in T_{\orthmax{\bv_{i,\epsilon}}}$
for some $\delta$. Indeed
\begin{multline*}
\bv_{i,\epsilon}= \epsilon(\bv_{i+1}+\bu) + (1-\epsilon)(\bv_{i}+\bu)\\ = \epsilon(\bv_{i} + \ba_{i+1}) + (1-\epsilon)(\bv_{i})+\bu = \bv_{i}+\bu+\epsilon{\ba_{i+1}}.
\end{multline*}
Since $\norm{\bu} \ge \norm{\ba_{i+1}}\ge \norm{\epsilon\ba_{i+1}}$, we have $\bv_{i,\epsilon} \ge \bv_{i}$, thus, by monotonicity $\delta(\bv_{i+1}- \bv_{i})\in T_{\orthmax{\bv_{i,\epsilon}}}$.
Hence $\bv_{i}+\bu\stepsc \bv_{i+1} + \bu$ by \cref{lemma:runs_continous}.

Conversely, we fix $\bv \in \R^d$ and aim to prove that there is a run $\bv \steps \bw$ for some $\bw \in \R^d_{\ge 0}$. Take any vector $\bv_0 \in \R_{<0}$ such that $\bv_0 + \bf{1} \le \bv $, where ${\bf 1} = (1,\ldots,1)$. By assumption there is a continuous run $\bv_0,\bv_1,\ldots,\bv_n$ for some $\bv_n \in \R_{\ge 0}$. Let $A_0,\ldots, A_{n-1}$ be the corresponding witnessing sequence of orthants; $\ba_{i+1} = \bv_{i+1} - \bv_i$; $\delta_{i+1} \in \R_{>0}$ be such that $\delta_{i+1}\ba_{i+1} \in T_{A_i}$ for all $i \in \set{0,\ldots,n-1}$; and $\bb_{i+1}= \delta_{i+1}\ba_{i+1}$.

By \cref{lemma:kreals} there exists $m \in \N \setminus \set{0}$ such that for every $i \in \set{0,\ldots,n-1}$ there exist $n_{i+1} \in \N$ such that $|\frac{m}{\delta_{i+1}} - n_{i+1}| < \frac{1}{n\norm{\V}} $. 
Consider the scaled run $\bv_0',\ldots,\bv_n'$ defined by $\bv_0' = \bv_0$, $\bv_i' = \bv_{i-1}' + m\ba_i$. By \cref{lemma:runs_scaling} $\bv_{i}' \stepsc \bv_{i+1}'$ for all $i \in \set{0,\ldots,n-1}$. Moreover, it is easy to see that $\bv_n' \in \R_{\ge 0}$.
We define the sequence of points $\bv_0'',\ldots,\bv_n''$ as follows: $\bv_0'' = \bv$ and $\bv_{i}'' = \bv_{i-1}'' + n_i \bb_i$. Intuitively, this is like the $\bv_{i}'$ run approximated to integers, and shifted by $\bf{1}$ to compensate for errors. The rest of the proof is dedicated to show first that $\bv_0'' \stepsc \bv_n''$ and then that $\bv_0'' \steps \bv_n''$, which will conclude the proof. The first step is depicted in \cref{fig:fourpicturesaboutvass}.

We prove that $\bv_{i}'' \stepsc \bv_{i+1}''$ for every $i \in \set{0,\ldots,n-1}$. Notice that $\frac{1}{n_i}(\bv_{i+1}'' - \bv_{i}'') = \bb_{i+1} \in T_{A_i}$. Moreover, 

\begin{align*}
\bv_i'' \; = \; & \bv + \sum_{j = 1}^{i} n_j \bb_j \; \ge \; \bv_0' + {\bf 1} + \sum_{j = 1}^{i} n_j \bb_j \\[-0.1cm] &\qquad\qquad\qquad\qquad\qquad \text{by the choice of } \bv_0 = \bv_0' \\[-0.2cm]
\ge \; & \bv_0' + {\bf 1} + \sum_{j = 1}^{i}\left(\frac{m}{\delta_{j}} \bb_j - \frac{1}{n\norm{\V}} \norm{\bb_j}\right)\\[-0.1cm] &\qquad\qquad\qquad\qquad\qquad \text{by the choice of } m \\[-0.15cm]
\ge \; & \bv_0' + {\bf 1} + \sum_{j = 1}^{i}\left(m\ba_j - \frac{1}{n}{\bf 1} \right)\\[-0.1cm]&\qquad\qquad\qquad\qquad\qquad \bb_j = \delta_j\ba_j \text{ and } \norm{\bb_j} \le \norm{\V} \\[-0.3cm]
= \; & \frac{n-i}{n} {\bf 1} + \bv_0' + \sum_{j = 1}^{i} m\ba_j \; \ge \; \bv_i'. 
\end{align*}
Since $\bv_{i}' \stepsc \bv_{i+1}'$ by \cref{lemma:continuous_shift} we get $\bv_{i}'' \stepsc \bv_{i+1}''$.

To conclude we need to prove that $\bv_{i}'' \steps \bv_{i+1}''$ for every $i \in \set{0,\ldots,n-1}$. Recall that $\bv_{i+1}'' - \bv_i'' = n_{i+1}\bb_{i+1}$ and that $n_{i+1}$ is a natural number. Let $\bu_0, \ldots, \bu_{n_{i+1}}$ be defined as $\bu_0 = \bv_i''$ and $\bu_{j+1} = \bu_j + \bb_{i+1}$. Notice that $\bu_{n_{i+1}} = \bv_{i+1}''$. By \cref{lemma:runs_continous} and \cref{remark:runs} $\bu_{j} \steps \bu_{j+1}$ for every $j \in \set{0,\ldots,n_{i+1} - 1}$. Therefore, $\bv_{i}'' \steps \bv_{i+1}''$, which concludes the proof. 
\end{proof}
 
\section{Coverability for OVAS}
\label{sec:coverability}
In this section we present our undecidability and decidability results for coverability and universal coverability, respectively.
In the definition of $d$-OVAS for each orthant $A$ the set of transitions $T_A$ is a subset of $\R^d$.
For a set $S \subseteq \R$ an OVAS \emph{over S} is an OVAS using numbers only from
$S$ in its transitions, namely for each orthant $A$ we have $T_A \subseteq S^d$.

\begin{restatable}{theorem}{ThmCoverUndec}\label{theorem:coverability}
The coverability problem for OVAS over $\Z$ is undecidable.
\end{restatable}

Our two decidability results are the following.

\begin{theorem}\label{theorem:q-universal}
The universal continuous coverability problem for $3$-OVAS over $\Q$ is decidable in \exptime.
\end{theorem}

\begin{theorem}\label{theorem:logq-universal}
Assuming Schanuel's conjecture the universal continuous coverability problem for $3$-OVAS over $\LQ$ is decidable.
\end{theorem}

Together with \cref{thm:rational_to_tropical}, \cref{theorem:logq-universal} completes the proof of \cref{theorem:boundedness_min}.

\begin{remark}
Computations with elements from $\LQ$ require considerable care. Whilst they are easy to represent (\eg by storing $2^x\in \Q$ in place of  $x\in\LQ$), note that $(\LQ,+,\cdot)$ is not a semiring.  In particular, the product of two elements, \eg $\log_2(a)\log_2(b)$, is not necessarily an element of $\LQ$. In general, such computations are curiously difficult; for example (unconditionally) deciding whether $\log_2(a)\log_2(b)\le\log_2(c)\log_2(d)$ for $a, b, c, d \in \Q_{>0}$ is, to the best of our knowledge, open and related to the four exponential conjecture which asks if they can ever be equal (see e.g.,\cite[Sec. 1.3 and 1.4]{Waldschmidt00}). However, Schanuel's conjecture implies decidability of the first order theory of the reals with exponential function $FO(+,\cdot,\exp,<)$~\cite{schanuelsconj}. In particular, this allows arithmetic operations between elements of $\LQ$. 
\end{remark}

The rest of this section is devoted to the proofs of Theorems~\ref{theorem:q-universal}~and~\ref{theorem:logq-universal}.
We slowly introduce required notions and at the end we show how the developed techniques allow to prove both theorems.
Most of our steps will work for a $d$-OVAS in any dimension $d$.

We define a notion of a \emph{separator} with a property
that an OVAS $V$ is a negative instance of the universal coverability problem if and only if there exist a separator for $V$.
Finally we show that the existence of a separator in $3$-OVAS over $\Q$ can be expressed in the first order logic $FO(+,\cdot,<)$
with bounded quantifier alternation which is decidable in \exptime due to Tarski's theorem~\cite{DBLP:journals/jsc/Grigorev88}.
The existence of a separator in $3$-OVAS over $\LQ$
can be expressed in the first order logic $FO(+,\cdot,\exp,<)$ which is decidable, subject to Schanuel's conjecture~\cite{schanuelsconj}.

Given a $d$-OVAS we define the \emph{walls set} $W = \big\{\bv \in \R^d \mid \exists i \in [1,d] : \bv_i = 0\big\}$, \ie
vectors in $\R^d$ with some coordinate equal to zero. The set $W$ contains all of the faces of $d$-dimensional
orthants. Recall that the negative orthant is defined as $\R_{\leq 0}^d$. We define the \emph{strictly negative orthant}
as $\R_{< 0}^d = \R_{\le 0} \setminus W$.
For a $d$-OVAS $V$ let 
$$
\reach(V) = \set{\bv \mid \exists \bu \in \R_{<0}^d:  \bu \stepsc \bv}
$$
be the set of all
vectors reachable from the strictly negative orthant. 
We observe the following.

\begin{restatable}{claim}{ClaimReachabilitySet}\label{cl:reachability-set}
A $d$-OVAS $V$ is positive instance of universal continuous coverability problem if and only if
$\reach(V) \cap \R_{\ge 0}^d \neq \emptyset$.
\end{restatable}

By \cref{cl:reachability-set} it is enough to focus on deciding whether $\reach(V)$ intersects the positive orthant. 
For $S \subseteq \R^d$ we define its \emph{downward closure} as $S^\downarrow = \set{\bv \mid \exists \bs \in S : \bv \le \bs}$.
Suppose $S \subseteq T \subseteq \R^d$. We say that $S$ is \emph{downward closed inside $T$} if $S = S^\downarrow \cap T$.
Notice that $\reach(V)$ intersects the positive orthant if and only if $\reach(V)^\downarrow$ intersects the positive orthant.
\newcommand{\rwd}{\reach^\downarrow_W(V)}
Furthermore, let
\[
\rwd = \reach(V)^\downarrow \cap W.
\]

\begin{restatable}{claim}{ClaimRWD}\label{cl:rwd}
For every $d$-OVAS $V$ the following are equivalent:
$\rwd \cap \R_{\ge 0}^d \neq \emptyset$ if and only if
$\reach(V) \cap \R_{\ge 0}^d \neq \emptyset$.
\end{restatable}

We will focus on deciding whether $\rwd$ intersects the positive orthant.
For each orthant $A$ and $\bu_0, \bu_n \in A$ we write $\bu_0 \stepsca \bu_n$ if there is a continuous run $\bu_0, \ldots, \bu_n$ such that
$\bu_i \in A$ for all $0 \le i \le n$.

\begin{definition}\label{definition:separator}
Given a $d$-OVAS $V$
we say that $S \subseteq W$ is a \emph{separator for $V$} if the following conditions are satisfied:
\begin{enumerate}[leftmargin=*]
  \item $S$ is closed under scaling, namely for every $\lambda > 0$ and $\bs \in S$ we have $\lambda \cdot \bs \in S$;
  \item $S$ is downward closed inside $W$, namely $S = S^\downarrow \cap W$;
\item for every orthant $A$ if $u \in S$, $v \in W$ and $u \stepsca v$ then $v \in S$;
  \item $S \cap \R_{\ge 0}^d = \emptyset$.
\end{enumerate}
\end{definition}

\begin{restatable}{lemma}{LemmaSeparator}\label{lem:separator}
For every OVAS $V$:
$\rwd \cap \R_{\ge 0}^d = \emptyset$ if and only if
there exists a separator for $V$.
\end{restatable}

We aim to show that the existence of a separator in $3$-OVAS can be expressed in appropriate first order logics.
It is helpful to use the following observation about continuous VASes from~\cite{BlondinFHH17}, which helps us
to construct the needed first order sentences.

\goodbreak
\begin{proposition}[Reformulation of Proposition~3.2 in~\cite{BlondinFHH17}]\label{prop:cvass}
Fix a $d$-OVAS $V$ and an orthant $A \in \orthants$. Consider two vectors of variables $\bx = (x_1,\ldots,x_d)$ and $\by = (y_1,\ldots,y_d)$.
There is an existential formula $\varphi_{A}(\bx,\by)$ such that
$$
\sem{\varphi_{A}} = \set{(\bv,\bw) \in A^2 \mid \bv \stepsca \bw}.
$$
If $V$ is over $\Q$ then $\varphi_{A} \in FO(+,\cdot,<)$, and if $V$ is over $\LQ$ then $\varphi_{A} \in FO(+,\cdot,\exp,<)$.
\end{proposition}

Notice that in our setting we work over reals, while in~\cite{BlondinFHH17} they work over rationals.
The results in~\cite{BlondinFHH17} are stated for the logic $FO(+,<)$ over $\Q$, but it is easy to see that the same formulas work for our logics over $\R$. The reason why we need to consider irrational numbers is that whenever we deal with a number of the form $\log(\frac{p}{q})$ then we express this by $\exp(x) = \frac{p}{q}$.

The following lemma concludes the proofs in this section.

\begin{lemma}\label{lem:fo}
The existence of a separator in $3$-OVAS:
\begin{enumerate}[leftmargin=*]
  \item over $\Q$ is expressible in $FO(+, \cdot, <)$ over reals with fixed number of quantifier alternations;
  \item over $\LQ$ is expressible in $FO(+, \cdot, \exp, <)$ over reals.
\end{enumerate}
\end{lemma}

\begin{proof}
The key observation is that a set $S \subseteq W_3$, which is downward closed and closed under scaling,
can be described by at most $18$ real numbers.
Notice first that the set $W_3$ is a union of $12$ quarters of a plane.
Indeed, $W_3$ consists of three planes (defined by $x[1] = 0$, $x[2] = 0$ and $x[3] = 0$). Each of the three
planes is divided into exactly four quarters. Thus a quarter $Q$ is described by the choice of $i \in \{1,2,3\}$ such that
$\bx \in Q$ if $\bx[i] = 0$ and two signs for the other coordinates. For example consider a quarter $Q$ such that if $\bx \in Q$ then $\bx[3] = 0$.
Then $Q$ is determined by $\eps_1, \eps_2 \in \set{-1,1}$, defining $Q = \set{x \in \R^3 \mid x[1] \cdot \eps_1 \geq 0,
x[2] \cdot \eps_2 \geq 0, x[3] = 0}$. So every quarter is determined by a triple $(s_1, s_2, s_3) \in \set{+,-,0}^3$
such that there is exactly one zero among $s_1$, $s_2$ and $s_3$.
We will show that for every quarter $Q$ the set $S \cap Q$
can be described using either one or two real numbers. To simplify the notation whenever $Q$ is fixed we will think of quarters $Q$ and $S$ as subsets of $\R^2$ (projecting on the coordinates that are not fixed to $0$ in $Q$).

Given a quarter $Q$ consider two cases: 1) $\eps_1 = \eps_2$, 2) $\eps_1 \neq \eps_2$.
We show that in the first case if $S \cap Q \neq \emptyset$ then $Q \subseteq S$.
Indeed, assume $\bx \in S \cap Q$ and let $\by \in Q$. We aim to show that $\by \in S$.
Since $\eps_1 = \eps_2$, there are $\lambda \in \R_{\ge 0}$ and $\bv \in \R_{\le0}^2$
such that $\by = \lambda \bx + \bv$. As $\bx \in S$ and $S$ is closed under scaling we have $\lambda \bx \in S$.
As $S$ is downward closed and $\bv \in \R_{\le0}^2$ we have $\by = \lambda \bx + \bv \in S$.
To conclude either $S \cap Q$ is empty or it is the full quarter. Thus such quarters can be described by one variable $b \in \set{0,1}$
(one bit of information: $0$ for empty set, $1$ for full set).

Consider the second case when $\eps_1 \neq \eps_2$ and assume without loss of generality that $\eps_1 = 1$ and $\eps_2 = -1$. Thus
$Q = \set{(x_1, x_2) \mid x_1 \geq 0, x_2 \leq 0}$.
We observe that $Q \cap S$ is actually a part of $Q$ which is below some line $\alpha$. This will be the only step where we use the assumption that our OVAS is in $3$-dimensions (which implies that $Q$ is in $2$ dimensions).
Formally, a line $\set{(x_1, x_2) \mid \alpha_1 x_1 = \alpha_2 x_2}$ in $Q$ is described by  $\alpha = (\alpha_1,\alpha_2)$  with  $\alpha_1,\in \R_{\geq 0}$, $\alpha_2,\in \R_{\le 0}$, such that at least one $\alpha_i \neq 0$ (the sign of $\alpha_i$ comes from $\eps_i$).

\begin{restatable}{claim}{ClaimLine}\label{claim:line}
$
Q \cap S = \set{(x_1, x_2) \in Q \mid \alpha_1 x_1 \vartriangleleft \alpha_2 x_2},
$
for some line $\alpha$ in $Q$, where $\vartriangleleft$ is either $\le$ or $<$.
\end{restatable}

We note that this claim is not true in higher dimensions, and it is the main obstacle for the proof to work in general.
By \cref{claim:line} the set $Q \cap S$ is described by two real numbers (recall that $\alpha_i$ also determines the value of $\eps_i$).
Summarising: $6$ quarters need $1$ bit of description, and $6$ quarters need $2$ real numbers. In total $18$ numbers
are needed to describe the downward closed set $S$, which is also closed under scaling.

Note that our description will satisfy conditions 1 and 2 in \cref{definition:separator}. In order to check that given $18$ real numbers describe a separator it remains to check that:
\begin{itemize}[leftmargin=*]
  \item the descriptions of quarters are consistent on the intersections of quarters;
  \item conditions 3 and 4 in \cref{definition:separator} are satisfied.
\end{itemize}
It is not hard to see that the first item can be described in $FO(+,\cdot,<)$, we just need to guarantee that quarters are consistent
on the intersecting lines. In order to check condition 4 it is enough to guarantee that the quarters $(+,+,0)$,
$(+,0,+)$ and $(0,+,+)$ are all described by the bit $0$.
The most involved part is to check condition 3. 
Here, we invoke \cref{prop:cvass} that defines the reachability formula $\varphi_{A}(\bx,\by)$.

We express condition 3 as follows:
for every orthant $A$ if $\bx \in S$ and $\varphi_A(\bx, \by)$ then $\by \in S$.
It is easy to transform the above description to sentences of first order logic.
Moreover, observe that these sentences have quantifier alternation at most two: there exists
a separator $S$, such that for all $\bx \in S$ there exists $\by \in S$ fulfilling $\varphi_A(\bx,\by)$, where $\varphi_A$ has only
existential quantifiers. We have proved \cref{lem:fo}.
\end{proof}

\begin{remark}
Notice that in $d$-OVAS for $d > 3$ it is not clear whether a separator can be described by a bounded number of real numbers as \cref{claim:line} is no longer true.
A natural generalisation of techniques used in the proof of Lemma~\ref{lem:fo}
would result in expressing the existence of a separator in monadic second-order logic ${MSO(+,\cdot,<)}$ over the reals.
However, the validation problem for ${MSO(+,\cdot,<)}$ is undecidable as it is easy to express natural numbers in $MSO(+,\cdot)$ as follows:
$\N$ is the smallest set of numbers containing $1$ and closed under adding $1$. Thus decidability of ${MSO(+,\cdot,<)}$ over $\R$
would imply decidability of ${MSO(+,\cdot,<)}$ over $\N$ and in particular decidability of $FO(+,\cdot,<)$ over $\N$,
which is well known to be undecidable~\cite{Goedel1931}.
Extending the techniques used in the proof of Lemma~\ref{lem:fo} would probably require showing that we can describe separators in higher dimensional spaces using a bounded number of real number.
\end{remark}

\begin{acks}
This work was partially supported by the ERC grant INFSYS, agreement no. 950398.
We thank the anonymous reviewers for their helpful comments. In particular for the pointer to Dirichlet's approximation theorem (\cref{lemma:kreals}) and the reference to \simpleCRA~\cite{DaviaudJRV17}.
\end{acks}

\bibliographystyle{ACM-Reference-Format}

\begin{thebibliography}{40}



\ifx \showCODEN    \undefined \def \showCODEN     #1{\unskip}     \fi
\ifx \showDOI      \undefined \def \showDOI       #1{#1}\fi
\ifx \showISBNx    \undefined \def \showISBNx     #1{\unskip}     \fi
\ifx \showISBNxiii \undefined \def \showISBNxiii  #1{\unskip}     \fi
\ifx \showISSN     \undefined \def \showISSN      #1{\unskip}     \fi
\ifx \showLCCN     \undefined \def \showLCCN      #1{\unskip}     \fi
\ifx \shownote     \undefined \def \shownote      #1{#1}          \fi
\ifx \showarticletitle \undefined \def \showarticletitle #1{#1}   \fi
\ifx \showURL      \undefined \def \showURL       {\relax}        \fi
\providecommand\bibfield[2]{#2}
\providecommand\bibinfo[2]{#2}
\providecommand\natexlab[1]{#1}
\providecommand\showeprint[2][]{arXiv:#2}

\bibitem[Almagor et~al\mbox{.}(2020)]{AlmagorCMP20}
\bibfield{author}{\bibinfo{person}{Shaull Almagor},
  \bibinfo{person}{Micha{\"{e}}l Cadilhac}, \bibinfo{person}{Filip Mazowiecki},
  {and} \bibinfo{person}{Guillermo~A. P{\'{e}}rez}.}
  \bibinfo{year}{2020}\natexlab{}.
\newblock \showarticletitle{Weak Cost Register Automata are Still Powerful}.
\newblock \bibinfo{journal}{\emph{Int. J. Found. Comput. Sci.}}
  \bibinfo{volume}{31}, \bibinfo{number}{6} (\bibinfo{year}{2020}),
  \bibinfo{pages}{689--709}.
\newblock
\urldef\tempurl \url{https://doi.org/10.1142/S0129054120410026}
\showDOI{\tempurl}


\bibitem[Alur et~al\mbox{.}(2013)]{AlurDDRY13}
\bibfield{author}{\bibinfo{person}{Rajeev Alur}, \bibinfo{person}{Loris
  D'Antoni}, \bibinfo{person}{Jyotirmoy~V. Deshmukh}, \bibinfo{person}{Mukund
  Raghothaman}, {and} \bibinfo{person}{Yifei Yuan}.}
  \bibinfo{year}{2013}\natexlab{}.
\newblock \showarticletitle{Regular Functions and Cost Register Automata}. In
  \bibinfo{booktitle}{\emph{28th Annual {ACM/IEEE} Symposium on Logic in
  Computer Science, {LICS} 2013, New Orleans, LA, USA, June 25-28, 2013}}.
  \bibinfo{pages}{13--22}.
\newblock
\urldef\tempurl \url{https://doi.org/10.1109/LICS.2013.65}
\showDOI{\tempurl}


\bibitem[Barloy et~al\mbox{.}(2020)]{BarloyFLM20}
\bibfield{author}{\bibinfo{person}{Corentin Barloy},
  \bibinfo{person}{Nathana{\"{e}}l Fijalkow}, \bibinfo{person}{Nathan Lhote},
  {and} \bibinfo{person}{Filip Mazowiecki}.} \bibinfo{year}{2020}\natexlab{}.
\newblock \showarticletitle{A Robust Class of Linear Recurrence Sequences}. In
  \bibinfo{booktitle}{\emph{28th {EACSL} Annual Conference on Computer Science
  Logic, {CSL} 2020, January 13-16, 2020, Barcelona, Spain}}
  \emph{(\bibinfo{series}{LIPIcs}, Vol.~\bibinfo{volume}{152})},
  \bibfield{editor}{\bibinfo{person}{Maribel Fern{\'{a}}ndez} {and}
  \bibinfo{person}{Anca Muscholl}} (Eds.). \bibinfo{publisher}{Schloss Dagstuhl
  - Leibniz-Zentrum f{\"{u}}r Informatik}, \bibinfo{pages}{9:1--9:16}.
\newblock
\urldef\tempurl \url{https://doi.org/10.4230/LIPIcs.CSL.2020.9}
\showDOI{\tempurl}


\bibitem[Beimel et~al\mbox{.}(2000)]{BeimelBBKV00}
\bibfield{author}{\bibinfo{person}{Amos Beimel}, \bibinfo{person}{Francesco
  Bergadano}, \bibinfo{person}{Nader~H. Bshouty}, \bibinfo{person}{Eyal
  Kushilevitz}, {and} \bibinfo{person}{Stefano Varricchio}.}
  \bibinfo{year}{2000}\natexlab{}.
\newblock \showarticletitle{Learning functions represented as multiplicity
  automata}.
\newblock \bibinfo{journal}{\emph{J. {ACM}}} \bibinfo{volume}{47},
  \bibinfo{number}{3} (\bibinfo{year}{2000}), \bibinfo{pages}{506--530}.
\newblock
\urldef\tempurl \url{https://doi.org/10.1145/337244.337257}
\showDOI{\tempurl}


\bibitem[Berstel and Reutenauer(1988)]{BerstelR88}
\bibfield{author}{\bibinfo{person}{Jean Berstel} {and}
  \bibinfo{person}{Christophe Reutenauer}.} \bibinfo{year}{1988}\natexlab{}.
\newblock \bibinfo{booktitle}{\emph{Rational series and their languages}}.
  \bibinfo{series}{{EATCS} monographs on theoretical computer science},
  Vol.~\bibinfo{volume}{12}.
\newblock \bibinfo{publisher}{Springer}.
\newblock
\showISBNx{0387186263}
\urldef\tempurl \url{https://www.worldcat.org/oclc/17841475}
\showURL{\tempurl}


\bibitem[Blondel and Tsitsiklis(2000)]{BlondelT2000boundedness}
\bibfield{author}{\bibinfo{person}{Vincent~D. Blondel} {and}
  \bibinfo{person}{John~N. Tsitsiklis}.} \bibinfo{year}{2000}\natexlab{}.
\newblock \showarticletitle{The boundedness of all products of a pair of
  matrices is undecidable}.
\newblock \bibinfo{journal}{\emph{Systems \& Control Letters}}
  \bibinfo{volume}{41}, \bibinfo{number}{2} (\bibinfo{year}{2000}),
  \bibinfo{pages}{135--140}.
\newblock
\urldef\tempurl \url{https://doi.org/10.1016/S0167-6911(00)00049-9}
\showDOI{\tempurl}


\bibitem[Blondin et~al\mbox{.}(2017)]{BlondinFHH17}
\bibfield{author}{\bibinfo{person}{Michael Blondin}, \bibinfo{person}{Alain
  Finkel}, \bibinfo{person}{Christoph Haase}, {and} \bibinfo{person}{Serge
  Haddad}.} \bibinfo{year}{2017}\natexlab{}.
\newblock \showarticletitle{The Logical View on Continuous Petri Nets}.
\newblock \bibinfo{journal}{\emph{{ACM} Trans. Comput. Log.}}
  \bibinfo{volume}{18}, \bibinfo{number}{3} (\bibinfo{year}{2017}),
  \bibinfo{pages}{24:1--24:28}.
\newblock
\urldef\tempurl \url{https://doi.org/10.1145/3105908}
\showDOI{\tempurl}


\bibitem[Bojanczyk(2015)]{Bojanczyk15}
\bibfield{author}{\bibinfo{person}{Mikolaj Bojanczyk}.}
  \bibinfo{year}{2015}\natexlab{}.
\newblock \showarticletitle{Star Height via Games}. In
  \bibinfo{booktitle}{\emph{30th Annual {ACM/IEEE} Symposium on Logic in
  Computer Science, {LICS} 2015, Kyoto, Japan, July 6-10, 2015}}.
  \bibinfo{publisher}{{IEEE} Computer Society}, \bibinfo{pages}{214--219}.
\newblock
\urldef\tempurl \url{https://doi.org/10.1109/LICS.2015.29}
\showDOI{\tempurl}


\bibitem[Bumpus et~al\mbox{.}(2020)]{BumpusHKST20}
\bibfield{author}{\bibinfo{person}{Georgina Bumpus}, \bibinfo{person}{Christoph
  Haase}, \bibinfo{person}{Stefan Kiefer}, \bibinfo{person}{Paul{-}Ioan
  Stoienescu}, {and} \bibinfo{person}{Jonathan Tanner}.}
  \bibinfo{year}{2020}\natexlab{}.
\newblock \showarticletitle{On the Size of Finite Rational Matrix Semigroups}.
  In \bibinfo{booktitle}{\emph{47th International Colloquium on Automata,
  Languages, and Programming, {ICALP} 2020, July 8-11, 2020, Saarbr{\"{u}}cken,
  Germany (Virtual Conference)}} \emph{(\bibinfo{series}{LIPIcs},
  Vol.~\bibinfo{volume}{168})}, \bibfield{editor}{\bibinfo{person}{Artur
  Czumaj}, \bibinfo{person}{Anuj Dawar}, {and} \bibinfo{person}{Emanuela
  Merelli}} (Eds.). \bibinfo{publisher}{Schloss Dagstuhl - Leibniz-Zentrum
  f{\"{u}}r Informatik}, \bibinfo{pages}{133:1--133:19}.
\newblock
\urldef\tempurl \url{https://doi.org/10.4230/LIPIcs.ICALP.2020.133}
\showDOI{\tempurl}


\bibitem[Chistikov et~al\mbox{.}(2020)]{DBLP:conf/concur/ChistikovKMP20}
\bibfield{author}{\bibinfo{person}{Dmitry Chistikov}, \bibinfo{person}{Stefan
  Kiefer}, \bibinfo{person}{Andrzej~S. Murawski}, {and} \bibinfo{person}{David
  Purser}.} \bibinfo{year}{2020}\natexlab{}.
\newblock \showarticletitle{The Big-O Problem for Labelled Markov Chains and
  Weighted Automata}. In \bibinfo{booktitle}{\emph{31st International
  Conference on Concurrency Theory, {CONCUR} 2020}}
  \emph{(\bibinfo{series}{LIPIcs}, Vol.~\bibinfo{volume}{171})}.
  \bibinfo{publisher}{Schloss Dagstuhl - Leibniz-Zentrum f{\"{u}}r Informatik},
  \bibinfo{pages}{41:1--41:19}.
\newblock
\urldef\tempurl \url{https://doi.org/10.4230/LIPIcs.CONCUR.2020.41}
\showDOI{\tempurl}


\bibitem[Chistikov et~al\mbox{.}(2022)]{Chistikov21}
\bibfield{author}{\bibinfo{person}{Dmitry Chistikov}, \bibinfo{person}{Stefan
  Kiefer}, \bibinfo{person}{Andrzej~S. Murawski}, {and} \bibinfo{person}{David
  Purser}.} \bibinfo{year}{2022}\natexlab{}.
\newblock \showarticletitle{The Big-O Problem}.
\newblock \bibinfo{journal}{\emph{Log. Methods Comput. Sci.}}
  \bibinfo{volume}{18}, \bibinfo{number}{1} (\bibinfo{year}{2022}).
\newblock
\urldef\tempurl \url{https://doi.org/10.46298/lmcs-18(1:40)2022}
\showDOI{\tempurl}
\newblock
\shownote{Extended Journal version of \cite{DBLP:conf/concur/ChistikovKMP20}}.


\bibitem[Condon and Lipton(1989)]{CondonL89}
\bibfield{author}{\bibinfo{person}{Anne Condon} {and}
  \bibinfo{person}{Richard~J. Lipton}.} \bibinfo{year}{1989}\natexlab{}.
\newblock \showarticletitle{On the Complexity of Space Bounded Interactive
  Proofs (Extended Abstract)}. In \bibinfo{booktitle}{\emph{30th Annual
  Symposium on Foundations of Computer Science, Research Triangle Park, North
  Carolina, USA, 30 October - 1 November 1989}}. \bibinfo{publisher}{{IEEE}
  Computer Society}, \bibinfo{pages}{462--467}.
\newblock
\urldef\tempurl \url{https://doi.org/10.1109/SFCS.1989.63519}
\showDOI{\tempurl}


\bibitem[Czerwinski et~al\mbox{.}(2021)]{CzerwinskiLLLM21}
\bibfield{author}{\bibinfo{person}{Wojciech Czerwinski},
  \bibinfo{person}{Slawomir Lasota}, \bibinfo{person}{Ranko Lazic},
  \bibinfo{person}{J{\'{e}}r{\^{o}}me Leroux}, {and} \bibinfo{person}{Filip
  Mazowiecki}.} \bibinfo{year}{2021}\natexlab{}.
\newblock \showarticletitle{The Reachability Problem for Petri Nets Is Not
  Elementary}.
\newblock \bibinfo{journal}{\emph{J. {ACM}}} \bibinfo{volume}{68},
  \bibinfo{number}{1} (\bibinfo{year}{2021}), \bibinfo{pages}{7:1--7:28}.
\newblock
\urldef\tempurl \url{https://doi.org/10.1145/3422822}
\showDOI{\tempurl}


\bibitem[Daviaud et~al\mbox{.}(2017)]{DaviaudJRV17}
\bibfield{author}{\bibinfo{person}{Laure Daviaud},
  \bibinfo{person}{Isma{\"{e}}l Jecker}, \bibinfo{person}{Pierre{-}Alain
  Reynier}, {and} \bibinfo{person}{Didier Villevalois}.}
  \bibinfo{year}{2017}\natexlab{}.
\newblock \showarticletitle{Degree of Sequentiality of Weighted Automata}. In
  \bibinfo{booktitle}{\emph{Foundations of Software Science and Computation
  Structures - 20th International Conference, {FOSSACS} 2017, Held as Part of
  the European Joint Conferences on Theory and Practice of Software, {ETAPS}
  2017, Uppsala, Sweden, April 22-29, 2017, Proceedings}}
  \emph{(\bibinfo{series}{Lecture Notes in Computer Science},
  Vol.~\bibinfo{volume}{10203})}, \bibfield{editor}{\bibinfo{person}{Javier
  Esparza} {and} \bibinfo{person}{Andrzej~S. Murawski}} (Eds.).
  \bibinfo{pages}{215--230}.
\newblock
\urldef\tempurl \url{https://doi.org/10.1007/978-3-662-54458-7\_13}
\showDOI{\tempurl}


\bibitem[Daviaud et~al\mbox{.}(2021)]{DaviaudJLMP021}
\bibfield{author}{\bibinfo{person}{Laure Daviaud}, \bibinfo{person}{Marcin
  Jurdzinski}, \bibinfo{person}{Ranko Lazic}, \bibinfo{person}{Filip
  Mazowiecki}, \bibinfo{person}{Guillermo~A. P{\'{e}}rez}, {and}
  \bibinfo{person}{James Worrell}.} \bibinfo{year}{2021}\natexlab{}.
\newblock \showarticletitle{When are emptiness and containment decidable for
  probabilistic automata?}
\newblock \bibinfo{journal}{\emph{J. Comput. Syst. Sci.}}
  \bibinfo{volume}{119} (\bibinfo{year}{2021}), \bibinfo{pages}{78--96}.
\newblock
\urldef\tempurl \url{https://doi.org/10.1016/j.jcss.2021.01.006}
\showDOI{\tempurl}


\bibitem[Droste et~al\mbox{.}(2009)]{droste2009handbook}
\bibfield{author}{\bibinfo{person}{Manfred Droste}, \bibinfo{person}{Werner
  Kuich}, {and} \bibinfo{person}{Heiko Vogler}.}
  \bibinfo{year}{2009}\natexlab{}.
\newblock \bibinfo{booktitle}{\emph{Handbook of weighted automata}}.
\newblock \bibinfo{publisher}{Springer Science \& Business Media}.
\newblock


\bibitem[Fijalkow et~al\mbox{.}(2015)]{DBLP:journals/corr/FijalkowGKO15}
\bibfield{author}{\bibinfo{person}{Nathana{\"{e}}l Fijalkow},
  \bibinfo{person}{Hugo Gimbert}, \bibinfo{person}{Edon Kelmendi}, {and}
  \bibinfo{person}{Youssouf Oualhadj}.} \bibinfo{year}{2015}\natexlab{}.
\newblock \showarticletitle{Deciding the value 1 problem for probabilistic
  leaktight automata}.
\newblock \bibinfo{journal}{\emph{Logical Methods in Computer Science}}
  \bibinfo{volume}{11}, \bibinfo{number}{2} (\bibinfo{year}{2015}).
\newblock
\urldef\tempurl \url{https://doi.org/10.2168/LMCS-11(2:12)2015}
\showDOI{\tempurl}


\bibitem[Fijalkow et~al\mbox{.}(2017)]{DBLP:conf/concur/FijalkowR017}
\bibfield{author}{\bibinfo{person}{Nathana{\"{e}}l Fijalkow},
  \bibinfo{person}{Cristian Riveros}, {and} \bibinfo{person}{James Worrell}.}
  \bibinfo{year}{2017}\natexlab{}.
\newblock \showarticletitle{Probabilistic Automata of Bounded Ambiguity}. In
  \bibinfo{booktitle}{\emph{28th International Conference on Concurrency
  Theory, {CONCUR} 2017, September 5-8, 2017, Berlin, Germany}}
  \emph{(\bibinfo{series}{LIPIcs}, Vol.~\bibinfo{volume}{85})},
  \bibfield{editor}{\bibinfo{person}{Roland Meyer} {and} \bibinfo{person}{Uwe
  Nestmann}} (Eds.). \bibinfo{publisher}{Schloss Dagstuhl - Leibniz-Zentrum
  f{\"{u}}r Informatik}, \bibinfo{pages}{19:1--19:14}.
\newblock
\urldef\tempurl \url{https://doi.org/10.4230/LIPIcs.CONCUR.2017.19}
\showDOI{\tempurl}


\bibitem[Filiot et~al\mbox{.}(2019)]{FiliotMR19}
\bibfield{author}{\bibinfo{person}{Emmanuel Filiot}, \bibinfo{person}{Nicolas
  Mazzocchi}, {and} \bibinfo{person}{Jean{-}Fran{\c{c}}ois Raskin}.}
  \bibinfo{year}{2019}\natexlab{}.
\newblock \showarticletitle{Decidable weighted expressions with Presburger
  combinators}.
\newblock \bibinfo{journal}{\emph{J. Comput. Syst. Sci.}}
  \bibinfo{volume}{106} (\bibinfo{year}{2019}), \bibinfo{pages}{1--22}.
\newblock
\urldef\tempurl \url{https://doi.org/10.1016/j.jcss.2019.05.005}
\showDOI{\tempurl}


\bibitem[Fliess(1974)]{fliess1974matrices}
\bibfield{author}{\bibinfo{person}{Michel Fliess}.}
  \bibinfo{year}{1974}\natexlab{}.
\newblock \showarticletitle{Matrices de hankel}.
\newblock \bibinfo{journal}{\emph{J. Math. Pures Appl}} \bibinfo{volume}{53},
  \bibinfo{number}{9} (\bibinfo{year}{1974}), \bibinfo{pages}{197--222}.
\newblock


\bibitem[Fraca and Haddad(2015)]{FracaH15}
\bibfield{author}{\bibinfo{person}{Est{\'{\i}}baliz Fraca} {and}
  \bibinfo{person}{Serge Haddad}.} \bibinfo{year}{2015}\natexlab{}.
\newblock \showarticletitle{Complexity Analysis of Continuous Petri Nets}.
\newblock \bibinfo{journal}{\emph{Fundam. Informaticae}} \bibinfo{volume}{137},
  \bibinfo{number}{1} (\bibinfo{year}{2015}), \bibinfo{pages}{1--28}.
\newblock
\urldef\tempurl \url{https://doi.org/10.3233/FI-2015-1168}
\showDOI{\tempurl}


\bibitem[Gimbert and Oualhadj(2010)]{GimbertO10}
\bibfield{author}{\bibinfo{person}{Hugo Gimbert} {and}
  \bibinfo{person}{Youssouf Oualhadj}.} \bibinfo{year}{2010}\natexlab{}.
\newblock \showarticletitle{Probabilistic Automata on Finite Words: Decidable
  and Undecidable Problems}. In \bibinfo{booktitle}{\emph{Automata, Languages
  and Programming, 37th International Colloquium, {ICALP} 2010, Bordeaux,
  France, July 6-10, 2010, Proceedings, Part {II}}}
  \emph{(\bibinfo{series}{Lecture Notes in Computer Science},
  Vol.~\bibinfo{volume}{6199})}, \bibfield{editor}{\bibinfo{person}{Samson
  Abramsky}, \bibinfo{person}{Cyril Gavoille}, \bibinfo{person}{Claude
  Kirchner}, \bibinfo{person}{Friedhelm~Meyer auf~der Heide}, {and}
  \bibinfo{person}{Paul~G. Spirakis}} (Eds.). \bibinfo{publisher}{Springer},
  \bibinfo{pages}{527--538}.
\newblock
\urldef\tempurl \url{https://doi.org/10.1007/978-3-642-14162-1\_44}
\showDOI{\tempurl}


\bibitem[Grigoriev(1988)]{DBLP:journals/jsc/Grigorev88}
\bibfield{author}{\bibinfo{person}{Dima Grigoriev}.}
  \bibinfo{year}{1988}\natexlab{}.
\newblock \showarticletitle{Complexity of Deciding Tarski Algebra}.
\newblock \bibinfo{journal}{\emph{J. Symb. Comput.}} \bibinfo{volume}{5},
  \bibinfo{number}{1/2} (\bibinfo{year}{1988}), \bibinfo{pages}{65--108}.
\newblock


\bibitem[Gödel(1931)]{Goedel1931}
\bibfield{author}{\bibinfo{person}{Kurt Gödel}.}
  \bibinfo{year}{1931}\natexlab{}.
\newblock \showarticletitle{Über formal unentscheidbare Sätze der Principia
  Mathematica und verwandter Systeme}.
\newblock \bibinfo{journal}{\emph{Monatshefte für Mathematik und Physik}}
  \bibinfo{volume}{38}, \bibinfo{number}{1} (\bibinfo{year}{1931}),
  \bibinfo{pages}{173--198}.
\newblock


\bibitem[Haase and Halfon(2014)]{HaaseH14}
\bibfield{author}{\bibinfo{person}{Christoph Haase} {and}
  \bibinfo{person}{Simon Halfon}.} \bibinfo{year}{2014}\natexlab{}.
\newblock \showarticletitle{Integer Vector Addition Systems with States}. In
  \bibinfo{booktitle}{\emph{Reachability Problems - 8th International Workshop,
  {RP} 2014, Oxford, UK, September 22-24, 2014. Proceedings}}.
  \bibinfo{pages}{112--124}.
\newblock
\urldef\tempurl \url{https://doi.org/10.1007/978-3-319-11439-2\_9}
\showDOI{\tempurl}


\bibitem[Hashiguchi(1988)]{Hashiguchi88}
\bibfield{author}{\bibinfo{person}{Kosaburo Hashiguchi}.}
  \bibinfo{year}{1988}\natexlab{}.
\newblock \showarticletitle{Algorithms for Determining Relative Star Height and
  Star Height}.
\newblock \bibinfo{journal}{\emph{Inf. Comput.}} \bibinfo{volume}{78},
  \bibinfo{number}{2} (\bibinfo{year}{1988}), \bibinfo{pages}{124--169}.
\newblock
\urldef\tempurl \url{https://doi.org/10.1016/0890-5401(88)90033-8}
\showDOI{\tempurl}


\bibitem[Macintyre and Wilkie(1996)]{schanuelsconj}
\bibfield{author}{\bibinfo{person}{Angus Macintyre} {and}
  \bibinfo{person}{Alex~J. Wilkie}.} \bibinfo{year}{1996}\natexlab{}.
\newblock \showarticletitle{On the decidability of the real exponential field}.
\newblock \bibinfo{journal}{\emph{Kreiseliana. About and Around Georg Kreisel}}
  (\bibinfo{year}{1996}), \bibinfo{pages}{441--467}.
\newblock


\bibitem[Mandel and Simon(1977)]{MandelS77}
\bibfield{author}{\bibinfo{person}{Arnaldo Mandel} {and} \bibinfo{person}{Imre
  Simon}.} \bibinfo{year}{1977}\natexlab{}.
\newblock \showarticletitle{On Finite Semigroups of Matrices}.
\newblock \bibinfo{journal}{\emph{Theor. Comput. Sci.}} \bibinfo{volume}{5},
  \bibinfo{number}{2} (\bibinfo{year}{1977}), \bibinfo{pages}{101--111}.
\newblock
\urldef\tempurl \url{https://doi.org/10.1016/0304-3975(77)90001-9}
\showDOI{\tempurl}


\bibitem[Mazowiecki and Riveros(2015)]{MazowieckiR15}
\bibfield{author}{\bibinfo{person}{Filip Mazowiecki} {and}
  \bibinfo{person}{Cristian Riveros}.} \bibinfo{year}{2015}\natexlab{}.
\newblock \showarticletitle{Maximal Partition Logic: Towards a Logical
  Characterization of Copyless Cost Register Automata}. In
  \bibinfo{booktitle}{\emph{24th {EACSL} Annual Conference on Computer Science
  Logic, {CSL} 2015, September 7-10, 2015, Berlin, Germany}}
  \emph{(\bibinfo{series}{LIPIcs}, Vol.~\bibinfo{volume}{41})},
  \bibfield{editor}{\bibinfo{person}{Stephan Kreutzer}} (Ed.).
  \bibinfo{publisher}{Schloss Dagstuhl - Leibniz-Zentrum f{\"{u}}r Informatik},
  \bibinfo{pages}{144--159}.
\newblock
\urldef\tempurl \url{https://doi.org/10.4230/LIPIcs.CSL.2015.144}
\showDOI{\tempurl}


\bibitem[Mazowiecki and Riveros(2018)]{MazowieckiR18}
\bibfield{author}{\bibinfo{person}{Filip Mazowiecki} {and}
  \bibinfo{person}{Cristian Riveros}.} \bibinfo{year}{2018}\natexlab{}.
\newblock \showarticletitle{Pumping Lemmas for Weighted Automata}. In
  \bibinfo{booktitle}{\emph{35th Symposium on Theoretical Aspects of Computer
  Science, {STACS} 2018, February 28 to March 3, 2018, Caen, France}}
  \emph{(\bibinfo{series}{LIPIcs}, Vol.~\bibinfo{volume}{96})},
  \bibfield{editor}{\bibinfo{person}{Rolf Niedermeier} {and}
  \bibinfo{person}{Brigitte Vall{\'{e}}e}} (Eds.). \bibinfo{publisher}{Schloss
  Dagstuhl - Leibniz-Zentrum f{\"{u}}r Informatik},
  \bibinfo{pages}{50:1--50:14}.
\newblock
\urldef\tempurl \url{https://doi.org/10.4230/LIPIcs.STACS.2018.50}
\showDOI{\tempurl}


\bibitem[Mazowiecki and Riveros(2019)]{MazowieckiR19}
\bibfield{author}{\bibinfo{person}{Filip Mazowiecki} {and}
  \bibinfo{person}{Cristian Riveros}.} \bibinfo{year}{2019}\natexlab{}.
\newblock \showarticletitle{Copyless cost-register automata: Structure,
  expressiveness, and closure properties}.
\newblock \bibinfo{journal}{\emph{J. Comput. Syst. Sci.}}
  \bibinfo{volume}{100} (\bibinfo{year}{2019}), \bibinfo{pages}{1--29}.
\newblock
\urldef\tempurl \url{https://doi.org/10.1016/j.jcss.2018.07.002}
\showDOI{\tempurl}


\bibitem[Minsky(1967)]{minsky1967computation}
\bibfield{author}{\bibinfo{person}{Marvin~Lee Minsky}.}
  \bibinfo{year}{1967}\natexlab{}.
\newblock \bibinfo{booktitle}{\emph{Computation}}.
\newblock \bibinfo{publisher}{Prentice-Hall Englewood Cliffs}.
\newblock


\bibitem[Paz(1971)]{paz71}
\bibfield{author}{\bibinfo{person}{Azaria Paz}.}
  \bibinfo{year}{1971}\natexlab{}.
\newblock \bibinfo{booktitle}{\emph{Introduction to probabilistic automata}}.
\newblock \bibinfo{publisher}{Academic Press}.
\newblock


\bibitem[Schmidt(1980)]{Schmidt1980}
\bibfield{author}{\bibinfo{person}{Wolfgang~M. Schmidt}.}
  \bibinfo{year}{1980}\natexlab{}.
\newblock \bibinfo{booktitle}{\emph{Simultaneous Approximation. In: Diophantine
  Approximation}}.
\newblock \bibinfo{publisher}{Springer Berlin Heidelberg},
  \bibinfo{address}{Berlin, Heidelberg}, \bibinfo{pages}{27--47}.
\newblock
\showISBNx{978-3-540-38645-2}
\urldef\tempurl \url{https://doi.org/10.1007/978-3-540-38645-2_2}
\showDOI{\tempurl}


\bibitem[Sch{\"{u}}tzenberger(1961)]{Schutzenberger61b}
\bibfield{author}{\bibinfo{person}{Marcel~Paul Sch{\"{u}}tzenberger}.}
  \bibinfo{year}{1961}\natexlab{}.
\newblock \showarticletitle{On the Definition of a Family of Automata}.
\newblock \bibinfo{journal}{\emph{Information and Control}}
  \bibinfo{volume}{4}, \bibinfo{number}{2-3} (\bibinfo{year}{1961}),
  \bibinfo{pages}{245--270}.
\newblock
\urldef\tempurl \url{https://doi.org/10.1016/S0019-9958(61)80020-X}
\showDOI{\tempurl}


\bibitem[Simon(1990)]{Simon90}
\bibfield{author}{\bibinfo{person}{Imre Simon}.}
  \bibinfo{year}{1990}\natexlab{}.
\newblock \showarticletitle{Factorization Forests of Finite Height}.
\newblock \bibinfo{journal}{\emph{Theor. Comput. Sci.}} \bibinfo{volume}{72},
  \bibinfo{number}{1} (\bibinfo{year}{1990}), \bibinfo{pages}{65--94}.
\newblock
\urldef\tempurl \url{https://doi.org/10.1016/0304-3975(90)90047-L}
\showDOI{\tempurl}


\bibitem[Simon(1994)]{Simon94}
\bibfield{author}{\bibinfo{person}{Imre Simon}.}
  \bibinfo{year}{1994}\natexlab{}.
\newblock \showarticletitle{On Semigroups of Matrices over the Tropical
  Semiring}.
\newblock \bibinfo{journal}{\emph{{RAIRO} Theor. Informatics Appl.}}
  \bibinfo{volume}{28}, \bibinfo{number}{3-4} (\bibinfo{year}{1994}),
  \bibinfo{pages}{277--294}.
\newblock
\urldef\tempurl \url{https://doi.org/10.1051/ita/1994283-402771}
\showDOI{\tempurl}


\bibitem[Turakainen(1969)]{turakainen1969generalized}
\bibfield{author}{\bibinfo{person}{Paavo Turakainen}.}
  \bibinfo{year}{1969}\natexlab{}.
\newblock \showarticletitle{Generalized automata and stochastic languages}.
\newblock \bibinfo{journal}{\emph{Proc. Amer. Math. Soc.}}
  \bibinfo{volume}{21}, \bibinfo{number}{2} (\bibinfo{year}{1969}),
  \bibinfo{pages}{303--309}.
\newblock


\bibitem[Waldschmidt(2000)]{Waldschmidt00}
\bibfield{author}{\bibinfo{person}{Michel Waldschmidt}.}
  \bibinfo{year}{2000}\natexlab{}.
\newblock \bibinfo{booktitle}{\emph{Diophantine Approximation on Linear
  Algebraic Groups}}. \bibinfo{series}{Grundlehren der mathematischen
  Wissenschaften (A Series of Comprehensive Studies in Mathematics)},
  Vol.~\bibinfo{volume}{326}.
\newblock \bibinfo{publisher}{Springer}, \bibinfo{address}{Berlin, Heidelberg}.
\newblock


\bibitem[Weber and Seidl(1991)]{WeberS91}
\bibfield{author}{\bibinfo{person}{Andreas Weber} {and} \bibinfo{person}{Helmut
  Seidl}.} \bibinfo{year}{1991}\natexlab{}.
\newblock \showarticletitle{On the Degree of Ambiguity of Finite Automata}.
\newblock \bibinfo{journal}{\emph{Theor. Comput. Sci.}} \bibinfo{volume}{88},
  \bibinfo{number}{2} (\bibinfo{year}{1991}), \bibinfo{pages}{325--349}.
\newblock
\urldef\tempurl \url{https://doi.org/10.1016/0304-3975(91)90381-B}
\showDOI{\tempurl}


\end{thebibliography}

\newpage
\onecolumn

\newgeometry{margin=1.5in
     }\pagestyle{plain}
\appendix

\section{Additional material for Section~\ref{sec:preliminaries}} \label{app:preliminaries}

\begin{remark*}\label{remark:affine}
Originally, linear CRA are defined with linear updates (rather than affine)---giving rise to the name linear CRA.
Formally, to simulate affine updates with linear updates we extend the set of registers $\X$ with one additional register $y$.
This register is: initialised to one, \ie $I(y) = \one$; updated trivially, \ie if $\delta(q,a) = (q',\sigma)$ then $\sigma(y) = y$; and does not contribute to the output, \ie $F(q,y) = \zero$
for all $a \in \Sigma$ and $q,q'\in Q$.
Then every constant $c$ in affine expressions is replaced by $c \odot y$ (making it linear).
\end{remark*}

\section{Additional material for Section~\ref{sec:boundedness}} \label{app:boundedness}

Before proving the main result (\cref{theorem:boundedness_rational}), let us introduce some additional claims and prove the claims from \cref{sec:boundedness}.

We define $\valpref(\rho) = I(q_0) \cdot \val_\leftrightarrow(\rho)$.
Thus using those notations $\val(\rho) = \valpref(\rho) \cdot F(q_n)$.
We say that $\rho$ is an \emph{active run} if $\valpref(\rho) \neq \zero$.

\begin{remark}\label{remark:active}
Given a word $w$ consider the vector $\bv = I \cdot M_w$.
Notice that for every $q$ the value $\bv[q]$ is equal to the sum $\sum_\rho\valpref(\rho)$,  
ranging over all active runs $\rho$ over $w$ that end in state $q$.
\end{remark}

\begin{claim}\label{claim:monotonicity_active}
Let $\Aa = (\Sigma, I, F, (M_a)_{a\in \Sigma})$ be a WA over $\Qpos$.
Consider two words $w,w' \in \Sigma^*$ and suppose that $I \cdot M_w \le I \cdot M_{w'}$.
Then $\Aa(ww'') \le \Aa(w'w'')$ for every word $w''\in \Sigma^*$.
\end{claim}

\begin{claimproof}
It suffices to observe that $\Aa(wv) = I \cdot M_w \cdot M_v \cdot F$ and $\Aa(w'v) = I \cdot M_{w'} \cdot M_v \cdot F$.
The proof follows from $I \cdot M_w \le I \cdot M_{w'}$ and that the values in $\Aa$ are selected over $\Qpos$ and are thus nonnegative.
\end{claimproof}

\LemSimpleLinearAmb*

\begin{proof}
Let $\Aa = (\Sigma, Q, q_0, I, F, \X, \delta)$ be a copyless linear CRA.
We assume that all of its transitions are accessible (non-accessible transitions of $\Aa$ can be removed).
We construct the \ssla WA $\Bb= (\Sigma_\Bb, I_\Bb, F_\Bb, (M_b)_{b\in \Sigma_\Bb})$
where the alphabet $\Sigma_\Bb \subseteq \Sigma \times Q \times Q$ 
comprises the triples $(a,q,q')$ for which $\delta(q,a) = (q',\sigma)$ for some affine map
$\sigma \colon \Q^{|\X|} \to \Q$ (with $|\X|$ the number of registers in $\X$),
the states are $Q_\Bb = (\mathcal{X} \times Q) \cup \set{p}$
and the entries of the matrices $M_b$ satisfy (unlisted entries are set to $0$):
\begin{itemize}[leftmargin=*]
\item $p\xrightarrow{a|1} p$ for all $a\in\Sigma_\Bb$.
\item for every transition $\delta(q,a) = (q', \sigma_{q,a})$, denoting $\sigma_{q,a}(x) = \sum_{y \in \mathcal{X}} s_{x,y}\cdot y + c_x$ for every $x \in \X$ we have:
\begin{itemize}
\item $(y, q) \xrightarrow{(a,q,q')|s_{x,y}}(x, q')$ for every $y \in \mathcal{X}$,
\item $p \xrightarrow{(a,q,q')|c_{x}} (x, q')$.
\end{itemize}
\end{itemize}
Finally, the initial weights $I_\Bb$ are such that $I_\Bb((x_i,q_0))=I(x_i)$, $I_{\Bb}(p)= 1$, and all other weights are set to $0$.
Similarly, the final weights $F_\Bb$ are such that $F_\Bb((x, q)) = F(q,x)$ and $F_{\Bb}(p) = 0$.
That $\Bb$ is a \ssla WA follows easily from the assumption that $\Aa$ is copyless.
Observe that $\Bb$ restricted to $Q_{\Bb}\setminus \{p\}$ is deterministic: there can only be one successor state of $(y,q)$ for any character. Indeed, on any given character, the state transitions of $\Aa$ are deterministic, so there is only one successor state, say $q'$. Register updates are copyless, so there is only one register, say $x$, `receiving' the value of $y$. The only successor state then is $(x,q')$.

It remains to prove that
$
\sup_{w\in\Sigma^*}\! \Aa(w) = \sup_{\mathfrak{w}\in\Sigma_\Bb^*} \!\Bb(\mathfrak{w}).
$

We first show that $\sup_{w\in\Sigma^*}\Aa(w) \le \sup_{\mathfrak{w}\in\Sigma_\Bb^*}\Bb(\mathfrak{w})$.
Fix $w=w_1\cdots w_k \in \Sigma^*$ and let \\$(q_0,\sigma_0),\ldots, (q_{|w|}, \sigma_{|w|})$ be the unique corresponding run on $\Aa$.
By a simple induction we have
\begin{align}\label{eq:a=b}
\Aa(w)  = \Bb((w_1,q_0,q_1)(w_2,q_1,q_2)\cdots(w_{k},q_{k-1},q_{k})),
\end{align}
from which the claim follows.

However, for the converse inequality, there are input words for $\Bb$ that do not correspond to 
a valid run on $\Aa$.
Such unfaithful runs will not over-approximate $\sup_{w\in\Sigma_\Aa^*}\Aa(w)$. 
Thus for every
\[\mathfrak{w} = (a_1,p_1,p_1')(a_2,p_2,p_2')\cdots(a_{k},p_{k},p_{k}') \in \Sigma_\Bb^*\]
we show that $\Bb(\mathfrak{w}) \le \sup_{w \in \Sigma_\Aa^*}\! \Aa(w)$.

We consider three cases depending on the input word $\mathfrak{w}$.
First, suppose that $p_1 = q_0$ and $p_{i-1}' = p_i$ for all $i \in \set{2,\ldots,k}$.
Then defining $w = a_1 \cdots a_{k}$ a simple induction as in \Cref{eq:a=b} shows that $\Aa(w) = \Bb(\mathfrak{w})$.

In the second case suppose that $p_1 \neq q_0$. Since $(a_1,p_1,p_1')\in \Sigma_{\Bb}$ the transition $p_1\trans{a_1} p_1'$ exists in $\Aa$. Now there
exists a word $u = u_0u_1\dots u_\ell$ such that $q_0 \trans{u} p_1$ in $\Aa$. This is because we assume all states of $\Aa$ are accessible. Let
$\mathfrak{u} = (u_0,q_0,q_1)(u_1,q_1,q_2)\dots(u_\ell,q_\ell,q_{\ell+1})$, with $q_{\ell+1} = p_1$, be the word simulating the run. 
By definition of $\Bb$ it is straightforward to verify that $I_{\Bb}\cdot M_{(a_1,p_1,p_1')} \le I_{\Bb} \cdot M_{\mathfrak{u} (a_1,p_1,p_1')}$: indeed, $I_{\Bb}[(x,p_1)] = 0$ for all $x\in \X$, so the only positive contributions to
$I_{\Bb}\cdot M_{(a_1,p_1,p_1')}[(x,p_1')]$ come from transitions of the form
$p \trans{a_1} p_1'$. Such contributions are also present when reading $ M_{\mathfrak{u} (a_1,p_1,p_1')}$, thus the inequality.
We conclude by \cref{claim:monotonicity_active}: $\Bb(\mathfrak{w}) \le \Bb(\mathfrak{u}\mathfrak{w})$.
This case therefore reduces to one of the other two cases.

For the last case let $i$ be the largest index such that $p_{i-1}' \neq p_i$.
We aim to replace the prefix that that does not simulate the original automaton with another one that does.
Notice that after reading $(a_1,p_1,p_1') \cdots (a_{i-1},p_{i-1},p_{i-1}')$ all active runs in $\Bb$ end in states of the form $(x,p_{i-1}')$ or $p$.
This is because, when reading a letter of the form $(\cdot,\cdot,p_{i-1}')$,
all transitions with nonzero values go to either $p$ (from $p$) or to states of the form 
$(\cdot,p_{i-1}')$ and $p$.
Since $p_{i-1}' \neq p_i$ after reading $(a_i,p_{i},p_i')$ all active runs that previously ended in $(x,p_{i-1}')$ become runs
of value zero by construction.
The only positive weights in the system will come either from transitions $p \xrightarrow{(a_i,p_{i},p_i')|c_{x}} (x,p_i')$ or the $p \xrightarrow{a|1} p$ transition.
The former only occurs if $\delta(p_i,a_i) = (p_i',\sigma)$ with $\sigma(x) = \sum_{y\in \mathcal{X}} s_{x,y}\cdot y + c_{x}$ and $c_{x} \neq 0$.
The value of such an active run is $c_{x}$ (recall that $p$ loops with value $1$).

Since all transitions in $\Aa$ are assumed accessible, there exists a word $w' = b_1\cdots b_{\ell} \in \Sigma_{\Aa}^*$ such that
$(q_0,I)
	\xrightarrow{w'}
		(q,\sigma)
			\xrightarrow{a_i}
			(p_i,\sigma')$ for some valuations $\sigma$ and $\sigma'$.
Let now $\mathfrak{w}' = (b_1,q_0,q_1)\cdots (b_{\ell},q_{\ell},p_i) \in \Sigma_{\Bb}^*$, with $q_{\ell} = q$, be the word which simulates $\Aa$ on $b_1\cdots b_{\ell}$. We observe that, after reading the word
$\mathfrak{w}'\cdot (a_{i},p_{i},p_{i}')$, the automaton $\Bb$ also has active runs of values $c_{x}$ that end in states $(x,p_i')$ corresponding to the aforementioned runs in $\Aa$.
Thus by \cref{remark:active} and \cref{claim:monotonicity_active}
$\Bb(\mathfrak{w}) \le \Bb(\mathfrak{w}'(a_{i},p_{i},p_{i}')\ldots (a_{k},p_{k},p_{k}'))$.
Notice now that 
$w'$ is
such that $\Bb(\mathfrak{w}') = \Aa(w')$ as in \cref{eq:a=b}.
Then, by the choice of $i$,
\begin{equation*}
\Bb(\mathfrak{w}) \le \Bb(\mathfrak{w}'(a_{i},p_{i},p_{i}')\cdots (a_{k},p_{k},p_{k}'))= \Aa(w' a_{i} \ldots a_{k}),
\end{equation*}
which concludes the proof.
\end{proof}

\begin{restatable}{claim}{ClaimUpper}\label{claim:upper}
The admissible weights are upper bounded by a computable $a \in \Q$.
\end{restatable}

\begin{claimproof}
Consider a run $\rho$ over $w$ such that $\valinf(\rho) = r > 0$.
If $|w| > |Q|$ then there is a subrun $\rho'$ of $\rho$ such that $\rho'$ is a cycle.
Since $\Bb$ is as in \cref{lemma:boundednesslemma} we have $\valinf(\rho') \le 1$.
Thus by removing $\rho'$ from $\rho$, we do not decrease the value of the run.
We conclude since the number of words of length at most $|Q|$
is finite, whence $a$ can be chosen to be the maximum of values
of runs labelled by such words.
\end{claimproof}

\ClaimLower*

\begin{claimproof}
We have $s_q, e_q \le 1$ since the value of runs over the empty word is $1$. 
The proof follows by the same arguments as the proof of Claim~\ref{claim:upper}.
\end{claimproof}

\ClaimFinitely*

\begin{claimproof}
Let $a$ be the upper bound of admissible weights as in \cref{claim:upper}.
Consider a run $\rho$ over $w$ such that $\valinf(\rho) = r$.
We can assume that $\valinf(\rho') < 1$ for all subruns $\rho'$ of $\rho$ that are cycles.
Indeed, recall that $\valinf(\rho')\le 1$ by the assumptions on $\Bb$ in \cref{lemma:boundednesslemma}. Thus the condition can be violated only by cycles such that $\valinf(\rho') = 1$. However, note that removing cycles of value $1$ does not change the value of the run, thus we can assume that they do not appear.

Since there are finitely many simple cycles we can define $\alpha < 1$ to be larger than the value of any simple cycle among those of value smaller than $1$.
Consider $T$ such that $\alpha^T a < x$.
Then if there are more than $T$ disjoint cycles in 
$\rho$, $\valinf(\rho) < x$. 
Therefore, $r \ge x$ for runs of bounded length, hence for finitely many runs.
\end{claimproof}

\ClaimEps*

\begin{claimproof}
The second item follows immediately from the definition.
For the first item recall from \cref{eq:matrix} that $\overline{M}_w(q,q') = \varepsilon$ if and only if $M_{w}(q,q') < e_q \cdot s_{q'}$.
We have $\rununique{q,w,q'}$ is a subrun of $\rho$ if and only if $\valinf(\rho) = \valinf(\rho') \cdot \valinf\rununique{q,w,q'} \cdot \valinf(\rho'')$ for some runs: $\rho'$ ending in $q$, and $\rho''$ starting in $q'$.
By definition $\valinf(\rho') \cdot \valinf(\rho'') \le \frac{1}{e_q \cdot s_{q'}}$.
We conclude since $\valinf\rununique{q,w,q'} = M_{w}(q,q')$.
\end{claimproof}

\ClaimMonoid*

\begin{claimproof}
The claim that the set is finite follows from \cref{claim:finitely}.
Indeed, if $M_w[q,q'] \not \in \set{0,\varepsilon}$ then $q,q' \neq p$, $\valinf\rununique{q,w,q'} \ge e_{q} \cdot s_{q'}$ and there are finitely many of such admissible weights.
We prove that $\overline{M}_{wu}[q,q'] = \overline{M}_w \otimes \overline{M}_u [q,q']$ for every $w,u\in \Sigma^*$ and $q,q' \in Q$.
Notice that if $q = p$ or $q' = p$ then this is trivial, as in that case only $0$ and $\varepsilon$ can occur in the matrix entries.
Thus we assume that $q,q' \neq p$.

Recall that $M_{wu} = M_w M_u$. By construction,
$\overline{M}[q,q'] = 0$ if and only if $M(q,q') = 0$ for any
matrix $M$. Therefore it is clear that
$\overline{M}_{wu}[q,q'] = 0$ if and only if $M_{wu}[q,q'] = 0$
if and only if $\overline{M}_w \overline{M}_u[q,q'] = 0$ if and
only if $\overline{M}_w \otimes \overline{M}_u [q,q'] = 0$. Thus
we may assume that $\overline{M}_{wu}[q,q'] > 0$ and
$\overline{M}_w  \otimes \overline{M}_u [q,q'] > 0$. Since the
states restricted to $Q \setminus \set{p}$ are deterministic
there is a unique state $q''$ such that
$\overline{M}_w [q,q''] > 0$, $\overline{M}_u[q'',q'] > 0$, and  $\overline{M}_w\overline{M}_u[q,q'] = \overline{M}_w [q,q''] \cdot \overline{M}_u[q'',q']$.

Suppose first that $\overline{M}_{wu} [q,q'] = \varepsilon$.
If either $\overline{M}_w [q,q''] = \varepsilon$ or $\overline{M}_u[q'',q'] = \varepsilon$ then we are done.
Suppose otherwise, then $\overline{M}_w [q,q''] = M_{w}[q,q'']$ and $\overline{M}_u[q'',q'] = M_u[q'',q']$ and so
\begin{equation*}
0 < \overline{M}_w [q,q''] \cdot  \overline{M}_u[q'',q'] = M_{wu}[q,q'] < e_{q} \cdot s_{q'}.
\end{equation*}
Thus $\overline{M}_w \otimes \overline{M}_u [q,q'] = \varepsilon$ as desired.

Conversely, suppose that $\overline{M}_w \otimes \overline{M}_u [q,q'] = \varepsilon$. It follows that at least one of $\overline{M}_w [q,q'']$ and $\overline{M}_u[q'',q']$ equals $\varepsilon$, or $\varepsilon < \overline{M}_w [q,q''] \cdot \overline{M}_u[q'',q'] < e_{q} \cdot s_{q'}$. Assuming first the latter, we have
\[
0 < \overline{M}_w [q,q''] \cdot \overline{M}_u[q'',q'] = M_{wu} [q,q'] < e_{q} \cdot s_{q'}
\]
and thus $\overline{M}_{wu} [q,q'] = \varepsilon$ as desired.

Assume then that $\overline{M}_w [q,q''] = \varepsilon$
(resp., $\overline{M}_u[q'',q'] = \varepsilon$).
By \cref{claim:epsilon}, all runs $\rho$ which contain  $\rununique{q,w,q''}$ (resp., $\rununique{q'',u,q'}$) as a subrun have $\valinf(\rho) < 1$.
In particular, every run $\rho'$ containing $\rununique{q,wu,q'}$ has $\rununique{q,w,q''}$ and $\rununique{q'',u,q'}$ as subruns, and so $\valinf(\rho') < 1$. Consequently, applying \cref{claim:epsilon} again we get $\overline{M}_{wu} [q,q'] = \varepsilon$. 

Finally, if both $\overline{M}_{wu} [q,q'] > \varepsilon$ and $\overline{M}_w \otimes \overline{M}_u [q,q'] > \varepsilon$ then $\overline{M}_w [q,q''] = M_w[q,q'']$ and $\overline{M}_u[q'',q'] = M_u[q'',q']$.
Thus $s_q \cdot e_{q'} \le \overline{M}_{wu}[q,q'] = \overline{M}_w [q,q''] \cdot \overline{M}_u [q'',q']  = \overline{M}_w \otimes \overline{M}_u [q,q']$.
\end{claimproof}

\ClaimIdemp*

\begin{claimproof}
First, recall that the only nonzero transitions that end in $p$ are from $p$.
Thus we define $i$ as the smallest index such that $q_i \neq p$.
We prove that $q_{i+1} = q_{i+2} \ldots = q_m$.
Suppose $i < m$, otherwise, the claim is trivial.
Since $M[q_i,q_{i+1}] > 0$ and $M$ is idempotent $M[q_{i+1},q_{i+1}] > 0$.
We conclude because the transitions are deterministic on $Q \setminus \set{p}$.
Now, since $M$ is idempotent $M[q_{i+1},q_{i+1}] = \varepsilon$ or $M[q_{i+1},q_{i+1}] = 1$.
However, in the latter case we would obtain an excluded pattern in \cref{pattern2} of \cref{lemma:boundednesslemma}.
Indeed, since $\overline{M}_{w^1\ldots w^m} = \overline{M}_{w^i} = M$ are all the same idempotents $\overline{M}_{w^1, \ldots, w^m}[p,q_{i+1}] = \varepsilon > 0$ and thus $M_{w^1, \ldots, w^m}[p,q_{i+1}] > 0$.
\end{claimproof}

We now recall and prove the key theorem of \cref{sec:boundedness}:

\ThmBoundednessRat*

\begin{proof}[Proof of \cref{theorem:boundedness_rational}]
By \cref{lemma:stupidlysimplylinearlyambiguous} we can assume that the automaton $\Bb$ is a \ssla WA.
Without loss of generality we assume that $\Bb$ is trimmed, \ie every state appears in at least one accepting run.
We claim that $\Bb$ is bounded if and only if it does not satisfy the assumptions of \cref{lemma:boundednesslemma}: either it has a $q$-cycle of weight strictly larger than $1$ (violating assumption~\ref{pattern1}), or
it has a $q$-cycle of weight $1$ over the word $u$, where $q \neq p$ is a state reachable from $p$ with $u$ (violating assumption~\ref{pattern2}).

(Only if) If either of the assumptions~\ref{pattern1} or \ref{pattern2} is violated, then it suffices to take an accepting run that contains $q$.
When violating assumption~\ref{pattern1} with a cycle of weight strictly larger than $1$, we find runs of arbitrarily large value straightforwardly.
When violating pattern~\ref{pattern2} but not \ref{pattern1}, consider the words $(u^{n+1})_{n\in \N}$: we have the runs
$p\trans{u^i} p \trans{u} q \trans{u^{n-1}} q$, $i=0,\ldots,n$, each contributing the weight of $p \trans{u} q$. Completing the runs to accepting ones, we deduce that $\Bb$ is unbounded.

(If) Let $a$ be such that $\val(\rho) \le a$ for all accepting runs $\rho$ ($a$ exists \eg because of \cref{claim:finitely}).
Without loss of generality we can assume that
$\val(\rho) \le 1$, dividing the initial values of states by $a$,
if necessary.
Similarly, we can assume that $\valinf(\rho) \ge \val(\rho)$ for all accepting runs.
This allows us to focus on $\Runs(w)$ and $\valinf(\rho)$ instead of accepting runs and their values.

We define the constant $c$ with the following property. 
For every word $w$ of length $n$ and for all $\rho \in \Runs(w)$ we have $\frac{1}{c^n} < \valinf(\rho)$.
It suffices to choose $c$ such that $\frac{1}{c}$ is smaller than all non-zero weights that occur in $\Bb$.
We can further assume that $c \ge 2$.

Fix a word $w$ of length $n$.
We define a partition $\Runs(w) = \bigcup_{i = 1}^{n} P_i$ as follows: $P_i = \{\rho \mid \frac{1}{c^{i}} < \valinf(\rho) \le \frac{1}{c^{i-1}}\}$ for $i \in \set{1,\ldots,n}$.
Notice that $P_i \subseteq \Runs[\frac{1}{c^i}](w)$.
Since $c^i \ge 2$ by \cref{lemma:boundednesslemma} we have $|P_i| \le \poly(\log c^i ) = \poly(i\log c) =  \poly(i)$.
Let $W(i)$ be a polynomial that upper bounds $|P_i|$.
We obtain
\begin{multline*}
\Bb(w) \le \sum_{\rho \in \Runs(w)}  \valinf(\rho)  \le  \sum_{i=1}^n \frac{|P_i|}{c^{i-1}}   \le   \sum_{i=1}^n \frac{W(i)}{c^{i-1}}  \le  \sum_{i=1}^\infty \frac{W(i)}{c^{i-1}}. \end{multline*}
Notice that the series converges; we conclude since its value
does not depend on $w$.

It remains to show that deciding whether assumptions~\ref{pattern1} and~\ref{pattern2} hold can be detected.
\begin{claim}
Whether assumptions~\ref{pattern1} and~\ref{pattern2} hold can be detected in polynomial time.
\end{claim}

\begin{claimproof}
We claim both patterns can be detected by modifications of the Bellman-Ford shortest path algorithm. Assumption~\ref{pattern1} requires detecting a cycle of multiplicative weight greater than $1$. This is directly an instance of the \emph{currency-arbitrage problem} (so named for the case in which the weights represent currency exchange rates, and we ask whether infinite profit could be generated by repeatedly switching currency according to the cycle).

The problem can be solved using an adapted version of the Bellman-Ford algorithm, which can be used to detect negative-cycles in polynomial time (under summation of edge weights).
Searching for negative-cycles after taking ${-}\log$ of each transition weight detects cycles with  product strictly larger than $1$ in the original graph.
A naïve implementation of $-\log$ could lead to approximation error, but one can instead store the exponentiation of the log. In such a representation, instead of taking the sum one should take the product of the exponentiations.

We rely upon the fact that when the Bellman-Ford algorithm runs in polynomially many steps in the size of the graph (not the edge weights). A priori polynomially many multiplications can grow to doubly-exponential size. However, whenever the Bellman-Ford algorithm takes addition (multiplication in our setting), this is always with one side as a weight from the automaton. In polynomial many steps the size of the numbers remain exponential (thus represented in polynomial space).

Assumption~\ref{pattern2} requires detecting a cycle from $q$ to $q$ of (multiplicative) weight $1$ on a word $w$, for which there is a path in $w$ from $p$ to $q$ (of non-zero weight). Let $\mathcal{C}$ be a non-deterministic finite automaton which represents the language of words starting from $p$ and ending in $q$ with non-zero weight (this is found by taking the weighted automaton, and keeping all edges with non-zero weight). We now consider the product automaton $\mathcal{C}\times\mathcal{B}$, where $\mathcal{B}$ is the weighted automaton: on each edge from $(s,s')$ to $(t,t')$  we take the weight to be $-\log(v_i)$ where $v_i$ is the weight of the edge from $s'$ to $t'$ in $\mathcal{B}$, providing there is an edge from $s$ to $t$ in $\mathcal{C}$. The weight is $\infty$ if either there is no edge from $s$ to $t$ in $\mathcal{C}$ or the weight from $s'$ to $t'$ is $0$ in $\mathcal{B}$. Since we can assume patterns from assumption~\ref{pattern1} have been ruled out in the previous step, we know there is no cycle of (multiplicative) weight greater than $1$ in $\mathcal{B}$, so every path from $(p,q)$ to $(q,q')$ in $\mathcal{C}$ has weight at least  $0$ due to $-\log(1)= 0$. We run the Bellman-Ford algorithm to detect whether the shortest path from $(p,q)$ to $(q,q)$ is $0$, which is the case if and only if there is a word from $p$ to $q$ which has weight $1$ from $q$ to $q$. The procedure is repeated for each candidate $q$.
\end{claimproof}

\end{proof}

\section{Additional material for Section~\ref{0-isolation}} \label{app:zero}

\subsection{Additional notation}

For \simpleCRA $\Bb = (\Sigma, I, F, \X, \delta)$ it is easy to define an equivalent weighted automaton $\Aa = (\Sigma, I', F', (M_a)_{a\in \Sigma})$. Indeed, let $Q = \X \cup \set{y}$ for some fresh register $y$. Then $I'(x) = I(x)$, $F'(x) = F(x)$ for $x \in \X$ and $I'(y) = 1$, $F'(y) = 0$. Suppose $\sigma_a(x) = (c_{a,x}\odot x) \oplus d_{a,x}$. Then $M_a(x,x) = c_{a,x}$, $M_a(y,x) = d_{a,x}$, and $M_a(y,y)= 1$ for all $x \in \X$. The remaining entries of $M_a$ are $0$.
For example this construction applied to $\Bb$ in \cref{example:CRA} results exactly in $\Bb$ from \cref{example:weightedn}.

By summing over all accepting runs in the equivalent weighted automaton, we get for every \simpleCRA $\Bb$:
\begin{align}\label{eq:simple}
 \Bb(w) = \bigoplus_{x \in \X} F(x) \odot \Bb_x(w),
\end{align}
where
\begin{align}\label{eq:simple2}
\Bb_x(w) =  \bigoplus_{i \in \set{1,\ldots,n+1}} \left( d_{w_{i-1},x} \odot \bigodot_{j \in \set{i,\ldots,n}}c_{w_j,x}\right).
\end{align}
To ease the notation we assume that $d_{w_{0},x} = I(x)$.
For $i = n+1$ the product is empty and it evaluates to $\one$.
Assuming implicitly the translation of $\Bb$ to weighted automata, we consider accepting runs on $\Bb$.
Fix a word $w = w_1\dots w_m$ and for every $i \in [1,m]$ let $\sigma_{w_i}(x) = (c_{i,x}\odot x) \oplus d_{i,x}$. A run on $w$ to is a sequence $\rho = x_0,c_1,x_1,c_2,\dots,c_{m},x_m \in \Acc(\Bb,w)$, where:
$x_i \in \X \cup \set{y}$, and for every $i \in [1,m]$: either $x_{i-1} = x_i$ and $c_i = c_{i,x_i}$; $x_{i-1} = x_i = y$ and $c_i = 1$; or $x_{i-1} = y \neq x_i$ and $c_i = d_{i,x_i}$.
Notice that the sequence of registers $x_0,\ldots, x_m$ will always be of the form $y,\ldots,y,x,\ldots,x$, \ie the register changes at most once (from the additional register $y$ to one of the registers $x \in \X$).

For every $x \in \X$ (in particular $x \neq y$) we denote the set $\runs_{x}(w)$ of runs terminating in register $x$ on word $w$.
Notice that $|\runs_{x}(w)| \le |w|+1$.
Thus \Cref{eq:simple} and \Cref{eq:simple2} are essentially equivalent to $\Bb(w) = \bigoplus_{\rho \in \runs(w)}\val(\rho)$.

Over the semiring $\semiringZR$ the Equations~\eqref{eq:simple} and~\eqref{eq:simple2}
become
\begin{align}\label{eq:simple_tropical}
\begin{split}
&\Aa(w) = \min_{x \in \X} \set{F(x) + \Aa_x(w)}, \\
&\Aa_x(w) =  \min_{1 \le i \le n+1}\set{d_{a_{i-1},x} \; + \sum_{j \in \set{i,\ldots,n}}c_{a_j,x}}. 
\end{split}
\end{align}

\subsection{Zero isolation over rationals and boundedness over the min-plus semiring}
This subsection formally develops the relationship between zero isolation over the rationals and the boundedness problem over  $\semiringZR$ . In particular, we formalise exactly what we mean by \cref{thm:rational_to_tropical} by proving \cref{thm:equivalence}.

In this section we will need the \emph{max-product semiring} $\semiringQmax$ over nonnegative rationals, where the plus operation is $\max$ and product is the standard product.

Let $\Aa = (\Sigma, I, F, \X, (\delta_a)_{a \in \Sigma})$ be an \simpleCRA{} over a semiring $\semiringS$. We say that $\Aa$ is in \emph{normal form} if:
\begin{enumerate}[label=(\arabic*)]
\item $\delta_a(x) = (c_{a,x} \odot x) \oplus \one$ or $\delta_a(x) = c_{a,x} \odot x$ for all $x \in \X$ and $a \in \Sigma$ (\ie $d_{a,x} \in \set{\zero, \one}$); and
\item $I(x) = F(x) = \one$ for all $x \in \X$. 
\end{enumerate}

Let $\mathcal{A} = (\Sigma, I, F , \X, \delta)$ be an \simpleCRA.
We write $\mathcal{A}_{(+,\cdot)}$ and $\mathcal{A}_{(\max,\cdot)}$ to emphasise that $\Aa$ is evaluated over the semiring $\semiringQ$ and over the semiring $\semiringQmax$, respectively.

Let $\Aa_{\log}$ be the automaton $\Aa$, where every constant $c$ (appearing in $I$, $F$ or $\delta$) is replaced by $-\logbr{c}$. In the special case where $c = 0$ we replace it with $+\infty$. Since the automaton $\Aa_{\log}$ will be considered only over the semiring $\semiringZR$ we write $\Aa_{\log}$ without indicating the semiring operations.

The goal of this section is to establish links between some of the problems we consider. 
In particular, we first show the equivalence of the
zero isolation for $\mathcal{A}_{(+,\cdot)}$ with the boundedness
of $\Aa_{\log}$ (\cref{thm:equivalence}). Then we show interreducibility of the later and 
the universal coverability problem for orthant vector addition systems 
(\cref{theorem:boundedness_ovass}). Finally, we show how to move to a continuous
setting in such systems (\cref{theorem:continuous_semantics}). Studying the 
universal coverability problem in a continuous system will be done in the next section.

Before we proceed we discuss some notations on \simpleCRA. Recall that $\runs_{x}(w)$ is the set of runs terminating in register $x$ on word $w$. For every $\rho_x = x_0,c_1,x_1,c_2,\dots,c_{m},x_m \in \runs_{x}(w)$ we use similar notation as for runs in \cref{sec:boundedness}. Namely, $\valinf(\rho_x) = \bigodot_{i=1}^m c_i$, $\val(\rho_x) = I(x_0) \odot \valinf(\rho_x) \odot F(x_m)$. Notice that $\val(\rho_x)$ evaluates in the same way in $\Aa_{(+,\cdot)}$ and $\Aa_{(\max,\cdot)}$, but in $\Aa_{\log}$ the product is $+$. Then:
\begin{itemize}
 \item $\Aa_{(+,\cdot)}(w) = \sum_{x \in \X} \sum_{\rho_x \in \runs_x(w)} \val(\rho_x)$;
 \item $\Aa_{(\max,\cdot)}(w) = \max_{x \in \X} \max_{\rho_x \in \runs_x(w)} \val(\rho_x)$;
 \item $\Aa_{\log}(w) = \min_{x \in \X} \min_{\rho_x \in \runs_x(w)} \val(\rho_x)$.
\end{itemize}

We define the constant $R_\Aa$ (or just $R$ when $\Aa$ is clear from context) as the maximum ratio between two elements appearing in transitions of $\mathcal{A}$.
Formally, let
$$
\mathcal{C} = \left(\{1\}\cup \{c_{a,x} ,d_{a,x} \mid a \in \Sigma,x \in \X\} \right)
\setminus\{0\},
$$
where $c_{a,x}$ and $d_{a,x}$ come from \cref{eq:simple2}.
Then $R = \frac{\max\mathcal{C}}{\min\mathcal{C}}$ ($1$ is included in $\Cc$ to make sure $R$ is well-defined). Observe that: $R \ge 1$; and for every $c,d\in \mathcal{C}$, $cR \ge d$ and $cR \ge 1$.

We also define $S_\Aa$ (or just $S$) as the maximal product (potentially inverted)
of initial and final weights, \ie
$$
S = \max_{x,y \in \X} \set{I(x) \cdot F(y), \frac{1}{I(x)\cdot F(y)} \mid  I(x), F(y) \neq 0}.
$$
Then either $\val(\rho_x) = 0$ or
\begin{align}\label{eq:S}
\frac{1}{S} \cdot \valinf(\rho_x) \le \val(\rho_x) \le S \cdot \valinf(\rho_x)
\end{align}
for all runs $\rho_x$.

\begin{restatable}{theorem}{ThmCraOvas}\label{thm:equivalence}
Let $\Aa$ be an \simpleCRA. The following statements are equivalent:
\begin{enumerate}[label={(\roman*)}]
\item\label{eq:1} $\forall \epsilon >0$  $\exists w_i \in \Sigma^* $ such that $\mathcal{A}_{(+,\cdot)}(w_i) < \epsilon$.
\item\label{eq:2}  $\exists w $ such that $\mathcal{A}_{(\max,\cdot)}(w) < \frac{1}{RS}$.
\item\label{eq:3} $\exists  \tau < 1$, $\forall i$ $\exists w_i$ such that $\mathcal{A}_{(\max,\cdot)}(w_i) <  \tau^{i} \cdot S$.
\item\label{eq:4}  $\exists w $ such that $\Aa_{\log}(w) > -\logbr{\frac{1}{RS}}$.
\item\label{eq:5} $ \Aa_{\log}$ is unbounded: $\forall i > 0$ $\exists w_i$ s.t. $ \Aa_{\log}(w_i) > i$
	\end{enumerate}
\end{restatable}

Notice that~\ref{eq:1} is the complement of the zero isolation problem over $\semiringQ$ and that~\ref{eq:5} is the complement of the boundedness problem over $\semiringZR$ (assuming the domain is $\log(\Qposs) \cup \set{\infty}$).
Hence \cref{thm:rational_to_tropical} is a corollary of this theorem.

\begin{proof}[Proof of \cref{thm:equivalence}]
It is straightforward that \ref{eq:1} $\implies$ \ref{eq:2},
\ref{eq:3} $\implies$ \ref{eq:2}, and \ref{eq:5} $\implies$ \ref{eq:4}.
To complete the proof we show that: \ref{eq:2} $\implies$ \ref{eq:1}, \ref{eq:2} $\implies$ \ref{eq:3}, \ref{eq:2} $\iff$ \ref{eq:4}, and \ref{eq:3} $\implies$ \ref{eq:5}.

We start with \ref{eq:2} $\implies$ \ref{eq:1} and \ref{eq:2}$\implies$ \ref{eq:3}. We assume \ref{eq:2} and prove both \ref{eq:1} and \ref{eq:3}.
Let $w = a_1\dots a_m$ be such that $\mathcal{A}_{(\max,\cdot)}(w) < \frac{1}{RS}$.

Recall that $\set{\rho_1,\dots,\rho_k} = \runs(w)$ is the set of accepting runs of $\Aa_{(+,\cdot)}$ over $w$. Then by the assumption we have $\val(\rho_j) < \frac{1}{RS}$, and thus by \cref{eq:S}
\begin{align}\label{eq:R}
\valinf(\rho_j) < \frac{1}{R} 
\end{align}
for all $1 \le j \le k$.

\begin{claim}\label{claim:technical}
For all $i\in \mathbb{N}$, $x \in \X$, $1 \le t \le m$, and $\gamma\in \runs_{x}(a_1^i \cdots a_t^i)$ there exists $\rho_\gamma\in \runs_{x}(a_1\cdots a_t)$ such that $\val(\gamma) \le (\valinf(\rho_\gamma)R)^i S$.
\end{claim}

\begin{claimproof}We fix $i$ and $x\in\mathcal{X}$ and proceed by induction on $t$. We assume that $\val(\gamma) > 0$, otherwise the claim is trivial.
For the base case let $\gamma\in\runs_x(\epsilon)$. Then 
$\val(\gamma) \le S = (\rho_\gamma R)^0 S$ for
$\rho_\gamma = \gamma\in \runs_x(\epsilon)$.
For the induction step let $0 \le t < m$. We assume that for all $\gamma\in \runs_x(a_1^i \cdots a_t^i)$ there exists $\rho_\gamma \in \runs_x(a_1 \cdots a_t)$ such that $\val(\gamma) \le (\valinf(\rho_\gamma)R)^i S$. Let us prove the induction step for $t+1$.

Consider $\gamma\in \runs_x(a_1^i \cdots a_{t+1}^i)$ and let $\delta_{a_{t+1}}(x)= cx + d$. Then $\gamma$ can be characterised as one of the following cases:
\begin{enumerate}
	\item\label{case:1} $\valinf(\gamma) = dc^{v}$ for some $0\le v\le i-1$;
	\item\label{case:2} $\valinf(\gamma) = \val(\gamma')c^i$ for some $\gamma' \in \runs_x(a_1^i \cdots a_{t}^i)$.
\end{enumerate}

Case~\ref{case:1}.
Note that there exists $\rho_\gamma \in\runs_x(a_1 \cdots a_{t+1})$ such that $\val(\rho_\gamma) = d$. Thus
\begin{align*}
\frac{1}{S}\val(\gamma) & \le \valinf(\gamma) &
\\ &= dc^{v} \le d (dR)^v &\text{as } c \le d R
\\ &\le d (dR)^vR = (dR)^{v+1} &\text{as } R \ge 1
\\ &\le  (dR)^{i} = (\val(\rho_\gamma) R)^{i}   &\text{as } v+1\le i \text{ and } dR \ge 1.
\end{align*}

Case~\ref{case:2}. 
By induction we have $\rho_{\gamma'}\in\runs_x(a_1 \cdots a_t)$
such that $\val(\gamma') \le (\valinf(\rho_{\gamma'})R)^i S$.
Moreover, there exists
$\rho_{\gamma} \in \runs_x(a_1 \cdots a_{t+1})$ with
$\valinf(\rho_{\gamma}) =\valinf(\rho_{\gamma'}) \cdot c$. Thus
\begin{align*}
\frac{1}{S} \val(\gamma) &= \frac{1}{S} \val(\gamma') \cdot c^i \le (\valinf(\rho_{\gamma'})R)^i \cdot c^i &\text{by induction}
\\&= (\valinf(\rho_{\gamma'}) \cdot cR)^i = (\valinf(\rho_{\gamma}) \cdot R)^i .
\end{align*}
This completes the proof of \cref{claim:technical}.
\end{claimproof}

Let $w_i = a_1^i \ldots a_m^i$. To prove~\ref{eq:1} we show that $\mathcal{A}_{(+,\cdot)}(w_i)\to 0$.
Let $c = \max \{\valinf(\bar\gamma) \cdot R \mid \bar\gamma \in \runs(w)\}$
and let $\gamma \in \runs_x(w_i)$. \cref{claim:technical} implies that $\val(\gamma) \le c^i S$ (because there exists $\rho_\gamma \in  \runs_x(a_1 \cdots a_m)$ such that $\val(\gamma) \le (\valinf(\rho_\gamma) \cdot R)^iS$, and, by definition, $\valinf(\rho_\gamma) \cdot R \le c$). Thus
\begin{align}\label{eq:sumplustimes}
\mathcal{A}_{(+,\cdot)}(w_i) = \sum_{\gamma\in \runs(a_1^i \cdots a_m^i)} \val(\gamma) &\le
|\runs(a_1^i \cdots a_m^i)| \cdot c^i \cdot S.
\end{align}
Notice that $|\runs(w_i)| \le m |\X| i$, where $m$ and $|\X|$ are fixed. By \eqref{eq:R} $c < 1$, and
so $\mathcal{A}_{(+,\cdot)}(w_i) \le m|\X|i c^i S \xrightarrow{i\to\infty} 0$.

To prove~\ref{eq:3} we use the same sequence of words $w_i$. Since the sum becomes the maximum,  \cref{eq:sumplustimes} becomes
\begin{align*}
\mathcal{A}_{(\max,\cdot)}(w_i) = \max_{\gamma\in \runs(a_1^i \cdots a_m^i)} \val(\gamma) &\le
c^i \cdot S.
\end{align*} Then it suffices to define $\tau = c$.

We now show \ref{eq:2} $\iff$ \ref{eq:4}.
Observe that for every $C_1, C_2, c_1,\ldots, c_m \in \Qposs$ we have the following
\begin{equation}\label{eq:equivalence}C_1 >  \prod_{i = 1}^m c_i \ge C_2 \iff 
-\logbr{C_1} <\sum_{i = 1}^m -\logbr{c_i} \le -\logbr{C_2}.
\end{equation}
Notice also that for every run
$$
\rho = x_0,c_1,x_1,\ldots, c_m, x_m
$$
in $\mathcal{A}_{(\max,\cdot)}(w)$ over $w$ there is a unique corresponding run 
$$
\rho_{\log} = x_0,\logbr{c_1},x_1,\ldots, \logbr{c_m}, x_m
$$
in $\Aa_{\log}$ over $w$ (recall that by convention $\log(0) = +\infty$).

First assume \ref{eq:2}, \ie there is $w$ such that $\mathcal{A}_{(\max,\cdot)}(w) < \frac{1}{RS}$. By \cref{eq:equivalence} with $C_1 = \frac{1}{RS}$ we have $\mathcal{A}_{\log}(w) > -\log(\frac{1}{RS})$, which is \ref{eq:4}. Similarly, assume \ref{eq:2} does not hold. Then for all $w$: $\mathcal{A}_{(\max,\cdot)}(w) \ge \frac{1}{RS}$, hence by \cref{eq:equivalence} with $C_2 = \frac{1}{RS}$ for all $w$ we have $\mathcal{A}_{\log}(w) \le-\log(\frac{1}{RS})$, which is the complement of \ref{eq:4}.

Finally, we show \ref{eq:3} $\implies$ \ref{eq:5}. From \ref{eq:3} we get $\exists  \tau < 1$, $\forall i$ $\exists w_i$ such that $\mathcal{A}_{(\max,\cdot)}(w_i) <  \tau^{i} \cdot S$. This entails that 
$\forall \epsilon > 0$ $\exists w$ $\mathcal{A}_{(\max,\cdot)}(w)< \epsilon$. Hence by \cref{eq:equivalence} with $C_1 = \epsilon$ we have $\Aa_{\log}(w) > -\log(\epsilon)$. We conclude since $-\log(\epsilon)\to\infty$ for $\epsilon\to 0$.
\end{proof}

\subsection{Other proofs}

\begin{restatable}{lemma}{LemCraNormal}\label{lemma:normal}
 Let $\Aa$ be an \simpleCRA{} over $\semiringZR$. There exists an \simpleCRA{} in normal form $\Bb$ such that $\Aa$ is bounded if and only if $\Bb$ is bounded. Moreover, the only possible constant that appears in $\Bb$ and not in $\Aa$ is $0$.
The construction of $\Bb$ is effectively polynomial in $|\Aa|$.
\end{restatable}

\begin{proof}
Let $\Aa = (\Sigma, I, F, \X, (\delta_a)_{a \in \Sigma})$. Let $\Aa'$ be a copy of the automaton $\Aa$ with transitions $\delta'_a$ redefined as follows. Let $\delta_a(x) = \min(c_{a,x} + x, d_{a,x})$. If $d_{a,x} = + \infty$ then $\delta'_a(x) = \delta_a(x)$, otherwise $\delta'_a(x) = \min(c_{a,x} + x, 0)$. Then $\Aa$ is bounded iff $\Aa'$ is bounded. Indeed, let $d = \max \set{|d_{a,x}| \mid a \in \Sigma, x \in \X}$, we claim that $\Aa'(w) + d \ge \Aa(w)$ for all $w \in \Sigma^*$, which will conclude the proof. Indeed, it follows from \cref{eq:simple_tropical} and the fact that $d_{a,x}$ was changed only if $d_{a,x} \neq + \infty$.

Now from $\Aa'$ we construct the following automaton
$\Bb = (\Sigma \cup \set{\$}, I'', F'', \X'', (\delta_a'')_{a \in \Sigma})$. The set of registers is $\X'' = \set{x \in \X \mid F(x) \neq +\infty}$. In the special case where $\X''$ is empty, notice that $\Aa$ always outputs $+\infty$ and it is easy to define $\Bb$. Thus we assume $\X'' \neq \emptyset$ and we define $I''(x) = F''(x) = 0$ for all $x \in\X''$. Finally, for every $a \in \Sigma$ the transitions $\delta_a''$ are defined as $\delta_a$, where every $x \not \in \X$ is substituted with $+\infty$. The transitions for $\$$ are defined as follows: $\delta_a''(x) = \min(+\infty + x, I(x))$, \ie they reset the registers to the initial values.

We prove that $\Bb$ satisfies the claim.
First, notice that $\Aa'(w) = \Bb(v\$w)$ for every $w \in \Sigma^*$ and $v \in (\Sigma \cup \set{\$})^*$. Thus we understand the behaviour of $\Bb$ on inputs with at least one $\$$. For the remaining cases we show that $\Bb(w) \le \Bb(\$w) + b$, for all $w \in \Sigma^*$, where $b = \max \set{|I(x)| \mid x \in \X''}$.
Indeed, recall that $I(x)$ affects only $d_{a_{i-1},x}$ in \cref{eq:simple_tropical}. Thus
\begin{align*}
\Bb(w) = \min_{x \in \X}\min_{1 \le i \le n+1} \set{F(x) + d_{a_{i-1},x} \; + \sum_{j \in \set{i,\ldots,n}}c_{a_j,x}} =
\\ b + \min_{x \in \X}\min_{1 \le i \le n+1} \set{F(x) + (d_{a_{i-1},x} -b) \; + \sum_{j \in \set{i,\ldots,n}}c_{a_j,x}} \le
\\ b + \Bb(\$w).
\end{align*}
Thus $\Aa$ is unbounded if and only if $\Bb$ is unbounded.
\end{proof}

\ThmBoundOvas*

\begin{proof}First, we show syntactic translations between the models. Then in the second part we will prove the reductions. Let $\Aa = (\Sigma, I, F, \X, (\delta_a)_{a \in \Sigma})$ be an \simpleCRA. By \cref{lemma:normal} for the boundedness problem we can assume that $\Aa$ is in normal form. Let $|\X| = d$. 
We will define a $d$-OVAS $\V_{\Aa} = (T_A)_{A \in \orthants}$. Recall the notation $T = \bigcup_{A \in \orthants} T_A$.
For convenience we will treat $T$ as a subset of $\R^{\X}$ instead of $\R^d$.

Recall that $\delta_a(x) = \min(x + c_{a,x}, d_{a,x})$ for every $a \in \Sigma$, $x \in \X$, where $d_{a,x} \in \set{0,+\infty}$. For every $a \in \Sigma$ we define the vector $\bc_a \in \R^\X$, where $\bc_a[x] = c_{a,x}$. Then $T = \set{\bc_a \mid a \in \Sigma}$. For every $\bc_a \in T$ we define the orthant
\begin{align}
\begin{split}\label{eq:vector_transition}
 A_a = \{\bv \in \R^\X \mid \text{ if } d_{a,x} = +\infty \text{ then } \bv[x] \le 0, \\ \text{ otherwise } \bv[x] \ge 0\}.
 \end{split}
\end{align}
Note that more than one letter can define the same vector $\bc_a$. 
Then we need to store only the minimal orthants defining~$\bc_a$.

The conversion to an \simpleCRAlong $\Aa_\V = (\Sigma, I, F, \X, (\delta_a)_{a \in \Sigma})$ from a given a $d$-OVAS $\V = (T_A)_{A \in \orthants}$ with $T = \bigcup_{A \in \orthants}T_A$ is analogous. We define $\Sigma = T$ and $\X = \set{x_1,\ldots,x_d}$. The transitions are defined by $\delta_{\bv}(x_i) = \min(x_i + \bv[i],  d_{\bv,x_i})$, where $d_{\bv,x_i} = +\infty$ if the orthant $\orthmin{\bv}$ is defined by $x_i \le 0$ and $d_{\bv,x_i} = 0$ otherwise.
The vectors $I$ and $F$ are defined to be the zero vector $\bf{0}$ as required by the normal form.

Notice that given a $d$-OVAS $\V$ the $d$-OVAS $\V_{\Aa_{\V}}$ (\ie obtained from $\Aa_\V)$ is the same as $\V$ up to reordering the coordinates in vectors $\R^d$. To prove \cref{theorem:boundedness_ovass} it suffices to prove that $\Aa$ is unbounded if and only if $\V_\Aa$ is a positive instance of universal coverability. Indeed, for the converse reduction one would have to prove that $\V$ is a positive instance of universal coverability if and only if $\Aa_\V$ is unbounded. But this will follow from $\V$ being equivalent to $\V_{\Aa_\V}$.

Fix an \simpleCRA $\Aa$ and the $d$-OVAS $\V_\Aa$ obtained using the construction above.
We prove the following claim, that concludes the theorem.
\begin{restatable}{claim}{ClaimOvasCra} \label{claim:ovass_weighted}
For every $k \in \N \setminus \set{0}$ and every $w = a_1 \ldots a_n$:
$\Aa(w) \ge k$ if and only if  there is a run $\bv_0,\ldots,\bv_n$, where
$\bv_0 = (-k,\ldots,-k)$, $\bv_n \in \R_{\ge 0}^d$, and
$\bv_{i+1} = \bv_i + \bc_{a_{n-i}}$
for all $i \in \set{0,\ldots,n-1}$.
\end{restatable}
Intuitively, a word $w = a_1\ldots a_n$ witnesses a big value
for $\Aa$ if and only if the corresponding transitions in
$\V_\Aa$ applied in the reverse order witness a run from a low
starting point.

\begin{claimproof}
Recall that since $\Aa$ is in
normal form and by \cref{eq:simple_tropical} given a word
$w = a_1,\ldots a_n$
\begin{align*}
\Aa(w) = \min_{x \in \X} \min_{1 \le i  \le n+1}\set{d_{a_{i-1},x} \; + \sum_{j \in \set{i,\ldots,n}}c_{a_j,x}},
\end{align*}
where $d_{a_0,x} =0$ for all $x \in \X$.

Since $d_{a_i,x} \in \set{0,+\infty}$ we get that $\Aa(w) \ge k$ if and only if both items hold:
\begin{enumerate}
\item\label{eq:wa1} $\sum_{j \in \set{1,\ldots,n}}c_{a_j,x} \ge k$ for all $x \in \X$. 
 \item\label{eq:wa2} $\sum_{j \in \set{i+1,\ldots,n}}c_{a_j,x} \ge k$ for all $i \in \set{1,\ldots,n}$ and for all $x \in \X$ such that $d_{a_i,x} = 0$.
\end{enumerate}
We assume that an empty sum evaluates to $0$. Then, since $k > 0$, for $i = n$ \cref{eq:wa2} is valid only if $d_{a_n,x} = +\infty$ for all $x \in \X$.

Now, let $\bv_0 = (-k,\ldots,-k)$ and let $\bv_{i+1} = \bv_i + \bc_{a_{n-i}}$ for all $i \in \set{0,\ldots,n-1}$.
We prove that $\bv_n \in \R^d_{\ge 0}$ if and only if \cref{eq:wa1} holds and that $\bv_0,\ldots,\bv_n$ is a run if and only if \cref{eq:wa2} holds.

For the first part notice that $\bv_n \in \R_{\ge 0}^d$ if and only if $\sum_{j = 1}^n \bc_{a_j} \ge k$, which is equivalent to \cref{eq:wa1}. For the second part it suffices to prove that $\bv_{n-i} \in A_i$, for some orthant $A_i$ such that $A_{a_{i}} \preceq A_i$ for all $i \in \set{1,\ldots,n}$. Recall that $\bv_{n-i} = \bv_0 + \sum_{j=i+1}^{n}\bc_{j}$.
Since $\bv_0 = (-k,\ldots,-k)$ and by \cref{eq:vector_transition} this is equivalent to $\sum_{j=i+1}^{n}\bc_{j}(x) \ge k$ for all $x \in \X$ such that $d_{a_i,x} = 0$. This is equivalent to \cref{eq:wa2}.
\end{claimproof}
\end{proof}

\subsection{OVAS with continuous semantics}

\LemRunsCont*

\begin{proof}

Fix some $\bv', \bw' \in L$ as in the statement of the lemma. Let $A_0,\ldots,A_k$ be the sequence of elements in $\orthantpath{\bv', \bw'}$ ordered starting from the closest to $\bv'$. In particular $A_0 = \orthmax{\bv'}$ and $A_k = \orthmax{\bw'}$. Let $\bu_1,\ldots,\bu_{k} \in L$ be the unique sequence of points such that $\bu_i \in A_{i-1} \cap A_{i}$ for all $i \in \set{1,\ldots,k-1}$. For simplicity we also denote $ \bv' = \bu_0$ and $\bw' = \bu_{k+1}$. It remains to prove that $\bu_0,\bu_1,\ldots,\bu_k, \bu_{k+1}$ is a continuous run, where $A_0,\ldots, A_k$ are the corresponding witnessing orthants. Indeed, $\bu_{i-1}, \bu_{i} \in A_{i-1}$ for every $i \in \set{1,\ldots,k+1}$. Since  both $\bu_{i-1}$ and $\bu_{i}$ are on the line from $\bv$ to $\bw$ we have $A_{i-1}\in\orthantpath{\bv,\bw}$. Hence, we have that  there exists $\delta_i \in \R_0$ such that $\delta_i(\bu_i - \bu_{i-1}) = \delta_{A_{i-1}}(\bw - \bv) \in T_{A_{i-1}}$.
\end{proof}

\LemContShift*

\begin{proof}
The direction of the straight line route from $\bv$ to $\bw$ is the same direction as the straight line route from $\bv'$ to $\bw'$. We show that this direction is available in the transition set of all orthants on the route from $\bv'$ to $\bw'$. We are able to conclude this since the direction was available in all orthants on the route from $\bv$ to $\bw$ and every point between $\bv'$ and $\bw'$ is in some  orthant at least as large (w.r.t. $\preceq$) as on the route from $\bv$ to $\bw$.

We start by proving that for each $C'\in \orthantpath{\bv',\bw'}$ there exists $C\in\orthantpath{\bv,\bw}$ such that $C \preceq C'$. Indeed, by definition $\bv' = \bv + \bv''$ and $\bw' = \bw + \bw''$, where $\bv''$ and $\bw''$ have only nonnegative coefficients. Let
$$
\bu' = \delta\bv' + (1-\delta)\bw' = \delta\bv + (1-\delta)\bw + \delta\bv'' + (1-\delta)\bw'',
$$
for some $\delta \in [0,1]$ and denote $\bu = \delta\bv + (1-\delta)\bw$. By definition $\bu' \ge \bu$ and $\orthmax{\bu} \in \orthantpath{\bv,\bw}$. We conclude since $\bu'$ was defined for any $\delta \in [0,1]$ and $\orthmax{\bu} \preceq \orthmax{\bu'}$.

To finish the proof recall that by monotonicity $C \preceq C'$ implies $T_{C} \subseteq T_{C'}$. Since $\frac{\delta_{C}}{\delta'}(\bw' - \bv') \in T_C$
the conclusion follows from \cref{lemma:runs_continous}.
\end{proof}

\begin{corollary}
    \label{corollary:run_simple_shift}
    Suppose $\bu \stepsc \bv$ then for any $\Delta\in \R_{\ge 0}$, we have $\Delta+ \bu\stepsc \Delta+\bv$.
\end{corollary}

The following lemma is similar to a standard lemma for continuous VASS (\eg see \cite[Lemma 12]{FracaH15}).
\begin{restatable}{lemma}{LemScaling}\label{lemma:rescaling}
Suppose $\bu \stepsc \bv$ then for any $\lambda> 0$, we have $\lambda \bu\stepsc \lambda\bv$.
\end{restatable}

\begin{proof}
    Let $\bu=\bv_1, \bv_2,\dots, \bv_n = \bv$ be a run witnessing $\bu \stepsc \bv$  such that $\bv_i, \bv_{i+1} \in A_i$ for some orthant $A_i$.
    Hence $ \delta_i(\bv_{i+1} - \bv_i) \in T_{A_i}$ for some $\delta$. 

    Observe $\lambda \bv_{i+1}$ and $\lambda \bv_i$ are also both in $A_i$, and hence we have  $\delta_i/\lambda(\lambda \bv_{i+1} - \lambda \bv_i)\in T_{A_i}$, so $\lambda \bv_i\stepsc \lambda \bv_{i+1}$.
\end{proof}

\LemRunsScale*

\begin{proof}By \cref{lemma:rescaling} we have  $\bv_{i} \stepsc \bv_{i+1}$ then $m\bv_{i} \stepsc m\bv_{i+1}$.
    
    Let $\Delta = \bv_0 - m \bv_0$. Observe $\Delta \in \mathbb{R}_{\ge 0}^d$ since $\bv_0 \in \mathbb{R}_{< 0}^d$ and $m \ge 1$.

    Observe that $\bv_{i}' = \bv_{0} + ma_1 +\dots+ma_i = \bv_{0} -m\bv_{0}+m\bv_{0}+ ma_1 +\dots+ma_i = \bv_{0} -m\bv_{0} + m\bv_{i} = m\bv_{i} +\Delta$.

    Therefore $\bv_{i}' = m\bv_{i} + \Delta \stepsc m\bv_{i+1} + \Delta =\bv_{i+1}'$ by \cref{corollary:run_simple_shift}.
\end{proof}

\section{Additional material for Section~\ref{sec:coverability}} \label{app:cover}

\subsection{Undecidability of the OVAS coverability problem over \texorpdfstring{$\mathbb{Z}$}{Z}}
\ThmCoverUndec*

In order to prove \cref{theorem:coverability}, we provide a reduction from the halting problem for two-counter machines, which is an undecidable problem~\cite{minsky1967computation}, to the coverability problem for OVAS over $\Z$. Let us first define two-counter machines and the halting problem.

\begin{definition}\label{defn:two-counter-machine}
A \emph{two-counter machine} $M$ is a tuple\\ $(Q, q_1, q_{halt}, T)$.
Here $Q$ is a finite set of states, with dedicated initial and halting states $q_1,q_{halt}\in Q$. The transitions $T\subseteq \mathsf{Type}\times Q\times Q$, where $\mathsf{Type}$ refers to the counter actions; \\$\mathsf{Type} = \{ \mathsf{inc}_{c_1}, \mathsf{inc}_{c_2}, \mathsf{dec}_{c_1}, \mathsf{dec}_{c_2}, \mathsf{zero}_{c_1}, \mathsf{zero}_{c_2} \}$. We assume w.l.o.g.\ that each transition of $M$ changes each counter by at most one.

A \emph{configuration} of a run is an element of $Q\times \mathbb{N}\times \mathbb{N}$ representing the current state, and the value of each of the two counters. We write $(p,c_1,c_2) \in Q\times \mathbb{N}\times \mathbb{N}$ as  $p(c_1,c_2)$.

We write $\mathsf{inc}_{c_i}(q,q')$ to denote that the machine should move from state $q$ to $q'$ incrementing counter $c_i$. Therefore $\mathsf{inc}_{c_1}(p,q)$ maps $p(c_1,c_2)$ to $q(c_1+1,c_2)$ and similarly $\mathsf{inc}_{c_2}(p,q)$ maps $p(c_1,c_2)$ to $q(c_1,c_2+1)$.

The transition $\mathsf{dec}_{c_i}(q,q')$ should decrement $c_i$, provided the value of $c_i >0$.
A transition $\mathsf{dec}_{c_1}(p,q)$ maps $p(c_1,c_2)$ to $q(c_1-1,c_2)$ if $c_1 \ge 1$. Similarly $\mathsf{dec}_{c_2}(p,q)$ maps $p(c_1,c_2)$ to $q(c_1,c_2-1)$ if $c_2 \ge 1$.

Finally $\mathsf{zero}_{c_i}(q,q')$ should move to state $q'$, only if the value of $c_i = 0$, without affecting the value. A transition $\mathsf{zero}_{c_1}(p,q)$ maps $p(c_1,c_2)$ to $q(c_1,c_2)$ if $c_1 = 0$ and  $\mathsf{zero}_{c_2}(p,q)$ maps $p(c_1,c_2)$ to $q(c_1,c_2)$ if $c_2 = 0$.

Let $\sfstart(t)$ and $\sfend(t)$ denote the starting and ending state of $t$ respectively, for example $\sfstart(\incsf_{c_1}(q,q')) = q$ and $\sfend(\incsf_{c_1}(q,q')) = q'$.

We say the two-counter machine is \emph{deterministic} if for every configuration $p(c_1,c_2)$ there exists at most one transition in $T$ that can be applied. In particular, for every state there should either be only one $\mathsf{inc}_{c_i}$ transition or at most one of each $\mathsf{dec}_{c_i}$ and $\mathsf{zero}_{c_i}$ transition with the same $c_i$.

A \emph{run} of a two-counter machine is encoded by a sequence of configurations, interleaved by the transition between them:
$q_1(0,0)\xrightarrow{t_1} q^{1}(c_1^{(1)},c_2^{(1)})\xrightarrow{t_2} q^{2}(c_1^{(2)},c_2^{(2)})\xrightarrow{t_3}  \dots \xrightarrow{t_n} q^{n}(c_1^{(n)},c_2^{(n)})$ encodes a run of length $n$.
Without loss of generality, the state information is redundant since it is encoded in the choice of transition and the previous counter value.

The two-counter machine $M$ \emph{halts} if there exists a run from $q_1(0,0)$ to $q_{halt}(c_1,c_2)$ for any $c_1,c_2\in \mathbb{N}$. The \emph{halting problem} asks if $M$ halts.
\end{definition}

\begin{proof}[Proof of \cref{theorem:coverability}]
Fix a two-counter machine $M$ with the set of states $Q = \{q_1, \ldots, q_n\}$ of size $n$, initial state $q_1$, halting state $q_n$ and with two counters $c_1$ and $c_2$.
We define an $(n+4)$-OVAS $\V_M$ over $\Z$ with a distinguished vector $v_I \in \Z^{n+4}$ such that
there is a halting run in $M$ if and only if there is a run $v_I \steps v_P$ for some $v_P \in \R_{\geq 0}^{n+4}$.

Given a $d$-OVAS we number its coordinates: $1, \ldots, d$. We name the counter on coordinate number $i$
as counter $x_i$.
For a set of coordinates $S \subseteq [1,n+4]$ we define the vector $\1_S \in \set{0,1}^{n+4}$ by $\1_S[i] = 1$ if and only if $i \in S$. If $S = \{i\}$ then we write for simplicity $\1_i$. For an orthant $A = \set{\bx \mid \epsilon_1 x_1 \ge 0, \ldots, \epsilon_d x_d \ge 0}$
let $\pos(A)$ be the set of coordinates on which vectors from this orthant are nonnegative, \ie $\pos(A) = \set{i \in \set{1,\ldots,d} \mid \epsilon_i = 1}$.

The idea behind the construction is that for every state $q_i \in Q$ there is a distinguished counter $x_i$ in $\V_M$.
As an invariant we will keep that exactly one counter among $x_1,\ldots,x_n$ is nonnegative (indicating the state of the configuration).
Moreover, each counter $c_\ell$ for $\ell \in \{1,2\}$ in $M$ is simulated by two counters $x_{n+\ell}$ and $x_{n+\ell+2}$ in $\V_M$.
Intuitively, the counter $x_{n+\ell}$ stores the value of $c_\ell-1$, while $x_{n+\ell+2}$ stores $-c_\ell$.
Notice that always exactly one of $x_{n+\ell}$ and $x_{n+\ell+2}$ is nonnegative and it depends on whether $c_\ell > 0$.
This will allow us to simulate zero-tests, as the transitions checking if a certain counter is zero will be available only in the corresponding orthants.

Intuitively, for every configuration $C$ in $M$ there is a corresponding configuration in $\V_M$, which simulates $C$
using $n+4$ counters of the OVAS.
For a configuration $C = q_i(a_1, a_2)$ of $M$ let the \emph{vector corresponding to $C$}, denoted $\corr(C)$, be as follows:
\[
- \1_{[1,n]} + 2 \cdot \1_i + (a_1-1) \cdot \1_{n+1} +  (a_2-1) \cdot \1_{n+2} - a_1\cdot \1_{n+3} - a_2 \cdot \1_{n+4}.
\]
Notice that if $a_\ell > 0$ then $\corr(C)[n+\ell] \geq 0$ and $\corr(C)[n+\ell+2] < 0$,
but if $a_\ell = 0$ then $\corr(C)[n+\ell] = -1 < 0$ and $\corr(C)[n+\ell+2] = 0 \geq 0$.
In the first case $\corr(C) \in A$ for some $A$ with $n+\ell \in \pos(A)$,
while in the second case $\corr(C) \in A$ for some $A$ with $n+\ell+2 \in \pos(A)$.

The transitions of $\V_M$ are defined as follows:
\begin{itemize}[leftmargin=*]
  \item for each transition $t$ of $M$ which maps $q_i(c_1,c_2) \trans{t} q_j(c_1+v_1, c_2+v_2)$ and
  for each orthant $A$ with $i \in \pos(A)$ we add $2 \cdot (\1_{j} - \1_{i}) + v_1 \cdot (\1_{n+1} - \1_{n+3})
  + v_2 \cdot (\1_{n+2} - \1_{n+4})$ to the set $T_A$ of transitions available at orthant $A$ under the following conditions:
  \begin{enumerate}
    \item if $t = \zerosf_{c_\ell}(q_i,q_j)$ for $\ell \in \set{1,2}$ then we demand that $n+\ell+2 \in \pos(A)$
    (checking that in a faithful simulation this transition can be fired only if $c_\ell = 0$)
    \item if $t= \decsf_{c_\ell}(q_i,q_j)$ for $\ell \in \set{1,2}$ then we demand that $n+\ell \in \pos(A)$
    (checking that in a faithful simulation this transition can be fired only if $c_\ell > 0$)
  \end{enumerate}
  \item for each orthant $A$ with $n \in \pos(A)$ we add a vector $\1_{[1,n+4]} \in \N^{n+4}$ to the set $T_A$.
\end{itemize}
The aim of adding the last transition is to assure that reaching a vector $v$ in $\V_M$, which is positive
in the $n$-th counter (the one corresponding to final state $q_n$) guarantees reaching the positive orthant.
The initial vector $v_I$ in $\V_M$ is the one corresponding to $q_1(0,0)$, namely $\corr(q_1(0,0))$.

Let us first observe that $\V_M$ is a correctly defined OVAS, namely it is monotonic. Indeed, it is easy to see
that in $\V_M$ whenever a transition belongs to $T_A$ for some orthant $A$ then it also belongs to $T_B$
for any orthant $B$ satisfying $B \succeq A$. Moreover $\V_M$ uses only integers in its transitions, so it is an OVAS over $\Z$.

Therefore in order to finish the proof it is enough to argue that indeed there is a run of two-counter machine $M$
from the initial configuration $q_1(0,0)$ to a configuration with final state $q_n$
if and only if there is a run of OVAS $\V_M$ from the initial vector $v_I$ to the positive orthant.

Let $\reach(M)$ be the set of configurations $q_n(a_1, a_2)$ for $a_1, a_2 \in \N$ reachable
by a run of $M$ from $q_1(0,0)$. Let us denote by $\reach(\V_M)$ the set of vectors reachable in $\V_M$ from $v_I$
by a run which do not use the last transition $\1_{[1,n+4]}$.
We first show that
\begin{equation}\label{eq:correspondence}
\reach(\V_M) = \set{\corr(C) \mid C \in \reach(M)}.
\end{equation}

For the ''$\supseteq$'' direction assume that there is a run of $M$ from $q_1(0,0)$ to $q_n(a_1,a_2)$ for some $a_1, a_2 \in \N$.
We show by induction on the length of run in $M$ that if there is a run in $M$ reaching a configuration $C$
then there is a run in $\V_M$ reaching a vector $\corr(C)$. As the initial vector $v_I$ in $\V_M$ is defined as
vector corresponding to the initial configuration of $M$ the statement is trivial for runs of length $0$. To show an induction
step consider a run of $M$ reaching configuration $C$. Let the second to the last configuration on this run be $C'$.
By induction assumption $\corr(C')$ is reachable in $\V_M$. Let us now consider the transition $t$ of $M$ applied in $C'$ to reach $C$.
Let $t$ be a transition which maps $q_i(c_1,c_2) \trans{t} q_j(c_1+v_1, c_2+v_2)$.
We know that if $t$ zero-tests counter $c_\ell$ for $\ell \in \set{1,2}$ then necessarily $C'[c_\ell] = 0$
and thus $\corr(C')[n+\ell] = -1$ and $\corr(C')[n+\ell+2] = 0$. On the other hand if $t$ decreases counter $c_\ell$
(necessarily by at most one) then the counter value of $c_\ell$ in $C'$ is at least $1$, so $\corr(C')[n+\ell] \geq 1-1 = 0$.
Thus due to the definition of $\V_M$ vector $\corr(C')$ belongs to some orthant $A$ for which
\[
u_t = 2 \cdot (\1_{j} - \1_{i}) + v_1 \cdot (\1_{n+1} - \1_{n+3})  + v_2 \cdot (\1_{n+2} - \1_{n+4}) \in T_A.
\]
It is easy to verify that if $C \trans{t} C'$ then $\corr(C') + u_t = \corr(C)$, which finishes the induction step.

The proof of ''$\subseteq$'' direction proceeds similarly. We prove it by induction of the length of the run in $\V_M$.
For the length $0$ we trivially see that $v_I \in \set{\corr(C) \mid C \in \reach(M)}$ as $v_I = \corr(q_1(0,0))$.
For the induction step consider a run of $\V_M$, which does not use $\1_{[1,n+4]}$,
reaching some vector $v$. Let the previous vector on the run
to $v$ be $v'$, by induction assumption we know that $v' = \corr(C')$ for some $C' = q_i(a_1, a_2) \in \reach(M)$.
By definition of $\V_M$ we know that $v = v' +  2 \cdot (\1_{j} - \1_{i}) + v_1 \cdot (\1_{n+1} - \1_{n+3})
 + v_2 \cdot (\1_{n+2} - \1_{n+4})$ for some transition $t$ of $M$  which maps $q_k(c_1,c_2) \trans{t} q_j(c_1+v_1, c_2+v_2)$.
First notice that necessarily $q_i = q_k$ as we require that $k \in \pos(A)$ for the orthant $A$ such that
$\corr(C') \in A$, which means that necessarily $C' = q_k(b_1, b_2)$ for some $b_1, b_2 \in \Z$.
As this transition was possible in $\V_M$ we know that necessarily $a_\ell = 0$ if $t$ is a zero-test on counter $c_\ell$
and $a_\ell > 0$ if $t$ decreases counter $c_\ell$.
It is now easy to see that $C' \trans{t} C$ such that $v = \corr(C)$, which finishes the proof of the ''$\subseteq$'' direction.

Now we are ready to show that there is a run of $M$ from $q_1(0,0)$ to a configuration with final state $q_n$
if and only if there is a run of OVAS $\V_M$ from $v_I$ to the positive orthant.
For the only if implication notice that by~\eqref{eq:correspondence} if $C_F = q_n(a_1,a_2)$ is reachable in $M$ 
then in $\V_M$ vector $\corr(C_F)$ is reachable.
We have $\corr(C_F)[n] \geq 0$ thus $\corr(C_F)$ belongs to an orthant $A$ for which vector $\1_{[1,n+4]} \in T_A$.
Thus adding vector $\1_{[1,n+4]}$ sufficient number of times we can reach the positive orthant in $\V_M$, which finishes
this implication.
For the if implication notice first that each vector of the form $\corr(C)$ for $C = q_i(a_1, a_2) \in \reach(M)$ 
on coordinates from the set $[1,n] \setminus \set{i}$ has negative entries. Therefore if the positive orthant is reachable
by some run in $\V_M$ this run has to use the transition $\1_{[1,n+4]}$. Let $v'$ be the first vector where such a transition
is used. By~\eqref{eq:correspondence} we know that $v' = \corr(C')$ for some $C' \in \reach(M)$.
However as transition $\1_{[1,n+4]}$ is available in $v'$ we know that for the orthant $A$ such $v' \in A$
we have $n \in pos(A)$. This means that $C' = q_n(a_1, a_2)$, which finishes the proof of the if implication.
\end{proof}

\subsection{Decidability results for universal coverability}

\ClaimReachabilitySet*

\begin{claimproof}
The ``only if'' implication is immediate. As $V$ is a positive instance then in particular there is a run, \eg from $-\1_{[1,d]}$ to the positive orthant. 

For the ``if'' implication assume that $\reach(V)$ intersects the positive orthant. Thus there is a run $\bu \stepsc \bv$
for some $\bu \in \R_{<0}^d$ and $\bv \in \R_{\geq 0}^d$. Let $\bu' \in \R^d$. It is easy to see that there are some $\lambda \in \R_{>0}$
and $\Delta \in \R_{\geq 0}^d$ such that $\bu' = \lambda \bu + \Delta$. By
\cref{corollary:run_simple_shift} and \cref{lemma:rescaling}
we have
$\bu' = \lambda \bu + \Delta \stepsc \lambda \bv + \Delta = \bv'$. Notice that if $\bv \in \R_{\geq 0}^d$ then $\bv' \in \R_{\geq 0}^d$,
which concludes the proof. \end{claimproof}

\ClaimRWD*

\begin{claimproof}
The only if implication is trivial as $\rwd \subseteq \reach(V)^\downarrow$.
For the if implication, consider a run $\bu_0, \ldots \bu_n$ from $\bu_0 \in \R_{<0}^d$ to $\bu_n \in \R_{\geq 0}^d$. Let $1 \le i \le n$ be the first $i$ such that $\bu_i \in \R_{\geq 0}^d$. Since $\bu_{i-1} \not \in \R_{\geq 0}^d$ and there exists an orthant $A$ such that $\bu_{i-1}, \bu_{i} \in A$, we conclude that $\bu_i \in W$. This concludes the proof as $\bu_0 \stepsc \bu_i$.
\end{claimproof}

\LemmaSeparator*

\begin{proof}
For the only if implication we observe that if $\rwd$ does not intersect the positive orthant then $\rwd$ is a separator for $V$. Clearly, Condition 4. is satisfied.
Conditions 2. and 3. immediately follow from the definition of $\rwd$.
To see that Condition 1. is satisfied observe that
by \cref{lemma:rescaling} set $\reach(V)$ is closed under scaling. Then $\rwd$ is also closed under scaling.

For the if implication assume that there is a set $S$, which is a separator for $V$.
It suffices to show that $\rwd \subseteq S$.
Let $\bu_0, \ldots \bu_n$ be a continuous run such that $\bu_0 \in \R_{<0}$ and $\bu_n \in W$. Since $S$ is downward closed it suffices to show that $\bu_n \in S$. Let $\bu_{i_1}, \ldots, \bu_{i_m}$ be a subsequence of the continuous run of all elements in $W$ and let $\bu_{i_0} = \bu_0$. Since $\bu_n \in W$ we know that $m \ge 1$. Notice that $\bu_{i_j} \stepsca \bu_{i_{j+1}}$ for some orthant $A$ and for all $0 \le j \le m-1$. This is because $\bu_{j}, \bu_{j+1} \in A$ for some orthant $A$ for all $0 \le j \le n-1$ and an intersection of two different orthants is contained in $W$.
Thus by condition 3. we have
that $\bu_{i_j} \in S$ for all $1 \le j \le m$, which concludes the proof.
\end{proof}

\ClaimLine*

\begin{claimproof}
For every nonzero point $\bx = (x_1,x_2) \in Q$ let $f(\bx) = \frac{x_1}{-x_2}$ (if $x_2 = 0$ then it is $+\infty$).
Since $S$ is closed under scaling and downward closed, we get that if $f(\bx) \in S$ then $f(\by) < f(\bx)$ implies $f(\by) \in S$. Let $c = \sup_{\bx \in Q \cap S} f(\bx)$. If $0 < c < +\infty$ then $\alpha = (c,-1)$; if $c = 0$ then $\alpha = (1,0)$; and if $c = +\infty$ then $\alpha=(0,1)$. That $\vartriangleleft$ is strict or not depends on whether $(c,-1)\in S$, $(0,1) \in S$, and $(1,0) \in S$, respectively.
\end{claimproof}

\section{The zero isolation problem for Copyless CRA is undecidable}
\label{sec:copylessundecidable}

In this section we prove the negative side of \cref{theorem:isolation}. We assume that the semiring is fixed to $\semiringQ$.

\begin{theorem}\label{thm:zerocopylessundecidable}
The zero isolation problem is undecidable for copyless cost register automata. 
\end{theorem}
We first start by showing that it is enough to show that the $<$-threshold problem is undecidable.
\begin{lemma}\label{lemma:lethreshtozeroiso}
The $<$-threshold problem for copyless cost register automata reduces to the zero isolation problem for copyless cost register automata.
\end{lemma}

\begin{proof}
Let $\Aa= (\Sigma, Q, q_1, I, F , \X, \delta)$, for which we ask whether there exists a word $w$ such that $\Aa(w) < C$. 
We define $\Bb = (\Sigma', Q, q_1, I', F' , \X\cup\{z\}, \delta')$ over $\Sigma' = \Sigma \cup \{\#\}$ such that $\Bb(w_1 \# w_2 \dots \# w_n\#) = \frac{\Aa(w_1)}{C}\cdot \ldots \cdot \frac{\Aa(w_n)}{C}$. 
Hence, if there is a word $w\in\Sigma^*$ such that $\Aa(w) < C$, then $\Bb((w\#)^i) \xrightarrow{i\to\infty} 0$. 
If $\Aa(w) \ge C$ for all $w\in\Sigma^*$ then $\Bb(w_1 \# w_2 \dots \# w_n\#)\ge 1$ for all $n\in\mathbb{N},w_i\in\Sigma^*$.

It remains to show that we can implement $\Bb$ as a copyless cost register automaton.
 Introduce a new register $z$ and let $I(z) = 1$.
When $\#$ is read from state $q$, we move to state $q_1$ and we let
$z \leftarrow \frac{1}{C} \cdot z\cdot \sum_{x\in\mathcal{X}} F(q,x)\cdot x$. 
All other registers are updated such that $x \leftarrow I(x) $ for $x\in \mathcal{X}$.  To conclude, $\Bb$ outputs $z$, that is $F'(q_1, z) = 1$, and $F'(q,x) = 0$ for all $x\in\mathcal{X}$.
\end{proof}

Recall the definition of a two-counter machine from~\cref{defn:two-counter-machine}.
Now we reduce the halting problem for \emph{deterministic} two-counter machines, which is an undecidable problem~\cite{minsky1967computation}, to the $<$-threshold problem.

\begin{lemma} \label{lemma:halttolethresh}
The halting problem for two-counter machines reduces to the $<$-threshold problem for copyless CRA.
\end{lemma}
This lemma and \cref{lemma:lethreshtozeroiso} together show undecidability of zero isolation (\cref{thm:zerocopylessundecidable}).
Our proof resembles the proof of \cite[Proposition 27]{DaviaudJLMP021}. 
Both use automata to associate values to encodings of potential runs of a two-counter machine. 
However, the value of the weighted automaton will depend on whether the given run faithfully implements the proper run of the two-counter machine. 
In doing so the automaton must test for errors in the given run, some of these are regular properties which can be encoded in the state transitions. 
However, some of the tests rely on using test gadgets which change the weight of the run. 
The key difference between our proof and \cite[Proposition 27]{DaviaudJLMP021} is to show that the tests can be implemented using (non-linear) copyless transition, rather than the approach of weighted automata which introduce extra runs.

\begin{proof}[Proof of \cref{lemma:halttolethresh}]
Let $M =(Q, q_1, q_{halt}, T)$ be an arbitrary two-counter machine.
We define a copyless CRA $\Aa$, for which the words  over $\Sigma$ encode the run of a two-counter machine.
The CRA $\Aa$ will output $0.9$ on valid, halting runs and a value of at least $1.1$ on runs which either do not faithfully encode a run, or do not halt.
Thus there exists a word $w\in\Sigma^*$ such that $\Aa(w) < 1$ if and only if the two-counter machine halts.

The automaton $\Aa = (\Sigma, Q', q_1, I, F , \X, \delta)$ will read encodings of such runs.
Each counter configuration $(c_1,c_2)$ is encoded by the substring $a^{c_1}b^{c_2}$.
Thus we let $a,b\in \Sigma$.
These configurations will be interleaved by a symbol that represents the appropriate transition.
So for each $t\in T$, let $\#_t$ be a symbol in $\Sigma$. 
Finally we have characters  $\vdash,\dashv$ representing the start of the word and the end of the word.
Every accepted word must be of the form $\vdash(\#_t a^*b^*)^*\dashv$. 

An encoding of a run is \emph{faithful} if it correctly encodes a reachable sequence of transitions, correctly encodes the counter at each step, and is \emph{halting} if the final transition reaches the state $q_{halt}$.

We separate these conditions into two categories, those that can be verified with a regular language and those which require verifying the counter valuations. We start with the former category:

\begin{itemize}[leftmargin=*]
\item Correctness of state when reading $a^{n_1}b^{m_1} \#_t a^{n_2}b^{m_2} \#_{t'}$. 
\begin{itemize}
  \item If $t = \incsf_{c_i}(p,q)$ then $\sfstart(t') = q$.

  \item If $t = \zerosf_{c_1}(p,q)$ then $n_1 = 0$ and $\sfstart(t') = q$
  \item If $t = \zerosf_{c_2}(p,q)$ then $m_1 = 0$ and $\sfstart(t') = q$

  \item If $t = \decsf_{c_1}(p,q)$ then $n_1 \ge 1$ and $\sfstart(t') = q$
  \item If $t = \decsf_{c_2}(p,q)$ then $m_1 \ge 1$ and $\sfstart(t') = q$
\end{itemize}
\item Correctness of state when reading $ a^{n_1}b^{m_1}\#_t a^{n_2}b^{m_2} \dashv$. 
\begin{itemize}
  \item If $t = \incsf_{c_i}(p,q)$ then $q = q_{halt}$
  \item If $t = \zerosf_{c_1}(p,q)$ then $n_1 = 0$ and $q = q_{halt}$
  \item If $t = \zerosf_{c_2}(p,q)$ then $m_1 = 0$ and $q = q_{halt}$
  \item If $t = \decsf_{c_1}(p,q)$ then $n_1 \ge 1$ and $q = q_{halt}$
  \item If $t = \decsf_{c_2}(p,q)$ then $m_1 \ge 1$ and $q = q_{halt}$
\end{itemize}
\item Correctness of state when reading $ \vdash a^{n_1}b^{m_1} \#_t$. 
\begin{itemize}
  \item $n_1 = m_1 = 0$
  \item $t = \incsf_{c_1}(q_1,q)$ or $t = \zerosf_{c_i}(q_1,q)$ for some $q\in Q$ and $i\in\{1,2\}$.

\end{itemize}

\end{itemize}
Each of these conditions is recognized by a non-deterministic automaton that is plain to define.
In particular, the automaton can verify that adjacent transitions $\#_t$ and $\#_{t'}$ are legal, and the automaton can distinguish between $n_1 = 0$ or $n_1 \ge 1$ (and similarly $m_1$).
Hence, there exists a deterministic finite automaton $\Bb = (Q', \Sigma, q'_0, Q'_F\subseteq Q')$ recognising the complement of this NFA.

The states, and state transitions of $\Bb$ are the states and state transitions of our automaton $\Aa$.
Finally we encode the correctness of the counter using the registers of $\Aa$.

Let there be two distinguished output registers $x,y\in\X$.
The value of $y$ is fixed to $1.25$ throughout the run, that is $I(y) =1.25$ and every transition sets the value of $y\leftarrow y$.
Whenever $\Bb$ detects a violation the output is $y$, that is $F(q,y) = 1$ for $q\in Q_F$ (the accepting states of $\Bb$).
Whenever $\Bb$ does not detect a violation, the output is $0.9\cdot x$, that is $F(q,y)= 0.9\cdot x$ for $y\in Q\setminus Q_F$.
The role of $x$ is to verify the counter correctness.
Initially set $x = 1$, \ie $I(x) = 1$.

We now consider how to verify the counter valuations are correct.  If a counter valuation is correct, then we will allow $x$ to remain at what it was.
If there is a failure, the register value increases to at least $1.25$.
Thus $x = 1$ at the end of the run only if the run faithfully encodes the counters.

Note that every block $a^nb^m$ takes part in two verification procedures; checking it conforms with the block before it, and the block after it. We now consider the required properties when reading the block $X a^{n_1}b^{m_1} \#_t a^{n_2}b^{m_2}Y$ for $X,Y \in\Sigma\setminus\{a,b\}$.

We use 8 additional registers $r_1,r_2,r_3,r_4,r'_1,r'_2,r'_3,r'_4$. Registers $r_1,r_2,r'_1,r'_2$ are used to verify relations of the form $n_1 = n_2 \pm \{0,1,-1\}$ and $r_3,r_4,r'_3,r'_4$ for relations of the form $m_1 = m_2 \pm \{0,1,-1\}$.

 Assume any register is updated to itself where unspecified (for example, $y\leftarrow y$ is assumed for every transition). Observe each transition is copyless (although the update for $x$ is not linear).

The register updates when reading character $a$ and $b$ are:\\ $a:\begin{cases}
r_1 \leftarrow 2r_1 \\
r_2 \leftarrow \frac{1}{2} r_2 \\
r'_1 \leftarrow \frac{1}{2}r'_1 \\
r'_2 \leftarrow 2 r'_2 \\
\end{cases} \quad \quad b: \begin{cases}
r_3 \leftarrow 2r_3 \\
r_4 \leftarrow \frac{1}{2} r_4 \\
r'_3 \leftarrow \frac{1}{2}r'_3 \\
r'_4 \leftarrow 2 r'_4 \\
\end{cases}$

We formally specify the requirements and the register updates:

\begin{itemize}[leftmargin=*]
  \item If $t = \incsf_{c_1}(p,q)$ then we require that $n_2 = n_1 + 1$ and $m_1 = m_2$. The register updates for character $\#_{t}$ is:\\
$\begin{cases}
r_1,r_2,r_3,r_4 \leftarrow 1\\
r'_1 \leftarrow 2r_1 \\
r'_2 \leftarrow \frac{1}{2}r_2 \\
r'_3 \leftarrow r_3 \\
r'_4 \leftarrow r_4 \\
x \leftarrow x \cdot \frac{1}{2}( r'_1+ r'_2)\cdot \frac{1}{2}( r'_3+ r'_4)
\end{cases}$
  \item If $t = \incsf_{c_2}(p,q)$ then we require that $m_2 = m_1 + 1$ and $n_1 = n_2$
  The register updates for character $\#_{t}$ is:\\
$ \begin{cases}
r_1,r_2,r_3,r_4 \leftarrow 1\\
r'_1 \leftarrow r_1 \\
r'_2 \leftarrow r_2 \\
r'_3 \leftarrow 2r_3 \\
r'_4 \leftarrow \frac{1}{2}r_4 \\
x \leftarrow x \cdot \frac{1}{2}( r'_1+ r'_2)\cdot \frac{1}{2}( r'_3+ r'_4)
\end{cases}$

  \item If $t = \decsf_{c_1}(p,q)$ and $n_1 \ge 1$ then we require that $n_2 = n_1 - 1$ and $m_1 = m_2$
  The register updates for character $\#_{t}$ is:\\
$ \begin{cases}
r_1,r_2,r_3,r_4 \leftarrow 1\\
r'_1 \leftarrow \frac{1}{2}r_1 \\
r'_2 \leftarrow 2r_2 \\
r'_3 \leftarrow r_3 \\
r'_4 \leftarrow r_4 \\
x \leftarrow x \cdot \frac{1}{2}( r'_1+ r'_2)\cdot \frac{1}{2}( r'_3+ r'_4)
\end{cases}$

  \item If $t = \decsf_{c_2}(p,q)$ and $m_1 \ge 1$ then we require that $m_2 = m_1 - 1$ and $n_1 = n_2$. The register updates for character $\#_{t}$ is:\\
$ \begin{cases}
r_1,r_2,r_3,r_4 \leftarrow 1\\
r'_1 \leftarrow r_1 \\
r'_2 \leftarrow r_2 \\
r'_3 \leftarrow \frac{1}{2}r_3 \\
r'_4 \leftarrow 2r_4 \\
x \leftarrow x \cdot \frac{1}{2}( r'_1+ r'_2)\cdot \frac{1}{2}( r'_3+ r'_4)
\end{cases}$
\item If $t = \zerosf_{c_1}(p,q)$ or $t = \zerosf_{c_2}(p,q)$  then we require $n_1 = n_2$  and $m_1= m_2$. Additionally, we require $n_1=0$ or $m_1 =0$ respectively, these restrictions can additionally be added to the regular automaton $\Bb$.
The register updates for character $\#_{t}$ is:\\
$ \begin{cases}
r_1,r_2,r_3,r_4 \leftarrow 1\\
r'_1 \leftarrow r_1 \\
r'_2 \leftarrow r_2 \\
r'_3 \leftarrow r_3 \\
r'_4 \leftarrow r_4 \\
x \leftarrow x \cdot \frac{1}{2}( r'_1+ r'_2)\cdot \frac{1}{2}( r'_3+ r'_4)
\end{cases}$

\end{itemize}

Let us first argue that this correctly verifies $n_1 = n_2$ by focusing our attention on the registers $r_1,r_2,r'_2,r'_2$. Verifications on $m_1,m_2$ are identical. 

\begin{itemize}[leftmargin=*]
\item When reading $X$ registers $r_1,r_2$ are set to $1$.
\item Until $\#_t$ is read, reading $a$ updates $r_1 = 2r_1$ and $r_2 = \frac{1}{2}r_2$.
\item Upon reading $\#_{t}$, the values $r_1,r_2$ are moved to $r'_1$ and $r'_2$. The registers $r_1,r_2$ are reset to $1$ and are no longer used for verifying $\#_t$ (rather are used when verifying $Y$).
\item After reading $\#_t$, we update $r'_1 = \frac{1}{2}r'_1$ and $r'_2 = 2r'_2$ on reading $a$. Whenever $b$ is read we do not change the counters associated with verifying $n_1=n_2$. 
\end{itemize}
Then if $n_1 = n_2$ we have $r'_1 = r'_2 = 1$ and so $\frac{1}{2}(r_1+r_2) = 1 $.
On the other hand, if $n_1 \ne n_2$ then we have $r'_1 = 2^{n_1-n_2}$ and $r'_2 = 1/r'_1$, observe that $\frac{1}{2}(2^{n_1-n_2}+ 2^{n_2-n_1}) \ge 1.25$ whenever $n_1 \ne n_2$.

Finally, when reading $Y$ (indicating the verification of $t$ is complete) we set $x = x \cdot \frac{1}{2}( r'_1+ r'_2)\cdot \frac{1}{2}( r'_3+ r'_4)$. If the verification is successful, this is equivalent to leaving $x= x$, and if there is an error (\eg $n_1\ne n_2$), then $x \ge 1.25$. Thus finally, if there is a violation the weight of $\Aa(w) \ge0.9\cdot  1.25 $.

Similarly, we can verify $n_2 = n_1 +1$ by simulating reading an additional character in the first block on the $\#_t$ transition. That is updating $r'_1 = 2 r_1$ and $r'_2 = \frac{1}{2}r_2$ on $\#_t$. Finally verifying  $n_2 = n_1 -1$ is the same, but simulating an additional character from the second block whilst reading $\#_t$ by updating $r'_1 = \frac{1}{2}r_1$ and $r'_2 = 2r_2$.

Symmetrically, we  verify $m_2 = m_1 +\{-1,0,1\}$ on registers $r_3,r_4,r'_3,r'_4$ which behave exactly as $r_1,r_2,r'_1,r'_2$ but update when reading $b$ and stay constant when reading $a$.
\end{proof}

\begin{remark}
It is also undecidable to decide the $\le$-threshold problem. If the machine halts, then there exists a word of weight $0.9$. If the machine terminates then all words have weight at least $1.1$. Therefore there exists a $w$ s.t.  $A(w)\le 1 \iff A(w) < 1 \iff M \text{ halts}$.

On the other hand, it is not clear how to extend the proof to the $\ge$-threshold or $>$-threshold problems, and thus the boundedness problem. 
\end{remark}

\end{document}